\documentclass[11pt]{article}

\usepackage[english]{babel}
\usepackage[utf8x]{inputenc}
\usepackage{amsmath}
\usepackage{amssymb}
\usepackage[linesnumbered,ruled,vlined]{algorithm2e}

\usepackage[dvipsnames]{xcolor}
\definecolor{ao}{rgb}{0.0, 0.5, 0.0}

\usepackage{siunitx}

\SetKwInput{KwInput}{Input}                
\SetKwInput{KwOutput}{Output}              

\SetCommentSty{mycommfont}

\usepackage{todonotes}
\usepackage[flushmargin]{footmisc}
\usepackage{etoolbox} 
\makeatletter
\patchcmd\maketitle{\@makefntext}{\@@@ddt}{}{}
\patchcmd\maketitle{\rlap}{\mbox}{}{}
\makeatother

\usepackage{siunitx}
\usepackage{amsthm}
\usepackage{enumitem}
\usepackage{lmodern}
\usepackage{bbm}
\usepackage{etoolbox}
\usepackage{hyperref}

\usepackage{mathtools}

\mathtoolsset{showonlyrefs}

\usepackage{thmtools}
\usepackage{thm-restate}

\usepackage{titlesec}

\titleformat{\paragraph}
{\normalfont\normalsize\bfseries}{\theparagraph}{1em}{}
\titlespacing*{\paragraph}
{0pt}{3.25ex plus 1ex minus .2ex}{1.5ex plus .2ex}

\renewcommand{\theparagraph}{\thesubsubsection.(\roman{paragraph})}

\usepackage{empheq}
\usepackage{graphicx}

\usepackage{tikz}

\usepackage[left=1in,right=1in,top=1in,bottom=1in,bindingoffset=0cm]{geometry}

\def\D{\mathbb{D}}

\DeclareMathOperator{\Diag}{Diag}
\DeclareMathOperator{\bin}{bin}
\DeclareMathOperator{\good}{good}
\DeclareMathOperator{\poly}{poly}

\def\Q{\mathsf{Q}}
\def\ZZ{\mathsf{Z}}
\def\ZZb{\mathsf{Z}^{\bin}}
\def\HH{\mathsf{H}}
\def\WW{\mathsf{W}}

\def\HHb{\mathsf{H}^{\bin}}

\def\WWb{\mathsf{W}^{\bin}}

\def\mathbi#1{{\textbf #1}}

\DeclareMathOperator*{\E}{\mathbb{E}}

\let\P\relax
\DeclareMathOperator*{\P}{\mathbb{P}}

\def\N{\mathcal{N}}
\def\Ra{\mathcal{R}}

\def\D{\mathcal{D}}

\def\la{\lambda_{\alpha}}

\def\a{\alpha}

\def\N{\mathcal{N}}
\def\T{\mathcal{T}}
\def\O{\Omega}

\def\S{\mathcal{S}}

\def\g{g_{\a}}
\def\bz{\mathbi{0}}

\def\bit{\{0,1\}}
\def\l{\ell}

\usepackage{upref}

\usepackage{soul}

\newtheorem{thm}{Theorem}[section]
\newtheorem{claim}[thm]{Claim}
\newtheorem{prop}[thm]{Proposition}
\newtheorem{lem}[thm]{Lemma}
\newtheorem{fact}[thm]{Fact}
\newtheorem{cor}[thm]{Corollary}

\newtheorem{defin}[thm]{Definition}
\newtheorem{remark}[thm]{Remark}

  \newenvironment{corbis}[1]
  {\renewcommand{\thecor}{\ref{#1}$'$}%
   \addtocounter{cor}{-1}%
   \begin{cor}}
  {\end{cor}}

\tikzset{
  block/.style    = {draw, thick, rectangle, minimum width = 2em},
sblock/.style      = {draw, thick, rectangle, minimum height = 6em,
minimum width = 6em}, 
}

\DeclareMathOperator{\Var}{Var}

\newcommand{\wB}{\widetilde{B}}

\newcommand{\cG}{\mathcal{G}}

\DeclareMathAlphabet{\mathbfsl}{OT1}{cmr}{bx}{sl}
\newcommand{\bU}{\mathbfsl{U}} 
\newcommand{\bZ}{\mathbfsl{Z}} 
\newcommand{\bu}{\mathbfsl{u}} 
\newcommand{\bX}{\mathbfsl{X}} 
\newcommand{\bs}{\mathbfsl{s}} 
\newcommand{\bY}{\mathbfsl{Y}} 
\newcommand{\by}{\mathbfsl{y}} 

\newcommand{\bzz}{\mathbfsl{z}} 
\newcommand{\bx}{\mathbfsl{x}} 
\newcommand{\bv}{\mathbfsl{v}} 
\newcommand{\bV}{\mathbfsl{V}} 
\newcommand{\bw}{\mathbfsl{w}} 
\newcommand{\wV}{\widetilde{\bV}}
\newcommand{\wv}{\widetilde{\bv}}
\newcommand{\wu}{\widetilde{\bu}}
\newcommand{\wg}{\widetilde{g}}
\newcommand{\Wj}{W^{(j)}}
\newcommand{\zj}{z^{(0)}_j}
\newcommand{\zzj}{z^{(1)}_j}

\def\Y{\mathcal{Y}}

\def\wt{\widetilde}
\def\wW{\widetilde{W}}

\newcommand\blfootnote[1]{%
  \begingroup
  \renewcommand\thefootnote{}\footnote{#1}%
  \addtocounter{footnote}{-1}%
  \endgroup
}

\makeatletter
\def\blfootnote{\gdef\@thefnmark{}\@footnotetext}
\makeatother

\title{\textsl{Ar\i kan meets Shannon:} \\
Polar codes with near-optimal convergence to channel capacity
\blfootnote{An extended abstract of this paper was presented at the 2020 ACM Symposium
on Theory of Computing (STOC)~\cite{Guruswami-Riazanov-Ye-STOC}.\\}}

\author{
 Venkatesan Guruswami\thanks{Computer Science Departmemt, Carnegie Mellon University, Pittsburgh, PA 15213. Email: \texttt{venkatg@cs.cmu.edu}. Research supported in part by NSF grants CCF-1422045, CCF-1563742, and CCF-1814603, and a Google Research Award.} 
\and
Andrii Riazanov\thanks{Computer Science Departmemt, Carnegie Mellon University, Pittsburgh, PA 15213. Email: \texttt{riazanov@cs.cmu.edu}. Research supported in part by NSF grants CCF-1422045, CCF-1563742, and CCF-1814603.} \and
Min Ye\thanks{Data Science and Information Technology Research
Center, Tsinghua-Berkeley Shenzhen Institute, Shenzhen, China. Email: \texttt{yeemmi@gmail.com}. Some of this research was carried out when the author was visiting Carnegie Mellon University.}
}

\date{}

\usepackage{tocloft}

\begin{document}

\maketitle
\thispagestyle{empty}

\parskip=0.5ex

\begin{abstract}
Let $W$ be a binary-input memoryless symmetric (BMS) channel with Shannon capacity $I(W)$ and fix any $\alpha > 0$. We construct, for any sufficiently small $\delta > 0$, binary linear codes of block length $O(1/\delta^{2+\alpha})$ and rate $I(W)-\delta$ that enable reliable communication on $W$ with quasi-linear time encoding and decoding.
Shannon's noisy coding theorem established the \emph{existence} of such codes (without efficient constructions or decoding) with block length $O(1/\delta^2)$. This quadratic dependence on the gap $\delta$ to capacity is known to be best possible. Our result thus yields a constructive version of Shannon's theorem with near-optimal convergence to capacity as a function of the block length. This resolves a central theoretical challenge associated with the attainment of Shannon capacity. Previously such a result was only known for the erasure channel.

\smallskip
Our codes are a variant of Ar\i kan's polar codes based on multiple carefully constructed local kernels, one for each intermediate channel that arises in the decoding. A crucial ingredient in the analysis
is a strong converse of the noisy coding theorem when communicating using random linear codes on arbitrary BMS channels. Our converse theorem shows extreme unpredictability of even a single message bit for random coding at rates slightly above capacity.

\end{abstract}

\begin{center}
\textbf{\small Keywords:}\\
{\small Polar codes, capacity-achieving codes, scaling exponent, finite blocklength}
\end{center}

\newpage
{\footnotesize 
\tableofcontents}
\thispagestyle{empty}

\newpage
\setcounter{page}{1}
\section{Introduction}

We construct binary linear codes that achieve the Shannon capacity of the binary symmetric channel, and indeed any binary-input memoryless symmetric (BMS) channel, with a near-optimal scaling between the code length and the gap to capacity. Further, our codes have efficient (quasi-linear time) encoding and decoding algorithms. Let us now describe the context of our result and its precise statement in more detail.

The binary symmetric channel (BSC) is one of the most fundamental and well-studied noise models in coding theory. 
The BSC with crossover probability $p \in (0,1/2)$ ($\mathrm{BSC}_p$) flips each transmitted bit independently with probability $p$. By Shannon's seminal noisy coding theorem \cite{Shannon48}, we know that the capacity of $\mathrm{BSC}_p$ is $1-h(p)$, where $h(\cdot )$ is the binary entropy function. This means that  reliable communication over $\mathrm{BSC}_p$ is possible at information rates approaching $1-h(p)$, and at rates above $1-h(p)$ this is not possible.
More precisely, for any $\delta > 0$, there \emph{exist} codes of rate $1-h(p)-\delta$ using which one can achieve miscommunication probability at most $2^{-\Omega(\delta^2 n)}$ where $n$ is the block length of the code. In fact, random linear codes under maximum likelihood decoding offer this guarantee with high probability.  Thus Shannon's theorem implies the existence of codes of block length $O(1/\delta^2)$ that can achieve small error probability on $\mathrm{BSC}_p$ at rates within $\delta$ of capacity. Conversely, by several classical results \cite{wolfowitz,strassen,Strassen09trans,Polyanskiy10}, we know that the block length has to be at least $\Omega(1/\delta^2)$ in order to approach capacity within $\delta$.

Shannon's theorem is based on the probabilistic method and does not describe the codes that approach capacity or give efficient algorithms to decode them from errors caused by $\mathrm{BSC}_p$. 
Thus the codes with rates $1-h(p)-\delta$ take at least time exponential in $1/\delta^2$ to construct as well as decode.
This is also true for concatenated coding schemes~\cite{forney} as the inner codes have to be decoded by brute-force, and either have to also be found by a brute-force search or allowed to vary over an exponentially large ensemble (leading to exponentially large block length).

The theoretical challenge of constructing codes of rate $1-h(p)-\delta$ with construction/decoding complexity scaling polynomially in $1/\delta$ in fact remained wide open for a long time. Finally, around 2013, two independent works~\cite{GX15,hassani-finite-scaling-paper-journal} gave an effective finite-length analysis of Ar\i kan's remarkable polar codes construction~\cite{arikan-polar}. (Ar\i kan's original analysis, as well as follow-ups like \cite{arikan-telatar}, proved convergence to capacity as the block length grew to infinity but did not quantify the speed of convergence.) Based on this, a construction of polar codes with block length, construction, and decoding complexity all bounded by a polynomial in $1/\delta$ to capacity was obtained in \cite{GX15,hassani-finite-scaling-paper-journal}. The result also applies to any BMS channel, not just the BSC.  

If the block length of the code scales as $O(1/\delta^\mu)$ as a function of the gap $\delta$ to capacity, we say that $\mu$ is the \emph{scaling exponent}. 
The above results established that the scaling exponent of polar codes is finite.  It is worth pointing out that polar codes are the \emph{only} known efficiently decodable capacity-achieving family proven to have a finite scaling exponent.  The work \cite{GX15} did not give an explicit upper bound on the scaling exponent of polar codes, whereas \cite{hassani-finite-scaling-paper-journal} showed the bound $\mu \le 6$.
Following some improvements in \cite{GB13,MHU16_unified}, the current best known upper bound on $\mu$ for the BSC (and any BMS channel) is $4.714$.

Note that random linear codes have optimal scaling exponent $2$. The above results thus raise the intriguing
challenge of constructing codes with scaling exponent close to $2$, a goal we could not even dream of till the recent successes of polar codes. 

Ar\i kan's original polar coding construction is based on a large tensor power of a simple $2 \times 2$ matrix,
which is called the \emph{kernel} of the construction.
For this construction, it was shown in \cite{hassani-finite-scaling-paper-journal} that the scaling exponent 
$\mu$ for Ar\i kan's original polar code construction is \emph{lower bounded} by $3.579$, even for the simple binary erasure channel. 
 Given this limitation, one approach to improve $\mu$ 
 is to consider polar codes based on $\ell \times \ell$ kernels for larger $\ell$. 
 However, better upper bounds on the scaling exponent of polar codes based on larger kernels have not been established except for the simple case of the binary erasure channel (BEC).\footnote{Polar codes based on $\ell \times \ell$ kernels have much larger block length $\ell^t$ compared to $2^t$ for the $2 \times 2$ case. So to get an improvement in $\mu$, one has to compensate for the increasing block length via better bounds on the local behavior of the kernel.}
 For the BEC, using large kernels, polar codes with scaling exponent $2+\alpha$ for any desired $\alpha > 0$ were given in the very nice paper~\cite{FHMV17} which spurred our work. (We will discuss this and other related works in more detail in Sections~\ref{subsec:prior-work}--\ref{subsec:polar-bec}.)

Our main result in this work is a polynomial time construction of polar codes based on large kernels that approach the optimal scaling exponent of $2$ for every BMS channel. Specifically, for any desired small $\alpha > 0$, by picking sufficiently large kernels (as a function of $\alpha$), the block length $N$ can be made as small as $O_\a(1/\delta^{2+\alpha})$ for codes of rate $I(W)-\delta$ (the notation $O_\a(\cdot)$ hides a constant that depends only on $\a$). The encoding and decoding complexity will be \emph{quasi-linear} in $N$, and thus can also have a near-quadratic growth with $1/\delta$.

\begin{thm}[Main]
\label{thm:intro-main}
Let $W$ be an arbitrary BMS channel with Shannon capacity $I(W)$.
For any desired $\a \in \left(0, \frac{1}{36}\right)$, if we choose a large enough constant $\l\ge \l_0(\a)$  to be a power of $2$, then there is a code $\mathcal{C}$ generated by the polar coding construction using kernels of size $\l\times\l$ such that the following four properties hold when $N$ is the code length: 
\begin{enumerate}
\itemsep=0ex
\vspace{-1ex}
\item 
the code construction has $N^{O_\a(1)}$ complexity;
\item both encoding and decoding have $O_\a(N\log N)$ complexity; 
\item the rate of $\mathcal{C}$ is at least  $I(W)-N^{-1/2+18\a}$; and
\item the block decoding error probability is bounded by  $\exp(-N^{\a})$  
when $\mathcal{C}$ is used for channel coding over $W$.
\end{enumerate}
\end{thm}

The above ``constructivizes" the quantitative finite-length version of Shannon's theorem with a small $\alpha$ slack in the speed of convergence to capacity. The lower bound on $\l$ can be chosen as $\ell_0(\a)=\exp(\Omega(\a^{-1.01}))$. More precisely, it should satisfy $\log\l_0(\a) \geq \frac{11}{\a}$ and $\frac{\log\l_0(\a)}{\log\log \l_0(\a) + 2} \geq \frac3{\a}$. 
Note that a similar lower bound on $\l$ also appears in the aforementioned result for the BEC from~\cite{FHMV17}.
Due to the requirement of extremely large $\ell$, our result is thus primarily theoretical in nature, and meant to illustrate that the polar coding framework is powerful enough to achieve asymptotically optimal rate of convergence to Shannon capacity with efficient algorithms.

We would like to point out that in the conference  version of this work~\cite{Guruswami-Riazanov-Ye-STOC} we only proved inverse polynomial decoding error probability, as opposed to the inverse sub-exponential $\exp(-N^{\a})$ bound which we show here. This improvement uses the subsequent analysis of polarization due to Wang and Duursma in~\cite{Wang-Duursma}, where they extended the results of Theorem~\ref{thm:intro-main} to arbitrary discrete memoryless channels, possibly non-binary and asymmetric, and proved the $\exp\left(-N^{O(\a)}\right)$ bound on the decoding error probability. 
However, this was done at a cost of losing the polynomial-time construction complexity of the code. We are able to non-trivially combine the analysis from~\cite{Wang-Duursma} with our approach of constructing the code to achieve \emph{both} polynomial time construction and sub-exponentially small decoding error probability simultaneously. Getting a polynomial-time constructible version of the results of \cite{Wang-Duursma} for general channels with arbitrary input alphabet remains a challenging and interesting open question.


\section{Overview of our construction and analysis} 

In order to better explain our work and situate it in the rich backdrop of related works on polar codes, we begin with some context and background on the phenomenon of channel polarization that lies at the heart of Ar\i kan's polar coding approach. 

\subsection{Channel transforms, entropy polarization, and polar codes}
\label{subsec:polarization-intro}
\sloppy Consider an arbitrary binary-input memoryless symmetric (BMS)\footnote{We say that a channel $W:\{0,1\}\to\Y$ is a BMS channel if there is a permutation $\pi$ on the output alphabet $\Y$ satisfying i) $\pi^{-1}=\pi$ and ii) $W(y|1)=W(\pi(y)|0)$ for all $y\in\Y$.} channel $W\,:\,\bit\to\Y$, and an ${\l \times \l}$ invertible binary matrix $K$ (referred to as the \emph{kernel}). Suppose that we are transmitting a binary vector $\bU = (U_1, U_2, \dots, U_{\l})$ uniformly chosen from $\bit^{\l}$ in the following way: first, it is transformed into $\bX = \bU K$, which is then transmitted through $\l$ copies of the channel $W$ to get the output $\bY = W^{\l}(\bX) \in \Y^{\l}$.

Now imagine decoding the input bits $U_i$ successively in the order of increasing $i$. This naturally leads to a binary-input channel $W_i : \bit \to \Y^{\l}\times\bit^{i-1}$, for each $i\in [\l]$, which is the channel ``seen" by the bit $U_i$ when all the previous bits $\bU_{<i}$ and all the channel outputs $\bY\in \Y^{\l}$ are known. Formally, the transition probabilities of this channel are
\begin{equation}
\label{Arikan_subchannels}
     W_i(\bY, \bU_{<i}\, |\, U_i) = \dfrac1{2^{\l - 1}} \sum_{\bV \in \bit^{\l - i}} W^{\l}\Big(\bY\, |\, (\bU_{<i}, U_i, \bV)K\Big), 
\end{equation}
where $\bU_{<i} \in \bit^{i-1}$ are the first $(i-1)$ bits of $\bU$, and the sum is over all possible values $\bV \in \bit^{\l-i}$ that the last $(\l-i)$ bits of $\bU$ can take. In this paper we will address the channel $W_i$ as ``Ar{\i}kan's bit-channel of $W$ with respect to $K$."

A \emph{polarization transform} associated with the kernel $K$ is then defined as a transformation that maps $\l$ copies of the channel $W$ to the channels $W_1$, $W_2, \dots, W_{\l}$.
For a BMS channel $W$, we define $H(W)$ as the conditional entropy of the channel input random variable given the channel output random variable when the channel input has uniform distribution.
Since $K$ is invertible, a direct implication of the chain rule for entropy gives \emph{entropy conservation property}, which is \begin{equation}
    \label{eq:entropy_conserv}
    \l \cdot H(W) = H(\bX|\bY)
    =H(\bU|\bY)
    =\sum_{i=1}^{\l}H(U_i|\bU_{<i},\bY)
    =\sum_{i=1}^{\l}H(W_i). 
\end{equation} 

If $K$ is invertible and is not upper-triangular under any column permutation (which we refer to as a \emph{mixing matrix}), then the bit-channels $W_1, W_2, \dots, W_{\l}$ start \emph{polarizing} -- some of them become better than $W$ (have smaller entropy), and some become worse \cite[Lemma 1 and Theorem 2]{KSU10}. The standard approach is then to recursively apply the polarization transform of $K$ to these bit-channels. This naturally leads to an $\l$-ary tree of channels. The $t$'th level of the tree corresponds to the linear transformation  $K^{\otimes t}$, the $t$-fold Kronecker product of $K$. \footnote{Actually, the analysis is more convenient if one applies a bit-reversal permutation of the $U_i$'s, and indeed we do so also in this paper, but this is not important for our current discussion.}

In his landmark paper~\cite{arikan-polar}, Ar\i kan proved that when $K = \left(\begin{smallmatrix}1 & 0\\ 1 & 1\end{smallmatrix}\right)$, at 
the $t$'th level, all but a $o(1)$ fraction of the channels (as $t \to \infty$) are either almost noiseless (have tiny entropy) or completely useless (have entropy very close to $1$). 
To get capacity-achieving codes from polarization, the idea is to use the almost-noiseless channels, which will constitute $\approx I(W)$ fraction by conservation of entropy, to carry the message bits, and ``freeze" the bits in the remaining positions to pre-determined values (eg. all $0$s). Thus the generator matrix of the code will consist of those rows of $K^{\otimes t}$ that correspond to the almost-noiseless positions.  Ar\i kan presented a successive cancellation (SC) decoder and proved that it can be implemented using $O(N \log N)$ operations where $N=\ell^t$ is the code length, thanks to the nice recursive structure of $K^{\otimes t}$.

For the parameters of the code, if one shows that at most $\delta_t$ fraction of the channels at the $t$'th level have entropies in the range $(\zeta_t,1-\zeta_t)$, then one (roughly) gets codes of length $2^t$, rate $I(W)-\delta_t-\zeta_t$, for which the SC decoder achieves decoding error probability $\zeta_t \ell^t$ for noise caused by $W$ (see, for example \cite[Theorem A.3]{Blasiok18}). Thus, one needs $\zeta_t$ sub-exponentially small in $t$ (i.e., at most $\exp(-\omega(t))$) to achieve good decoding error. For Ar\i kan's original $2 \times 2$ kernel, this was shown in \cite{arikan-telatar}. Korada, Sasoglu and Urbanke extended the analysis to arbitrary $\l\times\l$ mixing matrices over the binary field~\cite{KSU10}, and Mori and Tanaka established a similar claim over all finite fields~\cite{mori-tanaka}.
 
The fraction $\delta_t$  of \emph{unpolarized} channels (whose entropies fail to be sub-exponentially close to $0$ or $1$) governs the gap to capacity of polar codes. The above works established that $\lim_{t \to \infty} \delta_t = 0$, and thus polar codes achieve capacity asymptotically as the block length grows to infinity. However, they did not provide any  bounds on the speed at which $\delta_t \to 0$ as a function $t$, much less quantify a scaling exponent. Note that one would need to show $\delta_t \le O(\ell^{-t/\mu})$ to establish a scaling exponent of $\mu$, since the code length is $\ell^t$.

\subsection{Scaling exponents: prior work}
\label{subsec:prior-work}
For Ar\i kan's original kernel $\left(\begin{smallmatrix}1 & 0\\ 1 & 1\end{smallmatrix}\right)$, two independent works~\cite{hassani-finite-scaling-paper-journal,GX15} proved that $\delta_t$ drops to $0$ exponentially fast in $t$. This proved that Ar\i kan's polar codes have finite scaling exponent (i.e., converge to capacity polynomially fast in the block length), the first codes with this important feature. Blasiok {\it et al} generalized this result significantly \cite{Blasiok18}, proving that the entire class of polar codes, based on arbitrary mixing matrices over any prime field as kernels, has finite scaling exponent. 

For concrete upper bounds on the scaling exponent, the work of Hassani, Alishahi, and Urbanke~\cite{hassani-finite-scaling-paper-journal} had proved $\mu \le 6$ for Ar\i kan's original kernel. This was improved to $\mu \le 5.702$ in \cite{GB13}.
Mondelli, Hassani, and Urbanke~\cite{MHU16_unified} showed that $\mu \leq 4.714$ for any BMS channel $W$, and showed a better upper bound $\mu \leq 3.639$ for the case when $W$ is a binary erasure channel (BEC).
A \emph{lower bound} $\mu \ge 3.579$ appears in \cite{hassani-finite-scaling-paper-journal} for the case when successive cancellation decoder is used and analyzed using standard methods, which suggests that polar codes based on Ar\i kan's original $2 \times 2$ kernel fall short of the optimal scaling exponent of $2$.

For larger kernels, effective upper bounds on the scaling exponent are harder to establish as the local evolution of the channels is more complex. In fact, to the best of our knowledge, there is no such explicit bound in the literature, for any\footnote{Here we exclude special cases such as a block diagonal matrix with blocks of size at most $2$ which can be reduced to the $2 \times 2$ case but will only have a worse scaling exponent.} kernel of size bigger than $2$.
The analysis of polar codes is a lot more tractable for the case of erasure channels, where symbols get erased (replaced by a ``?" but never corrupted). Next we describe some results for erasure channels as well as the difficulty in extending these results to channels such as the BSC.

\subsection{Polar codes for erasure channels}
\label{subsec:polar-bec}
For the erasure channel, we have analyses of larger kernels and even codes with scaling exponent approaching $2$. 
Binary $\ell \times \ell$ kernels for powers of two $\ell \le 64$ optimized for the binary erasure channel appear in \cite{Miloslavskaya-Trifonov,Fazeli-Vardy,Yao-Fazeli-Vardy}; a $64 \times 64$ kernel achieving $\mu < 3$ is reported in \cite{Yao-Fazeli-Vardy}.

Pfister and Urbanke proved in~\cite{Pfister-Urbanke} that the optimal scaling exponent $\mu=2$ can be approached if one considers transmission over the $q$-ary erasure channel for large alphabet size $q$. They used polar codes based on $q \times q$ kernels.
Fazeli, Hassani, Mondelli, and Vardy~\cite{FHMV17} then established a similar result for the more challenging and also more interesting case of $q=2$, i.e., for the binary erasure channel, using $\ell \times \ell$ kernels for large $\ell$. They pose proving an analogous result for arbitrary BMS channels as an important challenge. Their conjecture that this can be accomplished provided some of the impetus for our work. Our analysis structure follows a similar blueprint to \cite{FHMV17} though the technical ingredients become significantly more complex for channels other than the BEC, as explained next.

The polarization process for erasure channels has a particularly nice structure. If the initial channel $W$ is the binary erasure channel with erasure probability $z$ (denoted $\mathrm{BEC}(z)$), then the Ar{\i}kan channels $W_i$, $i\in [\l]$, arising from the linear transformation by the kernel are also binary erasure channels (specifically, $\mathrm{BEC}(p_i(z))$ where $p_i(\cdot)$ are some polynomials of degree at most $\l$). 
Crucially, \emph{all} the channels in the recursive tree remain BEC. Therefore it suffices to prove the existence of a good polarizing kernel for the class of binary erasure channels, which is parameterized by a single number, the erasure probability, which also equals the entropy of the channel. As shown in \cite{FHMV17}, a random kernel works with good probability for all BEC universally. 
However, fundamentally the calculations for BEC revolve around the rank of various random subspaces, as decoding under the BEC is a linear-algebraic task. Moving beyond the BEC takes us outside the realm of linear algebra into information-theoretic settings where tight quantitative results are much harder to establish.


\subsection{The road to BSC: Using multiple kernels}
\label{sect:mix}

For the case when the initial channel $W$ is a BSC, a fundamental difficulty (among others) is that the channels in the recursion tree will no longer remain BSC (even after the first step). Further, to the best of our knowledge, the various channels that arise do not share a nice common exploitable structure.
 Therefore, we have to think of the intermediate channels as arbitrary BMS channels, a very large and diverse class of channels. It is not clear (to us) if there exists a single kernel to universally polarize \emph{all} BMS channels at a rapid rate. Even if such a kernel exists, proving so seems out of reach of current techniques. Finally, even for a specific BMS, proving that a random kernel polarizes it fast enough requires some very strong quantitative bounds on the performance and limitations of random linear codes for channel coding. We next describe these issues dealing with which constitutes the core of our contributions.
 
Since we are not able to establish that a single kernel could work for the whole construction universally, our idea is to pick different kernels, which depend on the channels that we face during the procedure. That way, by picking a suitable kernel for each channel in the tree, we can ensure that polarization is fast enough throughout the whole process. 

Though we use different kernels in the code construction, all of them have the same size $\l\times\l$. 
We say that a kernel is \emph{good} if all but a $\widetilde{O}(\l^{-1/2})$ fraction of the bit-channels obtained after polar transform by this kernel have entropy $\l^{-\Omega(\log\l)}$-close to either $0$ or $1$.
Given a BMS channel $W$, the code construction consists of $t$ steps, from Step $0$ to Step $t-1$.
At Step $0$, we find a good kernel $K_1^{(0)}$ for the original channel $W$.
After the polar transform of $W$ using kernel $K_1^{(0)}$, we obtain $\l$ bit-channels $W_1,\dots,W_{\l}$. Then in Step $1$, we find good kernels for each of these $\l$ bit-channels. More precisely, the good kernel for $W_i$ is denoted as $K_i^{(1)}$, and the bit-channels obtained by polar transforms of $W_i$ using kernel $K_i^{(1)}$ are denoted as $\{W_{i,j}:j\in[\l]\}$; see Figure~\ref{fig:12305} for an illustration.
At Step $j$, we will have $\l^j$ bit-channels $\{W_{i_1,\dots,i_j}:i_1,\dots,i_j\in[\l]\}$. For each of them, we find a good kernel
$K_{i_1,\dots,i_j}^{(j)}$. After polar transform of $\{W_{i_1,\dots,i_j}:i_1,\dots,i_j\in[\l]\}$ using these kernels, we will obtain $\l^{j+1}$ bit-channels. Finally, after the last step (Step $t-1$), we will obtain $N=\l^t$ bit-channels $\{W_{i_1,\dots,i_t}:i_1,\dots,i_t\in[\l]\}$.
Using the good kernels we obtained in this code construction procedure, we can build an $N\times N$ matrix (or we can view it as a large kernel) $M^{(t)}$ such that the $N$ bit-channels resulting from the polar transform of the original channel $W$ using this large kernel $M^{(t)}$ are exactly $\{W_{i_1,\dots,i_t}:i_1,\dots,i_t\in[\l]\}$.
We will say a few more words about this in Section~\ref{sect:encdec} and provide all the details in Section~\ref{sect:cons}.

Define now a random process by $\WW_0 = W$ and $\WW_{j+1} = (\WW_{j})_i$ for $i$ uniformly chosen from $[\l]$, where $(\WW_{j})_i$ is the $i^{\text{th}}$ Ar\i kan bit-channel of $\WW_j$ with respect to the appropriate kernel chosen in the construction phase.
 In other words, this is a random process of going down the tree of channels, where a uniformly random child of a current node is taken at each step. Finally, define another random process $\HH_j \coloneqq H\left(\WW_j\right)$. Since every kernel in the construction is chosen to be invertible, $\HH_j$ is a martingale due to the conservation of entropy property~\eqref{eq:entropy_conserv}. It is clear that $\WW_j$ marginally is a uniformly random channel of the $j^{\text{th}}$ level of channel tree, and then $\HH_j$ is the entropy of such a randomly chosen channel.

\begin{figure}
    \centering
    \resizebox{0.48\textwidth}{!}{
\begin{tikzpicture}
\node at (7, 9.4) (w0) {\LARGE $W$};
\node [block] at (7, 7.6) (q0) 
{\LARGE Find $K_1^{(0)}$};
\node at (2.5, 5.8) (w1) {\LARGE $W_1$};
\node at (7, 5.8) (w2) {\LARGE $W_2$};
\node at (11.5, 5.8) (w3) {\LARGE $W_3$};
\node [block] at (2.5, 4) (q1) 
{\LARGE Find $K_1^{(1)}$};
\node [block] at (7, 4) (q2) 
{\LARGE Find $K_2^{(1)}$};
\node [block] at (11.5, 4) (q3) 
{\LARGE Find $K_3^{(1)}$};
\node at (1, 2) (w11) {\LARGE $W_{1,1}$};
\node at (2.5, 2) (w12) {\LARGE $W_{1,2}$};
\node at (4, 2) (w13) {\LARGE $W_{1,3}$};
\node at (5.5, 2) (w21) {\LARGE $W_{2,1}$};
\node at (7, 2) (w22) {\LARGE $W_{2,2}$};
\node at (8.5, 2) (w23) {\LARGE $W_{2,3}$};
\node at (10, 2) (w31) {\LARGE $W_{3,1}$};
\node at (11.5, 2) (w32) {\LARGE $W_{3,2}$};
\node at (13, 2) (w33) {\LARGE $W_{3,3}$};

\node at (1, 1.5)  {\LARGE $\vdots$};
\node at (2.5, 1.5)  {\LARGE $\vdots$};
\node at (4, 1.5)  {\LARGE $\vdots$};
\node at (5.5, 1.5)  {\LARGE $\vdots$};
\node at (7, 1.5)  {\LARGE $\vdots$};
\node at (8.5, 1.5)  {\LARGE $\vdots$};
\node at (10, 1.5)  {\LARGE $\vdots$};
\node at (11.5, 1.5) {\LARGE $\vdots$};
\node at (13, 1.5)  {\LARGE $\vdots$};

\draw[->, thick] (w0)--(q0);
\draw[->, thick] (q0)--(w1);
\draw[->, thick] (q0)--(w2);
\draw[->, thick] (q0)--(w3);
\draw[->, thick] (w1)--(q1);
\draw[->, thick] (w2)--(q2);
\draw[->, thick] (w3)--(q3);
\draw[->, thick] (q1)--(w11);
\draw[->, thick] (q1)--(w12);
\draw[->, thick] (q1)--(w13);
\draw[->, thick] (q2)--(w21);
\draw[->, thick] (q2)--(w22);
\draw[->, thick] (q2)--(w23);
\draw[->, thick] (q3)--(w31);
\draw[->, thick] (q3)--(w32);
\draw[->, thick] (q3)--(w33);
\end{tikzpicture} 
}
\hfill
\resizebox{0.5\textwidth}{!}{
\begin{tikzpicture}
\node at (13.3, 9) (y9) {\LARGE $Y_1$};
\node at (13.3, 8) (y8) {\LARGE $Y_2$};
\node at (13.3, 7) (y7) {\LARGE $Y_3$};
\node at (13.3, 6) (y6) {\LARGE $Y_4$};
\node at (13.3, 5) (y5) {\LARGE $Y_5$};
\node at (13.3, 4) (y4) {\LARGE $Y_6$};
\node at (13.3, 3) (y3) {\LARGE $Y_7$};
\node at (13.3, 2) (y2) {\LARGE $Y_8$};
\node at (13.3, 1) (y1) {\LARGE $Y_9$};

\node [block] at (11.8, 9) (w9) {\LARGE $W$};
\node [block] at (11.8, 8) (w8) {\LARGE $W$};
\node [block] at (11.8, 7) (w7) {\LARGE $W$};
\node [block] at (11.8, 6) (w6) {\LARGE $W$};
\node [block] at (11.8, 5) (w5) {\LARGE $W$};
\node [block] at (11.8, 4) (w4) {\LARGE $W$};
\node [block] at (11.8, 3) (w3) {\LARGE $W$};
\node [block] at (11.8, 2) (w2) {\LARGE $W$};
\node [block] at (11.8, 1) (w1) {\LARGE $W$};

\node at (10.3, 9) (x9) {\LARGE $X_1$};
\node at (10.3, 8) (x8) {\LARGE $X_2$};
\node at (10.3, 7) (x7) {\LARGE $X_3$};
\node at (10.3, 6) (x6) {\LARGE $X_4$};
\node at (10.3, 5) (x5) {\LARGE $X_5$};
\node at (10.3, 4) (x4) {\LARGE $X_6$};
\node at (10.3, 3) (x3) {\LARGE $X_7$};
\node at (10.3, 2) (x2) {\LARGE $X_8$};
\node at (10.3, 1) (x1) {\LARGE $X_9$};

\node [sblock] at (7.8, 8) (K3) {\LARGE $K_1^{(0)}$};
\node [sblock] at (7.8, 5) (K2) {\LARGE $K_1^{(0)}$};
\node [sblock] at (7.8, 2) (K1) {\LARGE $K_1^{(0)}$};

\node [sblock] at (2.5, 8) (KK3) {\LARGE $K_1^{(1)}$};
\node [sblock] at (2.5, 5) (KK2) {\LARGE $K_2^{(1)}$};
\node [sblock] at (2.5, 2) (KK1) {\LARGE $K_3^{(1)}$};

\node at (0, 9) (uu9) {\LARGE $U_1$};
\node at (0, 8) (uu8) {\LARGE $U_2$};
\node at (0, 7) (uu7) {\LARGE $U_3$};
\node at (0, 6) (uu6) {\LARGE $U_4$};
\node at (0, 5) (uu5) {\LARGE $U_5$};
\node at (0, 4) (uu4) {\LARGE $U_6$};
\node at (0, 3) (uu3) {\LARGE $U_7$};
\node at (0, 2) (uu2) {\LARGE $U_8$};
\node at (0, 1) (uu1) {\LARGE $U_9$};

\draw[->, thick] (uu9)--(1.2,9);
\draw[->, thick] (uu8)--(1.2,8);
\draw[->, thick] (uu7)--(1.2,7);
\draw[->, thick] (uu6)--(1.2,6);
\draw[->, thick] (uu5)--(1.2,5);
\draw[->, thick] (uu4)--(1.2,4);
\draw[->, thick] (uu3)--(1.2,3);
\draw[->, thick] (uu2)--(1.2,2);
\draw[->, thick] (uu1)--(1.2,1);

\draw[red, line width=2pt] (3.8,9)--(4.5, 9);
\draw[red, line width=2pt] (3.8,8)--(4.5, 8);
\draw[red, line width=2pt] (3.8,7)--(4.5, 7);
\draw[blue, line width=2pt] (3.8,6)--(4.5, 6);
\draw[blue, line width=2pt] (3.8,5)--(4.5, 5);
\draw[blue, line width=2pt] (3.8,4)--(4.5, 4);
\draw[ao, line width=2pt] (3.8,3)--(4.5, 3);
\draw[ao, line width=2pt] (3.8,2)--(4.5, 2);
\draw[ao, line width=2pt] (3.8,1)--(4.5, 1);

\draw[red, line width=2pt] (4.5, 9)--(5.8, 9);
\draw[red, line width=2pt] (4.5, 8)--(5.8, 6);
\draw[red, line width=2pt] (4.5, 7)--(5.8, 3);
\draw[blue, line width=2pt] (4.5, 6)--(5.8, 8);
\draw[blue, line width=2pt] (4.5, 5)--(5.8, 5);
\draw[blue, line width=2pt] (4.5, 4)--(5.8, 2);
\draw[ao, line width=2pt] (4.5, 3)--(5.8, 7);
\draw[ao, line width=2pt] (4.5, 2)--(5.8, 4);
\draw[ao, line width=2pt] (4.5, 1)--(5.8, 1);

\draw[->, red, line width=2pt] (5.8, 9)--(6.5,9);
\draw[->, blue, line width=2pt] (5.8, 8)--(6.5,8);
\draw[->, ao, line width=2pt] (5.8, 7)--(6.5,7);
\draw[->, red, line width=2pt] (5.8, 6)--(6.5,6);
\draw[->, blue, line width=2pt] (5.8, 5)--(6.5,5);
\draw[->, ao, line width=2pt] (5.8, 4)--(6.5,4);
\draw[->, red, line width=2pt] (5.8, 3)--(6.5,3);
\draw[->, blue, line width=2pt] (5.8, 2)--(6.5,2);
\draw[->, ao, line width=2pt] (5.8, 1)--(6.5,1);

\draw[->, thick] (9.1,9)--(x9);
\draw[->, thick] (9.1,8)--(x8);
\draw[->, thick] (9.1,7)--(x7);
\draw[->, thick] (9.1,6)--(x6);
\draw[->, thick] (9.1,5)--(x5);
\draw[->, thick] (9.1,4)--(x4);
\draw[->, thick] (9.1,3)--(x3);
\draw[->, thick] (9.1,2)--(x2);
\draw[->, thick] (9.1,1)--(x1);

\draw[->, thick] (x9)--(w9);
\draw[->, thick] (x8)--(w8);
\draw[->, thick] (x7)--(w7);
\draw[->, thick] (x6)--(w6);
\draw[->, thick] (x5)--(w5);
\draw[->, thick] (x4)--(w4);
\draw[->, thick] (x3)--(w3);
\draw[->, thick] (x2)--(w2);
\draw[->, thick] (x1)--(w1);

\draw[->, thick] (w9)--(y9);
\draw[->, thick] (w8)--(y8);
\draw[->, thick] (w7)--(y7);
\draw[->, thick] (w6)--(y6);
\draw[->, thick] (w5)--(y5);
\draw[->, thick] (w4)--(y4);
\draw[->, thick] (w3)--(y3);
\draw[->, thick] (w2)--(y2);
\draw[->, thick] (w1)--(y1);
\end{tikzpicture}
}
    \caption{The left figure illustrates the code construction, and the right figure shows the encoding procedure for the special case of $\l=3$ and $t=2$. All the kernels in this figure have size $3\times 3$. One can show that the bit-channel $W_{i,j}$ in the left figure is exactly the channel mapping from $U_{3(i-1)+j}$ to $(\bU_{[1:3(i-1)+j-1]},\bY_{[1:9]})$ in the right figure.}
    \label{fig:12305}
\end{figure}

\subsection{Analysis of polarization via recursive potential function}
The principle behind polarization is that for large enough $t$, almost all of the channels on the $t$-th level of the tree from Figure~\ref{fig:12305} will be close to either the useless or noiseless channel, i.e., their entropy is very close to $1$ or $0$, correspondingly. To estimate how fast such polarization actually happens, one needs to bound the fraction of channels on the $t$-th level that are not yet sufficiently polarized, i.e., $\P\Big[\HH_t \in (\zeta, 1 - \zeta)\Big]$ for some tiny threshold $\zeta$, and show that this fraction vanishes rapidly with increasing $t$.

Specifically, we have the following result (stated explicitly in ~\cite[Theorem A.3]{Blasiok18}) already alluded to in Section~\ref{subsec:polarization-intro}: if for all $t$
\begin{equation}
\label{eq:strong_polarization}
     \P[\HH_t \in (p\l^{-t}, 1 - p\l^{-t})] \leq D\cdot\beta^t,
\end{equation}
then this corresponds to a polar code with block length $N = \l^{t}$, rate $(D\cdot\beta^t + p\l^{-t})$-close to the capacity of the channel, and decoding error probability at most $p$ under the successive cancellation decoder
\footnote{For this part the reader should think of $p$ as being inverse polynomial (of fixed degree) in $N$. We will discuss improving the decoding error probability in Section~\ref{sec:overview:exponential-decoding}.}.

To track the fractions of polarized and non-polarized channels at each level of the construction, we use a potential function 
\begin{equation}
\label{eq:potential}
    \g(h) = (h(1-h))^{\a},
\end{equation}
where $\a > 0$ is some small fixed parameter. This $\a$ corresponds to the gap to the scaling exponent in Theorem~\ref{thm:intro-main}, and in this paper we always consider $\a < \frac1{12}$ (and smaller in some cases). Such a potential function was also used for example in~\cite{MHU16_unified} and~\cite{FHMV17}. We are going to track the expected value $\E[\g(\HH_t)]$ as $t$ increases, since Markov's inequality implies
\begin{equation}
\label{eq:strong_Markov}
     \P[ \HH_t \in (p\l^{-t}, 1 - p\l^{-t})] = \P [\g(\HH_t) \geq \g(p\l^{-t})] \leq \dfrac{\E[\g(\HH_t)]}{\g(p\l^{-t})}\leq 2\left(\l^t/p\right)^{\a}\cdot \E[\g(\HH_t)]. 
\end{equation}

To upper bound $\E[\g(\HH_t)]$ we choose kernels in our construction so that the average of the potential function of all the children of any channel in the tree decreases significantly with respect to the potential function of this channel.
Precisely, we want for any channel $W'$ in the tree
\begin{equation}
\label{mult_decrease}
\E_{i\sim[\l]}\Big[\g\left(H(W'_i)\right)\Big] \leq \lambda_{\a}\cdot\g\left(H(W')\right),
\end{equation}
where $W'_i$ are the children of $W'$ in the construction tree for $i\in[\l]$, and the constant $\lambda_{\a}$ only depends on $\a$ and $\l$, but is universal for all the channels in the tree (and thus for all the kernels chosen during the construction). If one can guarantee that~\eqref{mult_decrease} holds throughout the construction process, then for the martingale process $\HH_t$ obtain
\begin{equation}
\begin{aligned}
\label{eq:Ega_evolution} \noeqref{eq:Ega_evolution}
\E\Big[\g\left(\HH_t\right)\Big] &= \E\left[\underset{j\sim [\l]}{\E}\Big[\g\left(H\left((\WW_{t-1})_j\right)\right) \Big]\, \bigg\rvert\, \WW_{t-1} \right]\\ &= 
\E\left[\dfrac1{\l}\dfrac{\sum_{j=1}^{\l}\g\left(H\left((\WW_{t-1})_j\right)\right)}{\g\left(H(\WW_{t-1})\right)}\cdot\g\left(H(\WW_{t-1})\right) \, \bigg\rvert\, \WW_{t-1}  \right] \\
&\leq \lambda_{\a}\cdot\E\Big[\g\left(\HH_{t-1}\right)\Big],
\end{aligned}
\end{equation}
and then inductively 
\begin{equation}
\label{eq:Ega_exponential}
    \E\Big[\g\left(\HH_t\right)\Big] \leq \la\cdot \E\Big[\g\left(\HH_{t-1}\right)\Big] \leq  \la^2\cdot \E\Big[\g\left(\HH_{t-2}\right)\Big] \leq \cdots \leq \la^t \HH_0 \leq  \la^t.
\end{equation}

Then~\eqref{eq:strong_Markov} and~\eqref{eq:strong_polarization} imply existence of code with rate $O\left((N/p)^{\a}\cdot\la^t\right)$-close to capacity of the channel. Since our main task is to achieve a gap which is close to $N^{-1/2} = \l^{-t/2}$, we need to argue that it is possible to choose kernels at each step in the construction so that~\eqref{mult_decrease} always holds for some $\a \to 0$ and  $\la \approx \l^{-1/2}$.

\subsection{Sharp transition in polarization}
\label{sect:sharp}
The main technical contribution of this paper consists in showing that if $\l$ is large enough, it is possible to choose kernels in the construction process for which $\la$ is close to $\l^{-1/2}$. Specifically, we prove that if $\l$ is a power of $2$ such that $\log \l = \O\left(\frac1{\a^{1.01}}\right)$, then it is possible to achieve 
\begin{equation}
    \label{eq:lambda_bound}
    \la \leq \l^{-1/2 + 5\a}.
\end{equation}
To obtain such a behavior, while choosing the kernel for the current channel $W'$ during the recursive process we differentiate between two cases:

\medskip \noindent {\bf Case 1: $W'$ is already very noisy or almost noiseless.} Such regime is called \emph{suction at the ends} (following \cite{Blasiok18}), and it is known that polarization happens (much) faster for this case. So in this case we take a standard Ar\i kan's polarization kernel $K = \left(\begin{smallmatrix}1 & 0\\ 1 & 1\end{smallmatrix}\right)^{\otimes \log\l}$ and prove~\eqref{mult_decrease} with a geometric decrease factor $\la \leq \l^{-1/2}$.
    
\medskip \noindent {\bf Case 2: $W'$ is neither very noisy nor almost noiseless.} We refer to this case as \emph{variance in the middle} regime (following~\cite{Blasiok18} again). For such a channel we adopt the framework from~\cite{FHMV17} and show a \emph{sharp transition in polarization} for a random kernel $K$ and $W'$. Specifically, we prove that with high probability over $K \sim \bit^{\l\times\l}$ (for $\l$ large enough) it holds
    \begin{equation}
    \label{eq:sharp_trans}
        \begin{aligned}
            &H(W'_i(K)) \leq \l^{-\O(\log \l)} &&\text{\normalfont{for}}  && i \geq \l\cdot H(W') + \O(\l^{1/2}\log^3\l),\\
    &H(W'_i(K)) \geq 1 - \l^{-\O(\log \l)} &&\text{\normalfont{for}}  && i \leq \l\cdot H(W') - \O(\l^{1/2}\log^3\l).
        \end{aligned}
    \end{equation}
    It then follows that only about $\widetilde{O}(\l^{-1/2})$ fraction of bit-channels are not polarized for some kernel $K$, which then easily translates into the bound~\eqref{eq:lambda_bound} on $\la$ that we desire. Note that we  can always ensure that we take an invertible kernel $K$, since a random binary matrix is  invertible with at least some constant probability. 
    
    Proving such a sharp transition constitutes the bulk of the technical work in this paper. In Section~\ref{BMS_section} we show that inequalities in~\eqref{eq:sharp_trans} correspond to decoding a single bit of a message which is transmitted through $W'$ using a random linear code. The first set of inequalities in~\eqref{eq:sharp_trans} then correspond to saying that one can decode this single bit with low error probability with high probability over the randomness of the code, if the rate of the code is at least approximately $\l^{-1/2}$ smaller than the capacity of the channel (where $\l$ is the blocklength of the code). This follows from the well-studied fact that random linear codes achieve Shannon's capacity over any BMS (\cite{Gallager65}, \cite{Barg_Forney_correspond}). 
    
    The second set of inequalities, on the other hand, corresponds to saying that for random linear codes with rate exceeding capacity by at least $\approx \l^{-1/2}$, even predicting a single bit of the message with tiny advantage over a uniform guess is not possible. While it follows from the converse Shannon's coding theorem that decoding the \emph{entire} message is not possible (even with small success probability) for \emph{any} code above capacity, it does not follow that one cannot recover \emph{a particular message bit} with accuracy noticeably better than random guessing. In fact, if we only want to decode a specific message bit and we do not put any constraints on the code, then we can easily construct codes with rate substantially above capacity that still allow us to decode this specific message bit with high probability. All we need to do here is to repeat the message bit sufficiently many times in the codeword, decode each copy based on the corresponding channel output, and then take a majority vote. The overal code rate does not even figure in this argument.
  Therefore, one can only hope that the converse theorem for bit-decoding holds for certain code ensembles, and for the purpose of this paper, we restrict ourselves to random linear code ensemble.
    While the converse for bit-decoding in this case is surely intuitive, establishing it in the strong quantitative form \eqref{eq:sharp_trans} that we need, and also for all BMS channels, turns out to be a challenging task. We describe some of the  ideas behind our strong converse theorem for bit-decoding in Section~\ref{sect:outline}.

\subsection{Encoding and decoding} \label{sect:encdec}
Once we have obtained the kernels in the code construction (see Section~\ref{sect:mix}), the encoding procedure is pretty standard; see \cite{Presman15,Ye15,Gabry17,Benammar17,Wang18} for discussions on polar codes using multiple kernels.
As mentioned in Section~\ref{sect:mix}, we can build an $N\times N$ matrix $M^{(t)}:=D^{(t-1)}Q^{(t-1)}D^{(t-2)}Q^{(t-2)}\dots D^{(1)}Q^{(1)} D^{(0)}$, where the matrices $Q^{(1)}, Q^{(2)},\dots,Q^{(t-1)}$ are some permutation matrices, and $D^{(0)}, D^{(1)},\dots,D^{(t-1)}$ are block diagonal matrices. In particular, all the blocks on the diagonal of $D^{(j)}$ are the kernels that we obtained in Step $j$ of the code construction. We illustrate the special case of $\l=3$ and $t=2$ in Figure~\ref{fig:12305}.
We take a random vector $\bU_{[1:N]}$ consisting of $N=\l^t$ i.i.d. Bernoulli-$1/2$ random variables and we transmit the random vector $\bX_{[1:N]}$ through $N$ independent copies of $W$. Denote the output vector as $\bY_{[1:N]}$. Then for every $i\in[N]$, the bit-channel mapping from $U_i$ to $(\bU_{[1:i-1]},\bY_{[1:N]})$ is exactly $W_{i_1,\dots,i_t}$, where $(i_1,\dots,i_t)$ is $\l$-ary expansion of $i$.

We have shown that almost all of the $N$ bit-channels $\{W_{i_1,\dots,i_t}:i_1,\dots,i_t\in[\l]\}$ become either noiseless or completely noisy. In the code construction, we can track $H(W_{i_1,\dots,i_t})$ for every $(i_1,\dots,i_t)\in[\l]^t$, and this allows us to identify which $U_i$'s are noiseless under successive decoding. Then in the encoding procedure, we only put information in these noiseless $U_i$'s and set all the other $U_i$'s to be some ``frozen" value, e.g., $0$. This is equivalent to saying that the generator matrix of our code is the submatrix of $M^{(t)}$ consisting of rows corresponding to indices $i$ of the noiseless $U_i$'s. In Section~\ref{sect:cons}, we will show that similarly to the classical polar codes, both the encoding and decoding of our code also have $O(N\log N)$ complexity.

As a final remark, we mention that  we need to quantize every bit-channel we obtain in every step of the code construction. More precisely, we merge the output symbols whose log-likelihood ratios are close to each other, so that after the quantization, the output alphabet size of every bit-channel is always polynomial in $N$. This allows us to construct the code in polynomial time. Without quantization, the output alphabet size would eventually be exponential in $N$.  We will provide more details about this aspect, and how it affects the code construction and the analysis of decoding, in Section~\ref{sect:local} and Section~\ref{sect:cons}.

\subsection{Inverse sub-exponential decoding error probability}
\label{sec:overview:exponential-decoding}
Up to this moment, the described construction only achieved inverse polynomial decoding error probability. One reason for this restriction comes from the quantization of the bit-channels that we do, which leads to only having approximations of the actual bit-channels. In particular this means that we only track the parameters (entropy and Bhattacharyya parameter) of the bit-channels approximately, with an additive error inverse polynomial in the blocklength. This directly translates to only claiming inverse polynomial decoding error probability. 

It a recent work Wang and Duursma~\cite{Wang-Duursma} show that it is possible to achieve a good scaling exponent ($2 + O(\a)$) and inverse sub-exponential decoding error probability ($\exp(-N^{\a})$) for polar codes simultaneously, using the idea of multiple kernels in the construction. However, the construction phase in~\cite{Wang-Duursma} tracked the exact bit-channels that are obtained in the $\l$-ary tree of channels (without quantization), which means that the construction of such polar codes is no longer doable in polynomial time. This is because (most of) the exact bit-channels cannot even be described in a tractable way, since they have exponential size of output alphabet. 

We combine our approach of using Ar\i kan's kernels for polarized bit-channels (Case $1$ in Section~\ref{sect:sharp}) with a strong analysis of polarization from~\cite{Wang-Duursma} to achieve good scaling exponent, inverse sub-exponential decoding error probability, and polynomial time complexity of construction, all at the same time. The main idea behind our argument is that even though we cannot track the exact bit-channels in the construction, we know how basic Ar\i kan's kernel evolves the parameters of the bit-channels. Then, if we start with a slightly polarized bit-channel, and take a sufficient amount of ``good" branches of Ar\i kan's $2\times 2$  kernels, we end up with a strongly polarized channel. The crucial observation here is that it suffices to only track the approximation of the bit-channel to verify that it is slightly polarized, and no additional computation is needed to check how many ``good" branches were taken in the tree of bit-channels. In such a way, we show that it is possible to prove very strong polarization for bit-channels, which leads to good decoding error probability, while still only tracking the approximations of the bit-channels, which keeps the construction complexity polynomial. All of these arguments, which lead to the main result of this paper, are made precise and proven in Section~\ref{sec:exponential-decoding}.

\vspace{-2ex}

\section{Outline of strong converse for random linear codes}
\label{sect:outline}
In this section we describe the plan of the proof for the strong converse theorem for bit-decoding random linear codes under the binary symmetric channel. In particular, we need to show the sharp transition as in~\eqref{eq:sharp_trans}, when the channel is BSC. The proof for the general BMS channel case follows the same blueprint by using the fact that a BMS channel can be represented as a convex combination of BSC subchannels, but executing it involves overcoming several additional technical hurdles. Let us fist formulate the precise theorem for the binary symmetric channel. 
\begin{thm}\label{thm:over:BSC_converse} Let $W$ be the $\text{BSC}_p$ channel, and let $\l$, $k$ be integers that satisfy ${\l \geq k \geq \l(1-H(W)) + \O(\l^{1/2}\log\l)}$ and $\l \geq 8$. Let $G$ be a random binary matrix uniform over $\bit^{k\times \l}$. Suppose a message $\bV\cdot G$ is transmitted through $\l$ copies of the channel $W$, where $\bV$ is uniformly random over $\bit^k$, and let $\bY$ be the output vector, i.e. $\bY = W^{\l}(\bV\cdot G)$. Then, with probability at least $1 - \l^{-\O(\log\l)}$ over the choice of~$G$ it holds  ${H\left(V_1\;\big\lvert\;\bY\right) \geq 1 - \l^{-\O(\log \l)}}$.
\end{thm}
We want to point out two quantitative features of the above theorem. First, it applies at rates $\approx \Omega(\l^{-1/2})$ above capacity. Second, it rules out predicting the bit $V_1$ with advantage $\ell^{-\omega(1)}$ over random guessing. Both these features are important to guarantee the desired bound  $\lambda_\alpha \lesssim \l^{-1/2}$. 

\medskip\noindent {\bf Proof plan.}\quad We prove the lower bound on $H\left(V_1\;\big\lvert\;\bY\right)$ by lower bounding $\E\limits_{g\sim G}\left[H\left(V_1\;\big\lvert\;\bY\right)\right]$ and using Markov's inequality. Thus we write
\begin{equation}
\label{eq:intro:confused_notation}
\E_{g\sim G}\big[H^{(g)}(V_1|\bY)\big] = \sum_g \P(G=g) H^{(g)}(V_1|\bY) = \sum_g \P(G=g)\hspace{-3pt} \left(\sum_{\by \in \Y^{\l}} \P\nolimits^{(g)}(\bY=\by) H^{(g)}(V_1|\bY=\by) \right),
\end{equation}
where the summation of $g$ is over $\bit^{k\times\l}$, and by $\P\nolimits^{(g)}(\cdot)$ and $H^{(g)}(\cdot)$ we denote probability and entropy over the randomness of the message $\bV$ and channel noise \emph{for a fixed matrix $g$}.

\smallskip \noindent {\bf 1:  Restrict to zero-input.}\quad The first step is to use the linearity of the (random linear) code and the additive structure of BSC to prove that we can change $\P\nolimits^{(g)}(\bY=\by)$ to ${\P\nolimits^{(g)}(\bY=\by|\bV = \mathbi{0})}$ in the above summation, where $\mathbi{0}$ is the all-zero vector. This observation is crucial for our arguments, since it allows to only consider the outputs which are ``typical" for the all-zero codeword, and there is no dependence on $g$ in this case. Formally, in Appendix~\ref{app:BMS_lemmas} we prove (Lemma~\ref{typical_entropy_BSC}):
\[ \E\limits_{g\sim G}\big[H^{(g)}(V_1|\bY)\big] = \sum\limits_{\by \in \Y^{\l}}  \P(\bY=\by|\bV = \mathbi{0}) \cdot\E\limits_{g\sim G}\left[H^{(g)}(V_1|\bY=\by)\right].\]
\smallskip \noindent {\bf 2: Define a typical set of outputs.}\quad We define a typical output set for the zero-input as $\mathcal{F} \coloneqq \Big\{\by \in \Y^{\l}\,:\, |wt(\by) - \l p| \leq 2\sqrt{\l}\log\l \Big\}$. It is clear that  if zero-vector is transmitted through the channel, the output will be a vector from $\mathcal{F}$ with high probability. It means that we do not lose too much in terms of accuracy if we restrict our attention only to this typical set, so the following inequality suffices as a good lower bound on the expectation.
\begin{equation}
\label{eq:intro:expect_of_entropy_typical}
    \E\limits_{g\sim G}\big[H^{(g)}(V_1|\bY)\big]  \geq \sum\limits_{\by \in \mathcal{F}}  \P(\bY=\by|\bV = \mathbi{0}) \cdot\E\limits_{g\sim G}\left[H^{(g)}(V_1|\bY=\by)\right].
\end{equation}

\noindent {\bf 3: Fix a typical output $\by\in \mathcal{F}$.} 
For a fixed choice of $\by\in\mathcal{F}$, we express $H^{(g)}(V_1|\bY=\by)  = h(\P\nolimits^{(g)}(V_1=0 | \bY=\by)) = h \left(\frac{\P\nolimits^{(g)}(V_1=0,\bY=\by)}{\P\nolimits^{(g)}(\bY=\by)} \right)$.
It suffices to show that the ratio of these probabilities is very close to $1/2$ w.h.p. 
To this end, we will show that both denominator and numerator are highly concentrated around their respective means for $g\sim G$, and that the means have a ratio nearly $1/2$ .
Focusing on the denominator (the argument for the numerator is very similar), we have:
\begin{equation}
\label{eq:intro:prob_of_y}
   2^k\cdot\P\nolimits^{(g)}(\bY = \by) =  \P(\bY=\by\,|\,\bV = \mathbi{0}) +
\sum_{d=0}^{\l} B_g(d, \by) p^d (1-p)^{{\l}-d},
\end{equation}
where $B_g(d,\by)$ is equal to the number of nonzero codewords in the code spanned by the rows of $g$ at Hamming distance $d$ from $\by$. We proceed with proving concentration on the summation above by splitting it into two parts.

\setlength{\leftskip}{0.3cm}

\smallskip \noindent {\bf 3a: Negligible part.}\quad It is very unlikely that an input codeword $\bx$ such that $|\text{dist}(\bx,\by)-\ell p| \geq 6\sqrt{\l}\log\l$ was transmitted, if $\by$ was received as the output. It is then possible to show that the expectation (over $g\sim G$) of $\sum\limits_{d\,:\,|d - \l p| \geq 6\sqrt{\l}\log\l} \hspace{-18pt}B_g(d, \by) p^d (1-p)^{{\l}-d}$ is negligible with respect to the expectation of the whole summation. Markov's inequality implies then that this sum is negligible with high probability over $g\sim G$.

\smallskip \noindent {\bf 3b: Substantial part.} On the other hand, for any $d$ such that $|d - \l p| \leq 6\sqrt{\l}\log\l$, the expectation of $B_g(d, \by)$ is going to be extremely large for the above-capacity regime. We can apply Chebyshev's inequality to prove concentration on every single weight coefficient $B_g(d, \by)$ with $d$ in such a range. A union bound then implies that they are all concentrated around their means simultaneously.

\setlength{\leftskip}{0pt}

\smallskip
\noindent This proves that the summation over $d$ is concentrated around its mean in~\eqref{eq:intro:prob_of_y}. Finally, since 
$|wt(\by) - \l p| \leq 2\sqrt{\l}\log\l$ for $\by\in \mathcal{F}$ and we leave enough room above the capacity of the channel, w.h.p. over choice of $g$ we have $B_g(wt(\by),\by)\gg 1$, and consequently
$\P(\bY=\by\,|\,\bV = \mathbi{0})=p^{wt(\by)}(1-p)^{\l-wt(\by)}$ is negligible compared to the second sum term 
in~\eqref{eq:intro:prob_of_y}.

\medskip \noindent {\bf 4: Concentration of entropy.}\quad Proving in the same way concentration on $\P\nolimits^{(g)}(V_1=0,\bY=\by)$, we derive that $\frac{\P\nolimits^{(g)}(V_1=0,\bY=\by)}{\P\nolimits^{(g)}(\bY=\by)}$ is close to $\frac12$ with high probability for any typical $\by\in\mathcal{F}$, and thus $\E_{g\sim G}[H^{(g)}(V_1|\bY=\by)]$ is close to $1$ with high probability for such $\by$. Recalling that the probability to receive $\by \in \mathcal{F}$ is overwhelming for zero-vector input, out of~\eqref{eq:intro:expect_of_entropy_typical} obtain the desired lower bound on 
 $\E\limits_{g\sim G}\big[H^{(g)}(V_1|\bY)\big]$.

\vspace{0.1cm}
The full proof for the BSC case is presented in Section~\ref{sec:BSC_converse}. In order to generalize the proof to general BMS channels we need to track and prove concentration bounds for many more parameters (in the BSC case, we had a single parameter $d$ that was crucial). More specifically, in the BSC case we have to deal with a single binomial distribution when trying to estimate the expectation of $B_g(d,\by)$. For general BMS channels, however, we have to cope with a multinomial distribution and an ensemble of binomially distributed variables that depend on the particular realization of that multinomial distribution. Moreover, we emphasize that Theorem~\ref{thm:over:BSC_converse} and its analogue for BMS must hold in the \emph{non-asymptotic regime}, namely for all code lengths above some absolute constant which does not depend on the channel. (In contrast, in typical coding theorems in information theory one fixes the channel and lets the block length grow to infinity.) 
We show how to overcome all these technical challenges for the general BMS case in Section~\ref{sec:bit-decoding}.


\subsection*{Organization of rest of the paper}
The rest of the paper, which contains all the formal theorem statements and full proofs, is organized as follows. In Section~\ref{sect:KC}, we describe how to find a good polarizing kernel for any BMS, and reduce its analysis to a strong coding theorem and its converse for bit-decoding of random linear codes.
The case when the BMS has entropy already reasonably close to either $0$ or $1$ is handled in Section~\ref{sec:suctions}.
Also, the analysis of the complexity of the kernel finding algorithm is deferred to Section~\ref{sect:cons}. 

Turning to the converse coding theorem for random codes, as a warmup this is first proven for the case of the binary symmetric channel in Section~\ref{sec:BSC_converse}. We then present the proof for general BMS channels in Section~\ref{sec:bit-decoding}. Finally, Section~\ref{sect:cons} has the complete details of our code construction based on the multiple kernels found at various levels, and a sketch of the encoding and decoding algorithms, which when all combined yield Theorem~\ref{thm:main1}, which is almost our main result, but with decoding error probability proven to be only inverse polynomial in the blocklength. 

Lastly, in Section~\ref{sec:exponential-decoding} we show how to combine the tight analysis of the polarization from~\cite{Wang-Duursma} and our construction of codes from Section~\ref{sect:cons} to obtain our final result, also stated in the introductory section as Theorem~\ref{thm:intro-main}, with inverse sub-exponential $\exp(-N^{\a})$ decoding error probability.

\parskip=0.5ex
\section{Useful entropic facts}

\subsection{Binary entropy function}
All the logarithms in this paper are to the base $2$. The binary entropy function is defined as $h(x) = x\log\frac{1}{x} + (1-x)\log\frac1{1-x}$ for $x\in [0,1]$, where $0\log 0$ is taken to be $0$. We will use a simple fact that $h(x) \leq 2x\log\frac1x$ for $x\in [0,1/2)$ several times in the proofs.
The following proposition follows from the facts that $h(x)$ is concave, increasing for $x\in [0, 1/2)$, and symmetric around $1/2$, i.e. $h(x) = h(1-x)$ for $x\in [0,1]$.
\begin{prop}
\label{prop:entropy_differ}
For any $x, y \in [0,1]$, $|h(x) - h(y)| \leq h(|x-y|)$.
\begin{proof}
The inequality is trivial when $x$ or $y$ is equal to $0$ or $h(x) = h(y)$. Without loss of generality, assume $x > y$. Further, consider first the case $h(x) > h(y)$. We have two cases:
\begin{enumerate}
    \item[(a)] $0 < y \leq (x-y) < x$. By the mean value theorem, we can write $(h(x) - h(x-y)) = h'(\xi_1)y$ for some $\xi_1 \in (x-y, x)$, and $h(y) = h(y) - h(0) = h'(\xi_2)y$ for some $\xi_2 \in (0, y)$. Then $\xi_2 \leq \xi_1$, and since $h$ is concave, it follows that $h'(\xi_2) \geq h'(\xi_1)$, thus $h(x) - h(x-y) \leq h(y)$. Rearranging, obtain the desired inequality.
    \item[(b)] $0 < (x-y) \leq y < x$. By the same argument, one has $(h(x) - h(y)) = h'(\xi_1)(x-y)$ for $\xi_1 \in (y, x)$ and $h(x-y) = h(x-y) - h(0) = h'(\xi_2)(x-y)$ for some $\xi_2 \in (0, x-y)$, and so $\xi_2 \leq \xi_1$, therefore $h'(\xi_2) \geq h'(\xi_1)$ by concavity. Thus $h(x) - h(y) \leq h(x-y)$.
\end{enumerate}
Next, if $h(x) < h(y)$, define $x' = 1-y$ and $y' = 1-x$. It follows that $x' > y'$ and $h(x') = h(y) > h(x) = h(y')$, so the inequality in the proposition holds for $x'$ and $y'$ by the cases (a)-(b) above. But clearly $|h(x) - h(y)| = |h(x') - h(y')| \leq h(|x' - y'|) = h(|x - y|)$ by symmetry of $h$ around~$\frac12$.
\end{proof}
\end{prop}

\begin{prop}
\label{prop:entropy_half}
$h(x) \leq 2x\log\frac1x$ for $x\in [0,1/2]$.
\begin{proof}
Consider the function $f(x) = 2x\log\frac1x - h(x) = x\log\frac{1}{x} - (1-x)\log\frac1{1-x}$ on $[0, 1/2]$. We have $f''(x) = \dfrac{2x - 1}{x(1-x)\ln 2} < 0$ on $(0, 1/2)$, so $f$ is strictly concave on this interval, and further $f(0) = f(1/2) = 0$. Therefore, $f(x)$ is positive on $(0, 1/2)$.
\end{proof}
\end{prop}

\subsection{Channel degradation}
\begin{defin}
\label{def:degrad}
Let $W\,:\, \bit \to \Y$ and $\wW\,:\, \bit \to \wt{\Y}$ be two BMS channels. We say that $\wW$ is \emph{degraded} with respect to $W$, or, correspondingly, $W$ is \emph{upgraded} with respect to $\wW$, denoted as $\wW \preceq W$, if there exists a discrete memoryless channel $W_1\, :\, \Y \to \wt{\Y}$ such that
\begin{equation*}
\label{eq:degrad}
     \wW(\wt{y}\, |\, x) = \sum_{y\in\Y}W(y\,|\,x)W_1(\wt{y}\,|\,y) \qquad\quad \forall\ x\in\bit,\ \wt{y}\in\wt{\Y}.
\end{equation*}
\end{defin}
Note that this is equivalent to saying that $\wW(x)$ and $W_1(W(x))$ are identically distributed for any $x\in\bit$. In other words, one can simulate the usage of $\wW$ by first using the channel $W$ and then applying some other channel $W_1$ to the output of $W$ to get a final output.

We will use some useful facts from~\cite[Lemma 3]{Tal_Vardy} and~\cite[Lemma IV.1]{Ye15}. Note that Proposition~\ref{prop:degrad_subchannel} below was first proved in \cite[Lemma 21]{KU-lossy} for the special case of Ar{\i}kan kernel and then generalized in \cite[Lemma IV.1]{Ye15} to general kernels.
\begin{prop}
\label{prop:degrad_entropy}
Let $W$ and $\wW$ be two BMS channels, such that $\wW\preceq W$. Then $H(\wW) \geq H(W)$.
\end{prop}

\begin{prop}
\label{prop:degrad_subchannel}
Let $W$ and $\wW$ be BMS channels, such that $\wW\preceq W$, and $K \in \bit^{\l\times\l}$ be any invertible matrix. Denote by $W_i$, $\wW_i$ the Ar\i kan's bit-channels of $W$ and $\wW$ with respect to the kernel $K$ for any $i \in [\l]$. Then  for any $i \in [\l]$, we have $\wW_i \preceq W_i$, and consequently $H(\wW_i) \ge H(W_i)$.
\end{prop}

\section{Give me a channel, I'll give you a kernel}
\label{sect:KC}
In this section we show that for any given binary-input memoryless symmetric (BMS) channel  $W$ we can find a kernel $K$ of size $\l\times\l$, such that the Ar{\i}kan bit-channels of $W$ with respect to this kernel will be highly polarized. By this we mean that the multiplicative decrease $\la$ defined in \eqref{mult_decrease} will be sufficiently close to $\l^{-1/2}$. The algorithm (Algorithm~\ref{algo:kernel_search}) to find such a kernel is as follows: if the channel is already almost noiseless or too noisy (entropy is very close to $0$ or $1$), we take this kernel to be a tensor power of original Ar{\i}kan's kernel for polar codes, $A_2 =\left( \begin{smallmatrix}1 & 0\\ 1 & 1\end{smallmatrix}\right)$. Otherwise, the algorithm will just try out all the possible invertible kernels in $\bit^{\l\times\l}$, until a ``good" kernel is found, which means that conditions~\eqref{eq:algo_stop_condition} should be satisfied. Before proving that Algorithm~\ref{algo:kernel_search} achieves our goals of bringing $\la$ close to $\l^{-1/2}$, we discuss several details about it.

\subsection{Local kernel construction}  \label{sect:local}

\begin{algorithm}[h]
 \caption{Kernel search}
\label{algo:kernel_search}
\DontPrintSemicolon
\SetAlgoLined
  \KwInput{BMS channel $\wW$ with output size $\leq \mathsf{Q}$, error parameter $\Delta$, and number $\l$} 
  \KwOutput{invertible kernel $K\in \bit^{\l\times\l}$}

 \eIf{$H(\wW)  < \l^{-4}$ \emph{\textbf{or}}  $H(\wW)  > 1 - \l^{-4} + \Delta$}
  { \Return $K = A_2^{\otimes \log\l}$}
  {\For{$K \in \bit^{\l\times\l}$, \emph{\textbf{if}} $K$ is invertible}{
  Compute Ar{\i}kan's bit-channels $\wW_i(K)$ of $\wW$ with respect to the kernel $K$, as in~\eqref{Arikan_subchannels}
  
  \If{ 
  \vspace{-2pt}
  \begin{equation}
      \begin{aligned}
  \label{eq:algo_stop_condition}
      &H(\wW_i(K)) \leq \l^{-(\log\l)/4} &&\text{\normalfont{for}}  && i \geq \l\cdot H(\wW) + \l^{1/2}\log^3\l\\
    &H(\wW_i(K)) \geq 1 - \l^{-(\log\l)/20} &&\text{\normalfont{for}}  && i \leq \l\cdot H(\wW) - 14\l^{1/2}\log^3\l
  \end{aligned}
  \end{equation}}{\Return $K$}
  }}
\end{algorithm}

As briefly discussed at the end of Section~\ref{sect:encdec}, we are unable to efficiently track all the bit-channels in the $\l$-ary recursive tree \emph{exactly}. This is because the size of the output alphabet of the channels increase \emph{exponentially} after each step deeper into the tree (this simply follows from the definition of bit-channels~\eqref{Arikan_subchannels}). Thus computing all the channels (and their entropies) cannot be done in poly$(N)$ time. To overcome this issue we follow the approach of~\cite{Tal_Vardy}, with subsequent simplification in~\cite{GX15}, of approximating the channels in the tree by degrading (see Definition~\ref{eq:degrad}) them. Degradation is achieved via the procedure of merging the output symbols, which (a) decreases the output alphabet size, and (b) does not change the entropy of the channel too much. This implies (with all the details worked out in Section~\ref{sect:cons}) that we can substitute all the channels in the tree of depth $t$ by their \emph{degraded approximations}, such that all the channels have output alphabet size at most $\Q$ (a parameter depending on $N = \l^{t}$ to be chosen), and that if $\wW$ is a degraded approximation of the channel $W$ in the tree, than $H(W) \leq H(\wW) \leq H(W) + \Delta$ for some $\Delta$ depending on $\Q$. Moreover, in Theorem~\ref{thm:kernel_seacrh_correct} which we formulate and prove shortly, we show that when we apply Algorithm~\ref{algo:kernel_search} to a degraded approximation $\wW$ of $W$ with small enough $\Delta$, then, even though  conditions~\eqref{eq:algo_stop_condition} only dictate a sharp transition for $\wW$, the same kernel will induce a sharp transition in polarization for $W$.

The second issue which such degraded approximation resolves is the running time of Algorithm~\ref{algo:kernel_search}. Notice that we are only going to apply it for channels with output size bounded by $\Q$, and recall that we think of $\l$ as of a constant (though very large). First of all, trying out all the possible kernels will then also take a constant number of iterations. Finally, within each iteration, just calculating all the Ar{\i}kan's bit-channels and their entropies in a straightforward way will take poly$(\Q^{\l})$ time, which is just poly$(\Q)$ when we treat $\l$ as a constant. Therefore by choosing $\Q$ to be polynomial in $N$, the algorithm indeed works in poly$(N)$ time. 

We now leave the full details concerning the complexity of the algorithm to be handled in Section~\ref{sect:cons}, and proceed with showing that Algorithm~\ref{algo:kernel_search} always returns a kernel which makes $\la$ from~\eqref{mult_decrease} close to $\l^{-1/2}$.

\begin{thm}
\label{thm:kernel_seacrh_correct} 
Let $\a \in \left(0, \frac1{12}\right)$ be a small fixed constant. Let $\l$ be an even power of $2$ such that $\log\l \geq \frac{11}{\a}$ and $\frac{\log\l}{\log\log\l + 2} \geq \frac{3}{\a}$. Let $W : \bit\to\Y$ and $\wW:\bit \to\wt{\Y}$ be two BMS channels, such that $\wW \preceq W$, ${H(\wW) - \Delta \leq H(W)\leq H(\wW)}$ for some $0 \leq \Delta \leq \l^{-\log \l}$, and $|\wt{\Y}| \leq \Q$. Then Algorithm~\ref{algo:kernel_search} on inputs $\wW$, $\Delta$, and $\l$ returns a kernel $K \in \bit^{\l\times\l}$ that satisfies
\begin{equation}
\label{eq:mult_decrease_thm}
    \dfrac{1}{\l\cdot\g(H(W))} \sum_{i=1}^{\l}\g\left(H(W_i)\right) \leq \l^{-\frac12 + 5\a},
\end{equation} 
where $W_1, W_2, \dots, W_{\l}$ are the Ar{\i}kan's bit-channels of $W$ with respect to the kernel $K$, and $g_\a(\cdot)$ is the potential function $g_\a(h) = \left(h(1-h)\right)^{\a}$ for any $h\in[0,1]$, as defined in~\eqref{eq:potential}.
\end{thm}
\begin{proof} As we discussed above, we consider two cases:

\vspace{0.2cm}

\noindent \textbf{Suction at the ends.} If $H(\wW) \notin (\l^{-4}, 1-\l^{-4} + \Delta)$,  Algorithm~\ref{algo:kernel_search} returns a standard Ar{\i}kan's kernel $K = A_2^{\otimes \log\l}$ on input $\wW$ and $\Delta$. For this case $H(W) \notin (\l^{-4}, 1 - \l^{-4})$, and fairly standard arguments imply that the polarization under such a kernel is much faster when the entropy is close to $0$ or $1$. For completeness, we present the full proofs for this case in a deferred Section~\ref{sec:suctions}. Specifically, Lemma~\ref{lem:suction_evolution} immediately implies the result of the theorem for this regime, as we pick $\log\l \geq \frac1{\a}$. 

\vspace{0.2cm}

\noindent \textbf{Variance in the middle.} \sloppy Otherwise, if $H(\wW) \in (\l^{-4}, 1-\l^{-4} + \Delta)$, it holds $H(W) \in {(\l^{-4} - \Delta, 1 - \l^{-4} + \Delta)}$, thus $H(W) \in \left(\l^{-4}/2, 1- \l^{-4}/2\right)$ since $0 \leq \Delta \leq \l^{-\log \l}$ and $\log\ell$ is large by the conditions of the theorem.

We first need to argue that the algorithm will at least return some kernel. This argument is one of the main technical contributions of this work, and we formulate it as Theorem~\ref{thm:Arikans_entropies_polarize} in Section~\ref{BMS_section}. The theorem essentially claims that for any $\wW$ an overwhelming fraction of possible kernels $K \in \bit^{\l\times\l}$ satisfies the conditions in~\eqref{eq:algo_stop_condition} for $\wW$ and $K$ (note that we do not use any conditions on the size of $\wt{\Y}$ or the entropy $H(\wW)$ at all at this point). Clearly then, there is a decent fraction of \emph{invertible} kernels from $\bit^{\l\times\l}$ which also satisfy these conditions. Therefore, the algorithm will indeed terminate and return such a good kernel. Moreover, since the theorem claims that a random kernel from $\bit^{\l\times\l}$ will satisfy~\eqref{eq:algo_stop_condition} with high probability, and it is also known that it will be invertible with at least some constant probability. It means that instead of iterating through all possible kernels in step $4$ of Algorithm~\ref{algo:kernel_search}, we could take a random kernel and check it, and then the number of iterations needed to find a good kernel would be very small with high probability. However, to keep everything deterministic, we stick to the current approach. 

Suppose now the algorithm returned an invertible kernel $K \in \bit^{\l\times\l}$, which means that relations~\eqref{eq:algo_stop_condition} hold for $\wW$ and Ar{\i}kan's bit-channels $\wW_1, \wW_2,\dots, \wW_{\l}$ (we omit dependence on $K$ from now on). Denote also $W_i = W_i(K)$ as an Ar{\i}kan's bit-channels of $W$ with respect to $K$. First, since degradation is preserved after considering Ar{\i}kan's bit-channels according to Proposition~\ref{prop:degrad_subchannel}, $\wW_i \preceq W_i$, thus $H(W_i) \leq H(\wW_i)$ for all $i\in[\l]$. Now, similarly to the proof of Proposition~\ref{prop:approx_accumulation}, since $K$ is invertible, conservation of entropy implies $\sum_{i =1}^{\l}\left(H(\wW_i) - H(W_i)\right) = \l \left(H(\wW) -  H(W)\right) \leq \l \cdot\Delta$, therefore derive $H(W_i) \leq H(\wW_i) \leq H(W_i) + \l\cdot\Delta$ for any $i \in [\l]$. Then deduce from \eqref{eq:algo_stop_condition}
  \begin{equation}
      \begin{aligned}
  \label{eq:algo_analysis_subchannels}
      &H(W_i) \leq H(\wW_i) \leq \l^{-(\log\l)/4} &&\text{\normalfont{for}}  && i \geq \l\cdot H(\wW) + \l^{1/2}\log^3\l\\
    &H(W_i) \geq H(\wW_i) - \l\cdot\Delta \geq 1 - \l^{-(\log\l)/21} &&\text{\normalfont{for}}  && i \leq \l\cdot H(\wW) - 14\cdot\l^{1/2}\log^3\l,
  \end{aligned}
  \end{equation}
where we used that $\Delta \leq \l^{-\log \l}$ and $\ell$ is large in the condition of the theorem.

 Recall that $H(W) \in \left(\l^{-4}/2, 1- \l^{-4}/2\right)$ for variance in the middle regime, and note that this implies 
 \begin{equation}
 \label{eq:g_a_lower}
 \g(H(W)) \geq \g(\l^{-4}/2) = \left(\frac12\cdot(1 - \l^{-4}/2)\right)^{\a}\cdot\l^{-4\a} \geq \left(\frac14\right)^{\a}\l^{-4\a}  \geq \frac12\l^{-4\a}, 
 \end{equation}
 since $\g$ is increasing on $(0, 1/2)$ and $\a < 1/2$. Using~\eqref{eq:algo_analysis_subchannels} and the trivial bound $\g(x) \leq 1$ for all the indices $i$ close to $\l\cdot H(\wW)$ obtain that the LHS of the desired inequality \eqref{eq:mult_decrease_thm} is at most
\begin{equation}
 \label{eq:mult_decrease_proof}
\begin{aligned}
     &\hspace{-2cm}\dfrac{1}{\l\cdot\g(H(W))} \Bigg(&&\sum_{i=1}^{ \l\cdot H(\wW) - 14\cdot\l^{1/2}\log^3\l}\g\left(1 - \l^{-(\log\l)/21}\right) + 15\l^{1/2}\log^3\l  \\ & &+&\hspace{3pt} \sum_{i=\l\cdot H(\wW) + \l^{1/2}\log^3\l}^{\l}\g\left(\l^{-(\log\l)/4}\right)\Bigg)   \\
     &&& \hspace{-2cm}\overset{(a)}{<} \rlap{$2\ell^{4\a - 1}\left(15\l^{1/2}\log^3\l + \ell\cdot H(\wW) \cdot \ell^{-(\a\log\l)/21} + (\ell - \ell\cdot H(\wW)) \cdot \ell^{-(\a\log\l)/4} \right)$} \\
     &&& \hspace{-2cm}< \rlap{$30\l^{-\frac12 + 4\a}\log^3\l +2\l^{-(\a \log \l)/21 + 4\a}$}  \\
     &&& \hspace{-2cm}\overset{(b)}{\leq} \rlap{$\l^{-\frac12 + 4\a}\left(30\log^3\l + 2\l^{-1/42}\right) < \l^{-\frac12 + 4\a} \cdot 32\log^3\l $},     \\
     &&& \hspace{-2cm}\overset{(c)}{\leq} \rlap{$\l^{-\frac12 + 5\a}$},  
 \end{aligned}
\end{equation}
where $(a)$ follows from~\eqref{eq:g_a_lower} and the fact that $\g(x) = \g(1-x) \leq x^{\a}$ for $x \in (0, 1)$; $(b)$ uses the condition $\log\l \geq \dfrac{11}{\a}$, and $(c)$ uses $\dfrac{\log\l}{\log\log\l + 2} \geq \dfrac{3}{\a}$ from the requirements that we have on~$\l$ in the conditions of this theorem. 
\end{proof}

\begin{remark} \label{rmk:rmk}
In this paper, we are interested in the cases where $\a$ is very close to $0$.
For such $\a$,
we can absorb the two conditions on $\l$ in Theorem~\ref{thm:kernel_seacrh_correct} into one condition $\log\l\geq \Omega(\a^{-1.01})$ for convenience of notation. 
\end{remark}

\subsection{Strong channel coding and converse theorems}
\label{BMS_section}
In this section we will show that Algorithm~\ref{algo:kernel_search}, which is used to prove the multiplicative decrease of almost $\l^{-1/2}$ as in~\eqref{eq:mult_decrease_thm} in the settings of Theorem~\ref{thm:kernel_seacrh_correct}, indeed always returns some kernel for the regime when the entropy of the channel is not close to $0$ or $1$. While the analysis of suction at the ends regime, deferred to Section~\ref{sec:suctions}, follows standard methods in the literature and only relies on the fact that polarization becomes much faster when the channel is noiseless or useless, in this section we will follow the ideas from~\cite{FHMV17} and prove a \emph{sharp transition in the polarization behaviour}, when we use a random and sufficiently large kernel. 

The sharp transition stems from the fact that when the kernel $K$ is large enough, with high probability (over randomness of $K$) all the Ar{\i}kan's bit-channel with respect to $K$, except for approximately $\l^{1/2}$ of them in the middle, are guaranteed to be either very noisy or almost noiseless. 
We formulate the main result of this section in the following theorem, which was used in the proof of Theorem~\ref{thm:kernel_seacrh_correct}:
\begin{thm}
\label{thm:Arikans_entropies_polarize}
Let $W$ be any BMS channel. Let $W_1, W_2, \dots, W_{\l}$ be the Ar{\i}kan's bit-channels defined in \eqref{Arikan_subchannels} with respect to the kernel $K$ chosen uniformly at random from $\bit^{\l\times\l}$, where $\ell$ is a large integer such that $\log\l > 40$.
Then for the following inequalities all hold with probability $(1 - o_{\l}(1))$ over the choice of~$K$:
\begin{enumerate}[label=(\alph*)]
    \item $H(W_i) \leq \l^{-(\log \l)/4}$\hspace{5pt}\qquad\qquad for\quad  $i \geq \l\cdot H(W) + \l^{1/2}\log^3\l$;
    \item $H(W_i) \geq 1 - \l^{-(\log \l)/20}$\hspace{15pt}\quad for\quad  $i \leq \l\cdot H(W) - 14\cdot\l^{1/2}\log^3\l$.
\end{enumerate}
\end{thm}

\medskip
\begin{remark}
One can notice that the above theorem is stated for any BMS channel $W$, independent of the value of $H(W)$. 
\end{remark}

\medskip
The proof of this theorem relies on results concerning bit-decoding for random linear codes that are interesting beyond the connection to polar codes. The following proposition shows how to connect Ar{\i}kan's bit-channels to this context.

\begin{prop}
\label{prop:Arikan-bit}
Let $W$ be a BMS channel, $K\in\bit^{\l\times\l}$ be an invertible matrix, and $i \in [\l]$. Set $k = \l - i + 1$, and let $G$ be a matrix which is formed by the last $k$ rows of $K$. Let $\bU$ be a random vector uniformly distributed over $\bit^{\l}$, and $\bV$ be a random vector uniformly distributed over $\bit^{k}$. Then
\begin{equation}
\label{eq:Arikan-bit_decod}
     H\left(U_i\;\Big\lvert\; W^{\l}(\bU\cdot K), \bU_{<i}\right) = H\left(V_1\;\Big\lvert\;W^{\l}(\bV \cdot G)\right)  .
\end{equation}
\end{prop}
We are implicitly using the concept of coset codes \cite[Section 6.2]{Gallager68} in this proposition, and the proof technique here is quite standard in the polar coding literature. For example, the same proof technique is used to show that the values of the frozen bits do not matter for polar codes  \cite{arikan-polar,KU-lossy}.
The proof of this proposition only uses basic properties of BMS channels and linear codes, and is deferred to Appendix \ref{app:BMS_lemmas}. Notice now that the LHS of~\eqref{eq:Arikan-bit_decod} is exactly the entropy $H(W_i)$ of the $i$-th Ar{\i}kan's bit-channel of $W$ with respect to the kernel $K$, by definition of this bit-channel. On the other hand, one can think of the RHS of~\eqref{eq:Arikan-bit_decod} in the following way: look at $G$ as a generator matrix for a linear code of blocklength $\l$ and dimension $k$, which is transmitted through the channel $W$. Then $H\left(V_1\;\Big\lvert\;W^{\l}(\bV \cdot G)\right)$ in some sense corresponds to how well one can decode the first bit of the message, given the output of the channel. Since in Theorem~\ref{thm:Arikans_entropies_polarize} we are interested in random kernels, the generator matrix $G$ is also random, and thus we are indeed interested in understanding bit-decoding of random linear codes.

\subsubsection{The BEC case}

When $W$ is the binary erasure channel, a statement 
very similar to Theorem~\ref{thm:Arikans_entropies_polarize} was established in \cite{FHMV17}. The situation for the BEC is simpler and we now describe the intuition behind this.

Suppose we map uniformly random bits $\bU \in \{0,1\}^\l$ to $\bX = \bU K$ for a \emph{random} $\ell \times \ell$ binary matrix $K$.
We will observe $\approx (1-z) \ell$ bits of $\bX$ after it passes through $\mathrm{BEC}(z)$; call these bits $\bZ$. For a random $K$, with high probability the first $\approx z \ell$ bits of $\bU$ will be almost independent of these observed bits $\bZ$. When this happens we will have $H(W_i) = 1$ for $i \lesssim z \ell$.  On the other hand, w.h.p. over the choice of $K$, the remaining bits $U_i$ for $i \gtrsim z \ell$ can be uniquely determined as linear combinations of $\bZ$ and $U_i, i \lesssim z \ell$, making the corresponding conditional entropies $H(W_i)=0$. Thus except for a few exceptional indices around $i \approx z \ell$, the entropy $H(W_i)$ will be very close to $0$ or $1$. The formal details and quantitative aspects are non-trivial as the argument has to handle the case when $z$ is itself close to $0$ or $1$, and one has to show the number of exceptional indices to be $\lesssim \sqrt{\ell}$ (which is the optimal bound). But ultimately the proof amounts to understanding the ranks of various random subspaces. When $W$ is a BMS channel, the analysis is no longer linear-algebraic, and becomes more intricate. This is the subject of the rest of this section as well as Sections~\ref{sec:BSC_converse} and~\ref{sec:bit-decoding}.

\subsubsection{Part (a): channel capacity theorem}

Part (a) of Theorem~\ref{thm:Arikans_entropies_polarize} corresponds to transmitting through $W$ random linear codes  with rates \emph{below} the capacity of the channel. For this regime, it turns out that we can use the classical result that random linear codes achieve the capacity of the channel with \emph{low error decoding probability}. Trivially, the bit-decoding error probability is even smaller, making the corresponding conditional entropy also very small. Therefore, the following theorem follows from classical Shannon's theory:
\begin{thm}
\label{thm:positive_Shannon_BMS}
Let $W$ be any BMS channel and $k \leq \l(1-H(W)) - \l^{1/2}\log^3\l$, where $\l \geq 4$. Let $G$ be a random binary  matrix uniform over $\bit^{k\times \l}$. Suppose a codeword $\bV\cdot G$ is transmitted through $\l$ copies of the channel $W$, where $\bV$ is uniformly random over $\bit^k$, and let $\bY$ be the output vector, i.e. $\bY = W^{\l}(\bV\cdot G)$. Then with high probability over the choice of $G$ it holds  ${H\left(V_1\;\big\lvert\;\bY\right) \leq \l^{-(\log \l)/4}}$.
\end{thm}
\begin{proof}
The described communication is just a transmission of a random linear code $C = \{ \bv G,\ \bv \in \bit^k\}$ through $W^{\l}$, where the rate of the code is $R = \frac{k}{\l} \leq I(W) - \l^{-1/2}\log^3\l$, so it is separated from the capacity of the channel. It is a well-studied fact that random (linear) codes achieve capacity for BMS, and moreover a tight error exponent was described by Gallager in~\cite{Gallager65} and analyzed further in~\cite{Barg_Forney_correspond},~\cite{Forney_survey},~\cite{Domb}. Specifically, one can show $\overline{P_e} \leq \text{exp}(-\l E_r(R, W))$, where $\overline{P_e}$ is the probability of decoding error, averaged over the ensemble of all linear codes of rate $R$, and $E_r(R, W)$ is the so-called \emph{random coding exponent}. 
It is proven in~\cite[Theorem~2.3]{Fabregas} that for any BMS channel $W$, one has $E_r(R,W) \geq E_r^{\text{BSC}}(R, I(W))$ where the latter is the error exponent for the BSC channel with the same capacity $I(W)$ as $W$. But the optimal scaling exponent for BSC channels for the regime when the rate is close to the capacity of the channel is given by the so-called sphere-packing exponent $E_r^{\text{BSC}}(R, I) = E_{\text{sp}}(R, I)$ (see, for instance,~{\cite[Section~1.2]{Forney_survey}}, which is easily shown to be almost quadratic in $(I - R)$. Specifically, we use the following
\begin{lem}
$E_{\text{sp}}(R, I) \geq \frac{2\log^4\l}{\l}$ for $R \leq I - \l^{-1/2}\log^3\l$.
\end{lem}
\begin{proof}
For the sphere-packing exponent we use the expression from~\cite[eq (1.4)]{Forney_survey}
\begin{equation}
    \label{eq:sph-exp}
    E_{\text{sp}}(R, I) = D_{\text{KL}}\bigg(\delta_{\text{GV}}(R)\, \big\lvert\big\lvert\, p\bigg),
\end{equation}
where $I = I(W) = 1 - H(W) = 1 - h(p)$ is the capacity of the BSC$_{p}$ channel (with $p < \frac12$), $D_{KL}$ stands for the Kullback–Leibler divergence, and $\delta_{GV}(R)$ is the relative Gilbert-Varshamov distance, which is defined as the solution to $1 - h(\delta) = R$ for $\delta \in \left(0, \frac12\right)$. For convenience, we will just write $\delta$ instead of $\delta_{\text{GV}}(R)$ below.

For $R \leq I - \l^{-1/2}\log^3\l = 1 - h(p) - \l^{-1/2}\log^3\l $, we then have $1 - h(\delta) \leq 1 - h(p) - \l^{-1/2}\log^3\l$, and so $h(\delta) - h(p) \geq  \l^{-1/2}\log^3\l$. Using Proposition~\ref{prop:entropy_differ}, obtain $h(\delta - p) \geq h(\delta) - h(p) \geq  \l^{-1/2}\log^3\l$. Next, since $h(x)$ in increasing on $\left(0, \frac12\right)$ and by Proposition~\ref{prop:entropy_half} $$h(\l^{-1/2}\log^2\l) \leq 2\l^{-1/2}\log^2\l\cdot\log\frac{\l^{1/2}}{\log^2\l} \leq 2\l^{-1/2}\log^2\l\cdot\frac12\log\l = \l^{-1/2}\log^3\l,$$ we conclude that $\delta - p \geq \l^{-1/2}\log^2\l$.

Finally, we use Pinsker's inequality $D_{\text{KL}}\left(P\,||\, Q\right) \geq 2\Delta^2(P, Q)$ between the KL divergence and the total variation distance $\Delta(P, Q) = \frac12||P - Q||_1$ of two distributions $P$ and $Q$ over the same probability space. Abusing the notation and denoting $\Delta(\delta, p)$ as the distance between Bern$(\delta)$ and Bern$(p)$, we have $\Delta(\delta, p) = |\delta - p|$, and so obtain
\begin{align} E_{\text{sp}}(R, I) = D_{\text{KL}}\big(\delta\, || \, p\big) \geq 2\Delta^2(\delta, p) = 2(\delta - p)^2 \geq \dfrac{2\log^4\l}{\l}. &\qedhere
\end{align}
\end{proof}

Therefore using this lemma we have $\overline{P_e} \leq \text{exp}(-\l E_r(R, W)) \leq \text{exp}(-\l E_{\text{sp}}(R, I(W))) \leq \text{exp}(-2\log^4\l)$. Then Markov's inequality implies that if we take a random linear code (i.e. choose a random binary matrix $G$), then with probability at least $1 - \l^{-2}$ the decoding error is going to be at most $\l^2\text{exp}(-2\log^4\l) \leq \text{exp}(-\log^4\l) \leq \l^{-\log\l}$. Consider such a good linear code (matrix $G$), and then $\bV$ can be decoded from $\bY$ with high probability, thus, clearly, $V_1$ can be recovered from $\bY$ with at least the same probability. Then Fano's inequality and Proposition~\ref{prop:entropy_half} gives us:
\begin{equation}
    \begin{aligned}
    H(V_1\,|\,\bY) \leq h_2(\l^{-\log \l}) &\leq 2\l^{-\log \l}\cdot \log\left(\frac1{\l^{-\log \l}}\right)\\
    & = 2\l^{-\log \l} \cdot \log^2\l \quad\leq\quad \l^{-(\log\l)/4},
    \end{aligned}
\end{equation}
where the last inequality follows from $2\log^2\l \leq 2^{{\frac{3\log^2\l}{4}}}$, which holds for $\l \geq 4$.
Thus we indeed obtain that the above holds with high probability (at least $1 - \l^{-2}$, though this is very loose) over the random choice of $G$.
\end{proof}

\subsubsection{Part (b): strong converse for bit-decoding under noisy channel coding}   
On the other hand, part (b) of Theorem~\ref{thm:Arikans_entropies_polarize} concerns bit-decoding of linear codes with rates \emph{above} the capacity of the channel. We prove that with high probability, for a random linear code with rate slightly above capacity of a BMS channel, any single bit of the input message is highly unpredictable based on the outputs of the channel on the transmitted codeword. Formally, we have the following theorem.
\begin{restatable}{thm}{converseShannon}
\label{thm:converse_Shannon_BMS}
\sloppy Let $W$ be any BMS channel, and $\l$ and $k$ be integers that satisfy ${\l \geq k \geq \l(1-H(W)) + 14\l^{1/2}\log^3\l}$, and let $\l$ be large enough so that {$\log\l \geq 20$}. Let $G$ be a random binary matrix uniform over $\bit^{k\times \l}$. Suppose a message $\bV\cdot G$ is transmitted through $\l$ copies of the channel $W$, where $\bV$ is uniformly random over $\bit^k$, and let $\bY$ be the output vector, i.e. $\bY = W^{\l}(\bV\cdot G)$. Then, with probability at least $1 - \l^{-(\log\l)/20}$ over the choice of~$G$ it holds  ${H\left(V_1\;\big\lvert\;\bY\right) \geq 1 - \l^{-(\log\l)/20}}$.
\end{restatable}
\noindent Since the theorem is of independent interest and of a fundamental nature, we devote a separate Section~\ref{sec:bit-decoding} to present a proof for it.

\bigskip

The above statements make the proof of Theorem~\ref{thm:Arikans_entropies_polarize} immediate:
\begin{proof}[Proof of Theorem~\ref{thm:Arikans_entropies_polarize}] 
\sloppy Denote $k = \l - i + 1$, then by Proposition~\ref{prop:Arikan-bit} ${H(W_i) = H\left(V_1\;\Big\lvert\;W^{\l}(\bV \cdot G_k)\right)}$, where $\bV \sim \bit^k$ and $G_k$ is formed by the last $k$ rows of $K$. Note that since $K$ is uniform over $\bit^{\l\times\l}$, this makes $G_k$ uniform over $\bit^{k\times\l}$ for any $k$. Then:
\begin{enumerate}[label=(\alph*)]
    \item For any $i \geq \l\cdot H(W) + \l^{1/2}\log^3\l$, we have $k \leq \l(1-H(W)) - \l^{1/2}\log^3\l$, and therefore Theorem~\ref{thm:positive_Shannon_BMS} applies, giving $H(W_i) \leq \l^{-(\log\l)/4}$ with probability at least $1 - \l^{-2}$ over $K$.
    \item Analogically, if $i \leq \l\cdot H(W) - 14\cdot\l^{1/2}\log^3\l$, then $k \geq \l(1-H(W)) + 14\l^{1/2}\log^3\l$, and Theorem~\ref{thm:converse_Shannon_BMS} gives $H(W_i) \geq 1 - \l^{-(\log\l)/20}$ with probability at least $1 - \l^{-(\log\l)/20}$ over~$K$.
\end{enumerate}
It only remains to take the union bound over all indices $i$ as in (a) and (b) and recall that we took $\ell$ large enough so that $\log\ell > 40$. This implies that all of the bounds on the entropies will hold simultaneously with probability at least $1 - \l\cdot\l^{-2} \geq 1 - \l^{-1}$ over the random kernel $K$.
\end{proof}

\section{Strong converse for BSC\texorpdfstring{$_p$}{\tiny p}}
\label{sec:BSC_converse} We present a proof of Theorem~\ref{thm:converse_Shannon_BMS} in the next two sections. It is divided into three parts: first, we prove it for a special case of $W$ being a BSC channel in this section. The analysis for this case is simpler (but already novel), and it provides the roadmap for the argument for the case of general BMS channel. Next, in Section~\ref{sec:BMS_large_alphabet} we prove Theorem~\ref{thm:converse_Shannon_BMS} for the case when the output alphabet size of $W$ is bounded by $2\sqrt{\l}$, which is the main technical challenge in the paper. The proof will mimic the approach for the BSC case to some extent. Finally, in Section~\ref{sec:BMS_any_alphabet}, we show how the case of general BMS channel can be reduced to the case of the channel with bounded alphabet via ``upgraded binning" to merge output symbols.

Throughout this section consider the channel $W$ to be BSC with the crossover probability $p\leq \frac12$. Denote $H = H(W) = h(p)$, where $h(\cdot)$ is the binary entropy function. For the BSC case we will actually only require $k \geq \l(1-H) + 8\sqrt{\l}\log \l$ in the condition of the Theorem~\ref{thm:converse_Shannon_BMS}. Thus we are in fact proving Theorem~\ref{thm:over:BSC_converse} here.

\begin{proof}[Proof of Theorem~\ref{thm:over:BSC_converse}]
We will follow the plan described in Section~\ref{sect:outline}. As we discussed there, we prove that $H(V_1\,|\,\bY)$ is very close to $1$ with high probability over $G$ by showing that its expectation over $G$ is already very close to $1$ and then using Markov inequality. So we want to prove a lower bound on
\begin{equation}
\label{eq:confused_notation}
\E_{g\sim G}\big[H^{(g)}(V_1|\bY)\big] = \sum_g \P(G=g) H^{(g)}(V_1|\bY),
\end{equation}
where $H^{(g)}(V_1|\bY)$ is the conditional entropy for the fixed matrix $g$. Similarly, in the remainder of this section, $\P\nolimits^{(g)}(\cdot)$  denotes probabilities of certain events \emph{for a fixed matrix $g$}. By $\sum_g$ we denote the summation over all binary matrices from  $\bit^{k\times\l}$. 

\medskip \noindent{\bf Restrict to zero-input.}\quad We rewrite
\begin{align}
    \label{entropy_expectation}
    \E_{g\sim G}\big[H^{(g)}(V_1|\bY)\big] &=\sum_g \P(G=g) \left(\sum_{\by \in \Y^{\l}} \P\nolimits^{(g)}(\bY=\by) H^{(g)}(V_1|\bY=\by) \right) \nonumber \\
      &=\sum_{\by \in \Y^{\l}}  \sum_g \P\nolimits^{(g)}(\bY=\by) \cdot\P(G=g)H^{(g)}(V_1|\bY=\by). \nonumber 
\end{align}
Our first step is to prove that in the above summation we can change $\P\nolimits^{(g)}(\bY=\by)$ to ${\P\nolimits^{(g)}(\bY=\by|\bV = \mathbi{0})}$, where $\mathbi{0}$ is the all-zero vector. This observation is crucial for our arguments, since it allows us to only consider the outputs $\by$ which are ``typical" for the all-zero codeword when approximating $\E\limits_{g\sim G}\big[H^{(g)}(V_1|\bY)\big]$. Precisely, we prove
\begin{lem}
\label{typical_entropy_BSC}
Let $W$ be a BMS channel, $\l$ and $k$ be integers such that $k \leq \l$. Let $G$ be a random binary matrix uniform over $\bit^{k\times \l}$. Suppose a message $\bV\cdot G$ is transmitted through $\l$ copies of $W$, where $\bV$ is uniformly random over $\bit^k$, and let $\bY$ be the output vector $\bY = W^{\l}(\bV\cdot G)$. Then
\begin{equation}
    \label{entropy_expectation_typical}
    \E_{g\sim G}\big[H^{(g)}(V_1|\bY)\big] = \sum_{\by \in \Y^{\l}}  \sum_g \P\nolimits^{(g)}(\bY=\by|\bV = \emph{\mathbi{0}}) \cdot\P(G=g)H^{(g)}(V_1|\bY=\by).
\end{equation}
\end{lem}
Note that the above lemma is formulated for any BMS channel, and we will also use it for the proof of the general case in Section~\ref{sec:bit-decoding}. The proof of this lemma uses the symmetry of linear codes with respect to shifting by a codeword and additive structure of BSC, together with the fact that a BMS channel can be represented as a convex combination of several BSC subchannels. We defer the proof to Appendix~\ref{app:BMS_lemmas}.

Note that $\P\nolimits^{(g)}(\bY=\by|\bV = \bz)$ does not in fact depend on the matrix $g$, since $\bz\cdot g = \bz$, and so randomness here only comes from the usage of the channel $W$. Specifically, $\P\nolimits^{(g)}(\bY=\by|\bV = \bz) = p^{wt(\by)}(1-p)^{{\l}-wt(\by)}$, where we denote by $wt(\by)$ the Hamming weight of $\by$. Then in~\eqref{entropy_expectation_typical} we obtain
\begin{align}
\label{entr_exp_weights_typical}
    \E_{g\sim G}\big[H^{(g)}(V_1|\bY)\big] = \sum_{\by \in \Y^{\l}} p^{wt(\by)}(1-p)^{{\l}-wt(\by)}  \E_{g\sim G}\big[H^{(g)}(V_1|\bY=\by)\big].
\end{align}

\medskip \noindent{\bf Define a typical set.}\quad  The above expression allows us to only consider ``typical" outputs $\by$ for the all-zero input while approximating $\E_{g\sim G}\big[H^{(g)}(V_1|\bY)\big]$. For the BSC case, we consider $\by$ to be typical when $|wt(\by) - \l p| \leq 2\sqrt{\l}\log\l$. Then we can write:
\begin{align}
\label{entr_exp_weights_only_typical}
    \E_{g\sim G}\big[H^{(g)}(V_1|\bY)\big] \geq \sum_{|wt(\by) - \l p| \leq 2\sqrt{\l}\log\l} p^{wt(\by)}(1-p)^{{\l}-wt(\by)} \E_{g\sim G}\big[H^{(g)}(V_1|\bY=\by)\big].
\end{align}

\medskip \noindent{\bf Fix a typical output.}\quad Let us fix any typical $\by \in \Y^{\l}$ such that $|wt(\by) - \l p| \leq 2\sqrt{\l}\log\l$, and show that $\E_{g\sim G}[H^{(g)}(V_1|\bY=\by)]$ is very close to $1$. To do this, we first notice that
\begin{equation}
\label{entropy_ratio}
     H^{(g)}(V_1|\bY=\by)  = h \left(\frac{\P\nolimits^{(g)}(V_1=0,\bY=\by)}{\P\nolimits^{(g)}(\bY=\by)} \right).
\end{equation} 
Denote $\wV = \bV^{[2:k]}$ to be bits $2$ to $k$ of vector $\bV$, and by $\wg = g[2:k]$ the matrix $g$ without its first row. Next we define the shifted weight distributions of the codebooks generated by $g$ and $\wg$:
\begin{align*}
B_g(d, \by) & :=|\{ \bv \in \bit^k \setminus  \mathbi{0}\quad : wt(\bv g+\by)=d \}|, \\
\widetilde{B}_g(d, \by) & :=|\{\wv\in \bit^{k-1} \setminus  \mathbi{0} : wt(\wv  \wg+\by)=d \}|.
\end{align*}
 Therefore,
\begin{align}
    \nonumber \frac{\P\nolimits^{(g)}(V_1=0,\bY=\by)}{\P\nolimits^{(g)}(\bY=\by)}
    &=\frac{\sum_{\wu}\P\nolimits^{(g)}(\bY=\by \big| V_1=0, \wV=\wu)}
{\sum_{\bu}\P\nolimits^{(g)}(\bY=\by \big| \bV=\bu)} \\
 \label{ratio} &=  \frac{p^{wt(\by)}(1-p)^{{\l}-wt(\by)}+ 
\sum_{d=0}^{\l} \widetilde{B}_g(d, \by) p^d (1-p)^{{\l}-d}}
{p^{wt(\by)}(1-p)^{{\l}-wt(\by)}+
\sum_{d=0}^{\l} B_g(d, \by) p^d (1-p)^{{\l}-d}}.
\end{align}
We will prove a concentration of the above expression around $1/2$, which will then imply that $H^{(g)}(V_1|\bY=\by)$ is close to $1$ with high probability by \eqref{entropy_ratio}. To do this, we will prove concentrations around means for both numerator and denominator of the above ratio. Since the following arguments work in exactly the same way, let us only consider the denominator for now.

By definition,
\begin{equation} \label{eq:smds}
\begin{aligned}
B_g(d, \by) & =\sum_{\bv\neq \mathbi{0}}
\mathbbm{1}[wt(\bv g +\by)=d].
\end{aligned}
\end{equation}
The expectation and variance of each summand is
\begin{align*}
&\underset{g\sim G}{\Var}\, \mathbbm{1} \big[wt(\bv g +\by)=d \big] \le
    \E_{g\sim G} \mathbbm{1} \big[wt(\bv g +\by)=d \big]
    =\binom{{\l}}{d} 2^{-{\l}}
    \quad\quad \forall \bv \in \bit^k \setminus  \mathbi{0}.
\end{align*}
Clearly, the summands in \eqref{eq:smds} are pairwise independent. Therefore,
\begin{align}
\label{variance_expectation_weight}
   \underset{g\sim G}{\Var} \big[B_g(d, \by)\big] & \le \E_{g\sim G} \big[B_g(d, \by)\big]
    =(2^k-1) \binom{{\l}}{d} 2^{-{\l}},
\end{align}
and then 
\begin{align*}
    \E_{g\sim G} \left[\sum_{d=0}^{\l} B_g(d, \by) p^d (1-p)^{{\l}-d}\right]
    =  (2^{k}-1)  2^{-{\l}} \left( \sum_{d=0}^{\l} \binom{{\l}}{d} p^d (1-p)^{{\l}-d} \right)
    = (2^{k}-1)  2^{-{\l}}.
\end{align*}

Let us now show that $\sum_{d=0}^{\l} B_g(d, \by) p^d (1-p)^{{\l}-d}$ is tightly concentrated around its mean for $g\sim G$. To do this, we split the range of $d$ into two parts: when $|d - \l p| > 6\sqrt{\l} \log \l $, and when $|d - \l p| \leq 6\sqrt{\l} \log \l$:
\[ \sum_{d=0}^{\l}B_g(d, \by)p^d(1-p)^{\l - d} = \sum_{|d - \l p| > 6\sqrt{\l} \log \l}B_g(d, \by)p^d(1-p)^{\l - d} + \sum_{|d - \l p| \leq 6\sqrt{\l} \log \l}B_g(d, \by)p^d(1-p)^{\l - d}.  \]

In the proof below we will use the following multiplicative form of Chernoff bound applied to a binomial random variable: 

\begin{equation}
\label{new_eq:multiplicative_Chernoff}
\P_{X \sim \text{Binom}(\l, p)}\left[|X - \l p| \geq \delta \l p\right] \leq 2e^{-\l p \delta^2 / 3} \qquad\quad \text{for all } 0 \leq \delta \leq 1.
\end{equation}
 Applying this for $\delta = \frac{6\log\l}{p\l^{1/2}}$, we have
\begin{equation}
\label{eq:chernoff_binomial}
    \P_{X \sim \text{Binom}(\l, p)}\left[|X - \l p| \geq 6\sqrt{\l} \log \l\right] = \sum_{|d - \l p| \geq 6\sqrt{\l} \log \l}\binom{\l}{d}p^d(1-p)^{\l - d} \leq 2e^{\frac{-12\log^2\l}{p}} < 2\l^{-12\log\l}.
\end{equation}

\medskip \noindent{\bf Negligible part.}\quad Denote $Z_g(\by) = \sum\limits_{|d - \l p| > 6\sqrt{\l} \log \l}B_g(d, \by)p^d(1-p)^{\l - d},$ and notice that
\begin{align} \E_{g\sim G}[Z_g(\by)] = (2^k-1)2^{-\l}\sum_{|d - \l p| > 6\sqrt{\l} \log \l}\binom{\l}{d}p^d(1-p)^{\l - d} &\leq (2^k-1)2^{-\l}\cdot 2\l^{-12\log\l}, \label{eq:EZ_bound}
\end{align}
where we used~\eqref{variance_expectation_weight} and~\eqref{eq:chernoff_binomial}.
Then Markov's inequality gives $\P_{g\sim G}\big[ Z_g(\by) \geq \E\limits_{g\sim G}[Z_g(\by)]\l^{2\log\l}\big] \leq \l^{-2\log\l}$, and so
\[ \P\big[Z_g(\by) < 2(2^k-1)2^{-\l}\l^{-10\l\log\l} \big] \geq 1 - \l^{-2\log\l}.\]
Define the set \begin{equation}
\label{eq:G1def}
    \cG_1 := \{g\in\bit^{k\times\l}\, :\, Z_g(\by) < 2(2^k-1)2^{-\l}\l^{-10\l\log\l} \},
\end{equation} and then $\P\limits_{g\sim G}[g \in \cG_1] \geq 1  - \l^{-2\log\l}.$

\vspace{0.3cm} \noindent{\bf Substantial part.}\quad
Now we deal with the part when $|d - \l p| \leq 6\sqrt{\l}\log\l$. For now, let us fix any $d$ in this interval, and use Chebyshev's inequality together with \eqref{variance_expectation_weight}:
\begin{equation}
\begin{aligned}
\label{chebyshev}
\P_{g\sim G}\bigg[\Big\lvert B_g(d,\by) - \E[B_g(d,\by)]\Big\lvert \geq \l^{-2\log\l}\E[B_g(d,\by)]\bigg] &\leq \dfrac{\Var[B_g(d,\by)]}{\l^{-4\log\l}\E^2[B_g(d,\by)]} \\
 &\leq \dfrac{\l^{4\log\l}}{\E\limits_{g\sim G}[B_g(d,\by)]} \leq \l^{4\log\l}\,\dfrac{2^{\l-k+1}}{\binom{\l}{d}}. 
 \end{aligned}
 \end{equation}
We use the following bound on the binomial coefficients
\begin{fact}[\cite{Macwilliams77}, Chapter 10, Lemma 7]
\label{binom_lemma}
For any integer $0\le d\le {\l}$,
\begin{equation}
\label{binom_lemma_eq} \frac{1}{\sqrt{2\l}} 2^{{\l} h(d/{\l})}
\le \binom{{\l}}{d} \le 2^{{\l} h(d/{\l})} 
\end{equation}
\end{fact}

Since we fixed $|d - \l p| \leq 6\sqrt{\l}\log\l$, Propositions~\ref{prop:entropy_differ} and ~\ref{prop:entropy_half} imply 
\begin{equation}
\label{eq:entropy_diff_bound}
      \left\lvert h(p) - h\left(\frac{d}{\l}\right) \right\lvert \leq h(6\l^{-1/2}\log\l) \leq 12\l^{-1/2}\log\l\cdot\log\dfrac{\l^{1/2}}{6\log\l} \leq 6\l^{-1/2}\log^2\l.
\end{equation}
Recalling that we consider the above-capacity regime with $k \geq \l(1 - h(p)) + 8\sqrt{\l}\log^2\l$, we derive from~\eqref{binom_lemma_eq} and~\eqref{eq:entropy_diff_bound}
\begin{equation}
\label{eq:binom_appr_1}
\dfrac{2^{\l-k+1}}{\binom{\l}{d}} \leq \sqrt{2\l}\,\cdot 2^{\l\left[h(p) -  h\left(\frac{d}{\l} \right) - 8\l^{-1/2}\log^2\l\right]}  \leq \sqrt{2\l}\,\cdot 2^{-2\l^{1/2}\log^2\l}.   
\end{equation}
Therefore, we get in \eqref{chebyshev}:
\begin{equation}
     \label{single_concentration}
\P_{g\sim G}\bigg[\Big\lvert B_g(d,\by) - \E[B_g(d,\by)] \Big\lvert \geq \l^{-2\log\l}\E[B_g(d,\by)]\bigg]\leq \sqrt{2\l}\,\cdot  \l^{4\log\l}\,2^{-2\l^{1/2}\log^2\l} \leq \l^{-\sqrt{\l}-1},
\end{equation}
where the last inequality holds since $\l \geq 8$.
Finally, denote 
\begin{equation}
\label{eq:G2def} \cG_2 := \bigg\{g \in \bit^{k\times\l}\,:\, \Big\lvert B_g(d,\by) - \E[B_g(d,\by)]\Big\lvert \leq \l^{-2\log\l}\E[B_g(d,\by)]\quad \text{~for all~} |d-\l p| \le 6\sqrt{{\l}}\log {\l} \bigg\}.
\end{equation}
Then by a simple union bound applied to \eqref{single_concentration} for all $d$ such that $|d-\l p| \leq 6\sqrt{\l}\log\l$ we obtain

$$
\P_{g\sim G}[g\in\cG_2]\ge 1-{\l}^{-\sqrt {\l}}. $$

\medskip
We are now ready to combine these bounds to get the needed concentration.
\begin{lem}
\label{binary_main_concentration_lem}
Fix $\by$.
With probability at least $1 - 2\l^{-2\log\l}$ over the choice of $g\sim G$, it holds that
\begin{equation}
    \label{weight_concentration}
(2^{k}-1)  2^{-{\l}} (1-2\l^{-2\log {\l}}) \le 
\sum_{d=0}^{\l} B_g(d, \by) p^d (1-p)^{{\l}-d}
    \le (2^{k}-1)  2^{-{\l}} (1+2\l^{-2\log {\l}}).
\end{equation}
\end{lem}
\begin{proof}
Indeed, by union bound $\P_{g\sim G}[g\in\cG_1\cap\cG_2] \geq 1 - \l^{-2\log\l} - \l^{-\sqrt{\l}} \geq 1 - 2\l^{-2\log\l}$. But for any $g \in \cG_1\cap\cG_2$ we derive
\begin{align*}
     \sum_{d=0}^{\l} B_g(d, \by) p^d (1-p)^{{\l}-d} 
    \ge & \sum_{|d-\l p| \le 6\sqrt{{\l}}\log {\l}} B_g(d, \by) p^d (1-p)^{{\l}-d} \\
    \overset{(a)}{\ge} & (2^{k}-1)  2^{-{\l}} (1-{\l}^{-2\log {\l}}) \sum_{|d-\l p| \le 6\sqrt{{\l}}\log {\l}} \binom{{\l}}{d} p^d (1-p)^{{\l}-d}  \\
    \overset{(b)}{\ge} & (2^{k}-1)  2^{-{\l}} (1-{\l}^{-2\log {\l}})
    (1- 2\l^{-12\log {\l}}) \\
    \ge & (2^{k}-1)  2^{-{\l}} (1-2\l^{-2\log {\l}}),
\end{align*}
where $(a)$ follows from~\eqref{eq:G2def} (since $g \in \cG_2$) and the expression in~\eqref{variance_expectation_weight} for $\E[B_g(d,\by)]$, and $(b)$ uses the concentration inequality for binomial r.v. from~\eqref{eq:chernoff_binomial}. On the other hand, we can upper bound this expression as
\begin{align*}
     \sum_{d=0}^{\l} B_g(d, \by) p^d &(1-p)^{{\l}-d} \\
    = & \sum_{|d-\l p| \le 6\sqrt{{\l}}\log {\l}} B_g(d, \by) p^d (1-p)^{{\l}-d} 
    + \sum_{|d-\l p| > 6\sqrt{{\l}}\log {\l}} B_g(d, \by) p^d (1-p)^{{\l}-d} \\
    \overset{(a)}{\le} & (2^{k}-1)  2^{-{\l}} (1+{\l}^{-2\log {\l}}) \sum_{|d-\l p| \le 6\sqrt{{\l}}\log {\l}} \binom{{\l}}{d} p^d (1-p)^{{\l}-d} + Z_g(\by) \\
    \overset{(b)}{\le} & (2^{k}-1)  2^{-{\l}} (1+{\l}^{-2\log {\l}})
    + 2(2^{k}-1)  2^{-{\l}} {\l}^{-10\log {\l}} \\
    \le & (2^{k}-1)  2^{-{\l}} (1+2\l^{-2\log {\l}}), 
\end{align*}
where $(a)$ is again from~\eqref{eq:G2def} and~\eqref{variance_expectation_weight} and the notation $Z_g(\by)$ for the negligible part, and $(b)$ is from~\eqref{eq:G1def} (as $g$ is in $\cG_1$).
\end{proof}

We similarly obtain the concentration for the sum in the numerator of \eqref{ratio}: with probability at least $1 - 2\l^{-2\log\l}$ over the choice of $g$, it holds
\begin{equation}
    \label{tilde_weight_concentration}
(2^{k-1}-1)  2^{-{\l}} (1-2\l^{-2\log {\l}}) \le 
\sum_{d=0}^{\l} \wB_g(d, \by) p^d (1-p)^{{\l}-d}
    \le (2^{k-1}-1)  2^{-{\l}} (1+2\l^{-2\log {\l}}).
\end{equation}

Next, let us use the fact that we took a typical output $\by$ with $|wt(\by) - \l p| \leq 2\sqrt{\l}\log\l$ to show that the terms $p^{wt(\by)}(1-p)^{\l-wt(\by)}$ are negligible in both numerator and denominator of \eqref{ratio}. We have
\begin{equation}
\label{eq:binary3}
     p^{wt(\by)}(1-p)^{\l-wt(\by)} = \left(\dfrac{1-p}{p}\right)^{\l p - wt(\by)} \cdot p^{\l p}(1-p)^{\l - \l p} = 2^{\left(\l p - wt(\by)\right) \cdot\log\left(\frac{1-p}{p}\right)} \cdot 2^{-\l h(p)}.
\end{equation}
Simple case analysis gives us:
\begin{enumerate}[label = (\alph*)]
    \item If $p < \frac1{\sqrt{\l}}$, then $\left(\l p - wt(\by)\right) \cdot\log\left(\frac{1-p}{p}\right) \leq \l p \log\frac1{p} < \l \frac1{\sqrt{\l}}\log\sqrt{\l} < \sqrt\l \log^2\l;$
    \item In case $p \geq \frac1{\sqrt{\l}}$, obtain $\left(\l p - wt(\by)\right) \cdot\log\left(\frac{1-p}{p}\right) \leq 2\sqrt{\l}\log\l \cdot \log\frac1{p} \leq \sqrt{\l}\log^2\l$.
\end{enumerate}
Using the above in~\eqref{eq:binary3} we derive for $k \geq \l(1-h(p)) + 8\sqrt{\l}\log^2\l$
\begin{equation}
\begin{aligned}
p^{wt(\by)}(1-p)^{\l-wt(\by)} \leq  2^{\sqrt{\l}\log^2\l -\l h(p)} \leq 2^{2\sqrt{\l}\log^2\l - \l h(p) - 2\log^2\l - 2} \leq \l^{-2\log\l}\,(2^{k-1}-1)2^{-\l}. 
\end{aligned}
\end{equation}  
Combining this with \eqref{weight_concentration} and \eqref{tilde_weight_concentration} and using a union bound we derive that with probability at least $1 - 4\l^{-2\log\l}$ it holds
\[ \left\lvert \left(p^{wt(\by)}(1-p)^{\l-wt(\by)} + \sum_{d=0}^{\l} B_g(d, \by) p^d (1-p)^{{\l}-d}\right)- (2^k-1)2^{-\l}  \right\rvert  \leq 3\l^{-2\log\l}\cdot(2^k-1)2^{-\l},  \]
\[ \left\lvert \left(p^{wt(\by)}(1-p)^{\l-wt(\by)} + \sum_{d=0}^{\l} \wB_g(d, \by) p^d (1-p)^{{\l}-d}\right)- (2^{k-1}-1)2^{-\l}  \right\rvert  \leq 3\l^{-2\log\l}\cdot(2^{k-1}-1)2^{-\l}. \]
Therefore, with probability at least $1 - 4\l^{-2\log\l}$ the expression in \eqref{ratio} is bounded as
\begin{equation}
\label{eq:BSC_ratio_1}
\dfrac{(1 - 3\l^{-2\log\l})(2^{k-1}-1)2^{-\l}}{(1 + 3\l^{-2\log\l})(2^{k}-1)2^{-\l}} \leq
         \frac{\P\nolimits^{(g)}(V_1=0,\bY=\by)}{\P\nolimits^{(g)}(\bY=\by)}       \leq 
      \dfrac{(1 + 3\l^{-2\log\l})(2^{k-1}-1)2^{-\l}}{(1 - 3\l^{-2\log\l})(2^{k}-1)2^{-\l}}.  
      \end{equation}
We can finally derive:
\begin{align}  \dfrac{(1 - 3\l^{-2\log\l})(2^{k-1}-1)}{(1 + 3\l^{-2\log\l})(2^{k}-1)} &\geq (1 - 6\l^{-2\log\l})\left(\dfrac12 - 2^{-k}\right)
\hspace{-27pt} &&\geq  (1 - 6\l^{-2\log\l})\left(\dfrac12 - \l^{-8\sqrt{\l}\log\l}\right)  \\
& &&\geq \dfrac12 - \l^{-\log\l}, \label{invis_eq:BSC_rat_1} \noeqref{invis_eq:BSC_rat_1} \\
  \dfrac{(1 + 3\l^{-2\log\l})(2^{k-1}-1)}{(1 - 3\l^{-2\log\l})(2^{k}-1)} &\leq \rlap{$(1 + 9\l^{-2\log\l})\dfrac12 \leq  \dfrac12 + \l^{-\log\l}.$}  && 
\end{align}

Therefore, with probability at least $1 - 4\l^{-2\log\l}$ over $g\sim G$ it holds
\begin{equation}
\label{eq:BSC_ratio_2}  \left\lvert
         \frac{\P\nolimits^{(g)}(V_1=0,\bY=\by)}{\P\nolimits^{(g)}(\bY=\by)}  - \dfrac12 \right\rvert     \leq  \l^{-\log\l}. 
\end{equation}
Since $h(1/2 + x) \geq 1 - 4x^2$  for any $x\in [-1/2,1/2]$ (\cite[Theorem~1.2]{Topsoe}), we then derive:
\[ \E_{g\sim G}\big[ H^{(g)}(V_1|\bY=\by)\big] = \E_{g\sim G}\hspace{-3pt}\left[h\hspace{-3pt}\left(\dfrac{\P\nolimits^{(g)}(V_1=0,\bY=\by)}{\P\nolimits^{(g)}(\bY=\by)}\right)\hspace{-2pt}\right]\hspace{-1.5pt}\geq \hspace{-1.3pt}  (1 - 4\l^{-2\log\l}) (1- 4\l^{-2\log\l}) \geq 1 - 8\l^{-2\log\l}.   \]

\medskip \noindent{\bf Concentration of entropy.}\quad
We are now ready to plug this into \eqref{entr_exp_weights_only_typical}:
\begin{align}
 \E_{g\sim G}\big[H^{(g)}(V_1|\bY) \big] &\geq (1 - 8\l^{-2\log\l})\sum_{|wt(\by) - \l p| \leq 2\sqrt{\l}\log\l} p^{wt(\by)}(1-p)^{{\l}-wt(\by)} \nonumber\\
&= (1 - 8\l^{-2\log\l})\sum_{|d - \l p| \leq 2\sqrt{\l}\log\l} \binom{\l}{d}p^d(1-p)^{\l-d} \nonumber\\
&= (1 - 8\l^{-2\log\l})\P_{X \sim \text{Binom}(\l, p)}\left[|X - \l p| \leq 2\sqrt{\l}\log\l\right]\nonumber\\
&\geq (1 - 8\l^{-2\log\l})(1-2e^{-(4\log^2\l)/3p})\nonumber\\
&\geq  (1 - 8\l^{-2\log\l})(1-2\l^{-2\log\l}) \nonumber\\
&\geq 1 - 10\l^{-2\log\l}, \label{eq:bound_exp_entropy}
\end{align}
where the second inequality is obtained from the Chernoff bound~\eqref{new_eq:multiplicative_Chernoff} with $\delta = \frac{2\log\l}{p\l^{1/2}}$, and the third inequality follows from $p \leq 1/2$ and $e^{-8/3} < 2^{-2}$.
Finally, using the fact that $H^{(g)}(V_1|\bY) \leq 1$, Markov's inequality, and \eqref{eq:bound_exp_entropy}, we get
\begin{equation}
\P_{g\sim G}\hspace{-1pt}\big[H^{(g)}(V_1|\bY) \leq 1 - \l^{-\log\l} \big] = \P_{g\sim G}\hspace{-1pt}\big[ 1 - H^{(g)}(V_1|\bY) \geq \l^{-\log\l}  \big]   \leq \dfrac{ \E\limits_{g\sim G}\big[1 - H^{(g)}(V_1|\bY)\big]}{\l^{-\log\l}} \leq 10\l^{-\log\l}.
\end{equation}
Thus we conclude that with probability at least $1 - 10\l^{-\log\l}$ over the choice of the kernel $G$ it holds that $H(V_1\,|\,\bY) \geq 1 - \l^{-\log\l}$ when $k \geq \l(1-h(p)) + 8\sqrt{\l}\log^2\l$ and the underlying channel is BSC. This completes the proof of Theorem~\ref{thm:over:BSC_converse}, which is a version of  Theorem~\ref{thm:converse_Shannon_BMS} for the BSC case.
\end{proof}

\section{Strong converse for BMS channel} 
\label{sec:bit-decoding}

To make this section completely self-contained, we restate the theorem here:
\converseShannon*

\subsection{Bounded alphabet size}
\label{sec:BMS_large_alphabet}
This section is devoted to proving Theorem~\ref{thm:converse_Shannon_BMS} for the case when $W\,:\,\bit \to \Y$ is a BMS channel which has a bounded output alphabet size, specifically we consider $|\Y| \leq 2\sqrt{\l}$. 
We will use the fact that any BMS can be viewed as a convex combination of BSCs (see for example \cite{Land, Korada_thesis}), and generalize the ideas of the previous section. Namely, think of the channel $W$ as follows: it has $m$ possible underlying BSC subchannels $W^{(1)}, W^{(2)}, \dots, W^{(m)}$. On any input, $W$ randomly chooses one of the subchannels it is going to use with probabilities $q_1, q_2, \dots, q_m$ respectively. The subchannel $W^{(j)}$ has crossover probability $p_j$, and without loss of generality $0 \leq p_1 \leq p_2 \leq\dots\leq p_m \leq \frac12$. The subchannel $W^{(j)}$ has two possible output symbols $z^{(0)}_j$ or $z^{(1)}_j$, corresponding to $0$ and $1$, respectively (i.e. $0$ goes to $z^{(0)}_j$ with probability $1-p_j$, or to $z^{(1)}_j$ with probability $p_j$ under $W^{(j)}$). Then the whole output alphabet is $\Y = \{z^{(0)}_1, z^{(1)}_1, z^{(0)}_2, z^{(1)}_2, \dots, z^{(0)}_m, z^{(1)}_m\}$, $|\Y| = 2m \leq 2\sqrt{\l}$. For the conditional entropy of the BMS channel $W$ we have $H(W) = \sum\limits_{i=1}^m q_ih(p_i)$, i.e. it is a convex combination of entropies of the subchannels $W^{(1)}, W^{(2)}, \dots, W^{(m)}$ with the corresponding coefficients $q_1, q_2, \dots, q_m$.
\begin{remark}
Above we ignored the case when some of the subchannels have only one output (i.e. BEC subchannels), see~\cite[Lemma 4]{Tal_Vardy} for a proof that we can do this without loss of generality.
\end{remark}

\vspace{0.3cm}
\noindent\textbf{Notations and settings.}\quad In this section the expectation is only going to be taken over the kernel $g\sim G$, so we omit this in some places. As in the BSC case, by $\P\nolimits^{(g)}[\cdot]$ and $H^{(g)}(\cdot)$ we denote the probability and entropy only over the randomness of the channel and the message, \emph{for a fixed kernel~$g$}. 

 For any possible output $\by\in\Y^\l$ we denote by $d_i$ the number of symbols from $\{z^{(0)}_i, z^{(1)}_i\}$ it has (i.e. the number of uses of the $W^{(i)}$ subchannel), so $\sum_{i=1}^md_i = \l$. Let also $t_i$ be the number of symbols $z^{(1)}_i$ in $\by$. Then
\begin{equation}
\label{def_prob_subchannels}
 \P[\bY=\by|\bV=\mathbi{0}]= \prod_{i=1}^mq_i^{d_i}p_i^{t_i}(1-p_i)^{d_i-t_i}.
\end{equation}
For this case of bounded output alphabet size, we will consider the above-capacity regime when $k \geq \l(1-H(W)) + 13\l^{1/2}\log^3\l$ (note that this is made intentionally weaker than the condition in Theorem~\ref{thm:converse_Shannon_BMS}). 

\vspace{0.3cm}

We will follow the same blueprint of the proof for BSC from Section~\ref{sect:outline}, however all the technicalities along the way are going to be more challenging.  In particular, while we were dealing with one binomial distribution in Section~\ref{sec:BSC_converse}, here we will face a multinomial distribution of $(d_1, d_2,\dots,d_m)$ as a choice of which subchannels to use, as well as binomial distributions $t_i \sim \text{Binom}(d_i, p_i)$ which correspond to ``flips" within one subchannel. 

\vspace{0.3cm}

\begin{proof}[Proof of Theorem~\ref{thm:converse_Shannon_BMS}]
As in the BSC case, we are going to lower bound the expectation of $H^{(g)}(V_1|\bY)$ and use Markov's inequality afterwards. 

\medskip \noindent{\bf Restrict to zero-input.}\quad  We use Lemma~\ref{typical_entropy_BSC} to write 

\begin{equation}
\label{mult_entropy_formula}
 \E_{g\sim G}\big[H^{(g)}(V_1|\bY)\big] = \sum_{\by\in\Y^{\l}}\P[\bY=\by|\bV=\mathbi{0}]\E_{g\sim G}\big[H^{(g)}(V_1|\bY=\by)\big].   
\end{equation}
Notice that there is no dependence of $\P[\bY=\by|\bV=\mathbi{0}]$ on the kernel $g$, since the output for the zero-input depends only on the randomness of the channel.

\paragraph*{Typical output set}

As for the binary case, we would like to consider the set of ``typical" outputs (for input $\mathbi{0}$) from $\Y^{\l}$. We define $\by\in\Y^{\l}$ to be typical if 
\begin{align}
\label{close_entropy_sum} \sum_{i=1}^m (\l\cdot q_i - d_i)h(p_i) &\leq 2\sqrt{\l}\log\l,\\
\label{close_coordinates_sum}  \sum_{i=1}^m (p_id_i - t_i)\log\left(\dfrac{1-p_i}{p_i}\right) &\leq 3\sqrt{\l}\log^2\l.
\end{align}
By typicality of this set we mean the following
\begin{lem}
\label{typical_lem}
$\sum\limits_{\by \text{ typical}}\P[\bY=\by|\bV=\emph{\mathbi{0}}] \geq 1 - \l^{-\log\l}$. In other words, on input $\emph{\mathbi{0}}$, the probability to get the output string which is not typical is at most $\l^{-\log\l}$.
\end{lem}
We defer the proof of this lemma until Section \ref{typical_proof_sect}, after we see why we are actually interested in these conditions on $\by$.

\subsubsection{Fix a typical output}
For this part, let us fix one $\by\in\Y^{\l}$ which is typical and prove that $\E_g\big[H^{(g)}(V_1|\bY)\big]$ is very close to~$1$. We have 
\begin{equation}
    \label{mult_entropy_ratio}
 H^{(g)}(V_1|\bY) = h\left(\dfrac{\P\nolimits^{(g)}\big[V_1=0, \bY = \by\big]}{\P\nolimits^{(g)}\big[\bY = \by\big]}\right).  
 \end{equation}

Similarly to the BSC case, we will prove that both the denominator and numerator of the fraction inside the entropy function above are tightly concentrated around their means. The arguments for the denominator and the numerator are almost exactly the same, so we only consider denominator for now.

\subsubsection*{Concentration for $\P\nolimits^{(g)}\big[\bY = \by\big]$}
Define now the shifted weight distributions for the codebook $g$ with respect to $m$ different underlying BSC channels. First, for any $x\in\bit^{\l}$ and $i = 1, 2, \dots, m$, define 
\[ \text{dist}_i(x,\by) = \lvert\{ \text{positions } j \text{ such that } (x_j=0, \by_j=z^{(1)}_i)\text{ or }(x_j=1, \by_j=z^{(0)}_i)  \}\rvert.      \]
That is, if you send $x$ through $W^{\l}$ and receive $\by$, then dist$_i(x,\by)$ is just the number of coordinates where the subchannel $i$ was chosen, and the bit was flipped. 

In our settings, we now need to think of ``distance" between some binary vector $x\in\bit^{\l}$ and $\by$ as of an integer vector $\bs = (s_1, s_2, \dots, s_m)$ , where $0\leq s_i \leq d_i$ for $i\in[m]$, where $s_i = \text{dist}_i(x,\by)$ is just the number of flips that occurred in the usage of $i^{\text{th}}$ subchannel when going from $x$ to $\by$. In other words, $s_i$ is just the Hamming distance between the parts of $x$ and $\by$ which correspond to coordinates $j$ where $\by_j$ is $z_i^{(0)}$ or $z_i^{(1)}$ (coming from the subchannel $W^{(i)}$).

Now we can formally define shifted weight distributions for our fixed typical $\by$. For an integer vector $\bs = (s_1, s_2, \dots, s_m)$ , where $0\leq s_i \leq d_i$ define
\[ B_g(\bs, \by) = \Big\lvert \bv\in\bit^k\setminus\mathbi{0} \ :\ \text{dist}_i(\bv\cdot g, \by) = s_i\quad \text{for } i=1, 2,\dots, m\Big\rvert.   \]

We can express $\P\nolimits^{(g)}[\Y=\by]$ in terms of $B_{g}(\bs,\by)$ as follows:
\begin{equation}
\label{prob_of_y}
    2^k\cdot\P\nolimits^{(g)}[\Y=\by] =  \P[\Y=\by|\bv=\mathbi{0}] + \sum_{0 \leq s_j \leq d_j \atop{j = 1, 2, \dots, m}}B_g(\bs,\by)\prod_{i=1}^mq_i^{d_i}p_i^{s_i}(1-p_i)^{d_i-s_i}, 
\end{equation} 
because $\prod_{i=1}^mq_i^{d_i}p_i^{s_i}(1-p_i)^{d_i-s_i}$ is exactly the probability to get output $\by$ if a $\bv$ is sent that satisfies $\text{dist}_i(\bv\cdot g, \by) = s_i\quad \text{for } i=1, 2,\dots, m$. 

We have:

\begin{equation}
    \label{shifted_weight_dist}
B_g(\bs, \by) = \sum_{\bv\not=\mathbi{0}}\mathbbm{1}\big[\text{dist}_i(\bv\cdot g, \by) = s_i, \quad \forall i=1, 2, \dots, m\big].
\end{equation}
For a fixed $\bv$ but uniformly random binary matrix $g$, the vector $\bv\cdot g$ is just a uniformly random vector from $\bit^{\l}$. Now, the number of vectors $x$ in $\bit^{\l}$ such that $\text{dist}_i(x, \by) = s_i\ \  \forall i=1, 2, \dots, m$ is $\prod_{i=1}^m\binom{d_i}{s_i}$, since for any $i=1, 2,\dots, m$, we need to choose which of the $s_i$ coordinates amongst the $d_i$ uses of the subchannel $W^{(i)}$, got flipped. So
\[\underset{g\sim G}{\P} \big[ \text{dist}_i(\bv\cdot g,\by) = s_i, \ \ \forall i=1, 2,\dots, m\big] = 2^{-\l}\prod_{i=1}^m\binom{d_i}{s_i}.   \]

Then for the expectation of the shifted weight distributions we obtain
\begin{equation}
\label{eq:ExpBgsy}
    \E_{g\sim G} [B_g(\bs, \by)] = \sum_{\bv\not=\mathbi{0}} \underset{g\sim G}{\P} \big[ \text{dist}_i(\bv\cdot g,\by) = s_i, \ \ \forall i=1, 2,\dots, m\big] = \dfrac{2^k-1}{2^{\l}}\prod_{i=1}^m\binom{d_i}{s_i}.  
\end{equation}
Then for the expectation of the summation in the RHS of \eqref{prob_of_y} we have:
\begin{align}
    E \coloneqq& \E_{g\sim G}\left[\sum_{0 \leq s_j \leq d_j \atop{j = 1, 2, \dots, m}}B_g(\bs,\by)\prod_{i=1}^mq_i^{d_i}p_i^{s_i}(1-p_i)^{d_i-s_i}\right] \\
    =& \left(\prod_{i=1}^m q_i^{d_i}\right)\cdot \sum_{0 \leq s_j \leq d_j \atop{j = 1, 2, \dots, m}}\left(\E_{g\sim G}\big[B_g(\bs, \by)\big] \cdot \prod_{i=1}^m p_i^{s_i}(1-p_i)^{d_i-s_i}\right) \\
    =& \dfrac{2^k-1}{2^{\l}}\left(\prod_{i=1}^mq_i^{d_i}\right)\cdot\sum_{0 \leq s_j \leq d_j \atop{j = 1, 2, \dots, m}}\prod_{i=1}^m\binom{d_i}{s_i}p_i^{s_i}(1-p_i)^{d_i-s_i} \\
=& \dfrac{2^k-1}{2^{\l}}\prod_{i=1}^mq_i^{d_i}\cdot\prod_{i=1}^m\left(\underbrace{\sum_{s_i=0}^{d_i} \binom{d_i}{s_i}p_i^{s_i}(1-p_i)^{d_i-s_i}}_{=1}\right) = \dfrac{2^k-1}{2^{\l}}\prod_{i=1}^mq_i^{d_i}.    \label{exp_of_sum}
   \end{align}
Next, by \eqref{shifted_weight_dist} we can see that $B_g(\bs, \by)$ is a sum of pairwise independent indicator random variables, since $\bv_1\cdot g$ and $\bv_2\cdot g$ are independent for distinct and non-zero $\bv_1, \bv_2$. Therefore
\begin{equation} 
\label{var-exp}
\underset{g\sim G}{\Var}[B_g(\bs, \by)] \leq \E_{g\sim G}[B_g(\bs,\by)].
\end{equation}

\subsubsection*{Splitting the summation in \eqref{prob_of_y}}
We will split the summation in \eqref{prob_of_y} into two parts: for the first part, we will show that the expectation of each term is very large, and then use Chebyshev's inequality to argue that each term is concentrated around its expectation. For the second part, its expectation is going to be very small, and Markov's inequality will imply that this part also does not deviate from its expectation too much with high probability (over the random kernel $g\sim G$). Putting these two arguments together, we will obtain that the sum in the RHS of \eqref{prob_of_y} is concentrated around its mean. 

To proceed, define a distribution $\Omega = \text{Binom}(d_1, p_1)\times\text{Binom}(d_2, p_2)\times\dots\times\text{Binom}(d_m, p_m)$, and consider a random vector $\chi\sim\Omega$.  In other words, $\chi$ has $m$ independent coordinates $\chi_i,\ i=1,\dots,m$, where $\chi_i$ is a binomial random variable with parameters $d_i$ and $p_i$. Note that by definition then for any vector $\bs = (s_1, s_2, \dots, s_m)$ , where $0\leq s_i \leq d_i$ and $s_i$ is integer for any $i$, we have
\[ \P_{\chi}[ \chi=\bs ] = \prod_{i=1}^m\P_{\chi}[\chi_i=s_i] = \prod_{i=1}^m\binom{d_i}{s_i}p_i^{s_i}(1-p_i)^{d_i-s_i}. \]

Let now $\T$ be some subset of $\S = [0:d_1]\times[0:d_2]\times\dots\times[0:d_m]$, where $[0:d] = \{0, 1, 2, \dots, (d-1), d\}$ for integer $d$. Let also $\N$ be $\S\setminus \T$. Then the summation in the RHS of \eqref{prob_of_y} we can write as
\begin{equation}
\label{split_sum}
\begin{aligned}
 \hspace{-5pt}\sum_{\bs\in\S} \hspace{-1.5pt} B_g(\bs,\by)\hspace{-3pt}\prod_{i=1}^mq_i^{d_i}p_i^{s_i}(1\hspace{-2pt}-\hspace{-2pt}p_i)^{d_i-s_i} =  \sum_{s\in \T}\hspace{-1.5pt}B_g(\bs,\by)\hspace{-3pt}\prod_{i=1}^mq_i^{d_i}p_i^{s_i}(1\hspace{-2pt}-\hspace{-2pt}p_i)^{d_i-s_i} +\hspace{-3pt}\sum_{s\in \N}\hspace{-1.5pt}B_g(\bs,\by)\hspace{-3pt}\prod_{i=1}^mq_i^{d_i}p_i^{s_i}(1\hspace{-2pt}-\hspace{-2pt}p_i)^{d_i-s_i}.
\end{aligned}
\end{equation} 
In the next section we describe how to choose $\T$.

\paragraph{Substantial part}
\label{substantial_sect}
Exactly as in the binary case, using \eqref{var-exp} and Chebyshev's inequality, we have for any $s\in \S$
\begin{multline}
\label{chebyshev_multi}
\P_{g\sim G}\bigg[\Big\lvert B_g(\bs,\by) - \E[B_g(\bs,\by)]\Big\lvert \geq \l^{-2\log\l}\E[B_g(\bs,\by)]\bigg] \leq \dfrac{\Var[B_g(\bs,\by)]}{\l^{-4\log\l}\E^2[B_g(\bs,\by)]} \\
 \leq \dfrac{\l^{4\log\l}}{\E_{g\sim G}[B_g(\bs,\by)]}
 \leq \l^{4\log\l}\,\dfrac{2^{\l-k+1}}{\prod_{i=1}^m\binom{d_i}{s_i}}. 
\end{multline}

We need the above to be upper bounded by $\l^{-2\sqrt{\l}}$ to be able to use union bound for all ${\bs \in \T \subset \S}$, since $|\S| \leq \l^{O(\sqrt{\l})}$. Recall that we have $k \geq \l(1-H(W)) + 13\l^{1/2}\log^3\l$, and then using a lower bound for binomial coefficients from Fact~\ref{binom_lemma} we obtain for the RHS of \eqref{chebyshev_multi}
 \begin{equation}
 \label{upper_bound}
      \l^{4\log\l}\,\dfrac{2^{\l-k+1}}{\prod_{i=1}^m\binom{d_i}{s_i}} \leq \l^{4\log\l}\cdot\left(2\prod_{i=1}^m\sqrt{2d_i}\right)\cdot 2^{\l H(W) - \sum_{i=1}^md_ih\left(\frac{s_i}{d_i}\right) - 13\l^{1/2}\log^3\l}.   
 \end{equation}
 We want to show that the term  $2^{-\O(\l^{1/2}\log^3\l)}$ is the dominant one. First, it is easy to see that $\l^{4\log\l} = 2^{4\log^2\l} \leq 2^{\l^{1/2}\log^3\l}$ for $\l \geq 4$. To deal with the factor $2\prod_{i=1}^m\sqrt{2d_i}$, recall that $\sum_{i=1}^md_i = \l$ and $m \leq \sqrt{\l}$ in this section (recall discussion at the beginning of Section~\ref{sec:BMS_large_alphabet}), then AM-GM inequality gives us
 \begin{equation}
 \label{prod_of_sqrt}
 2\prod_{i=1}^m\sqrt{2d_i} \leq 2\cdot 2^{m/2}\cdot \sqrt{\left(\dfrac{\sum_{i=1}^md_i}{m}\right)^m} = 2\cdot \left(\dfrac{2\l}{m}\right)^{m/2} \leq 2\cdot (2\sqrt{\l})^{\sqrt{\l}/2} \leq 2^{\l^{1/2}\log^3\l},  
 \end{equation}
 where we used the fact that $(a/x)^x$ is increasing while $x \leq a/e$ and the condition $\l \geq 4$. For the last factor of~\eqref{upper_bound} we formulate a lemma.

 \begin{lem}
 \label{multi_sum_lem} \sloppy  There exists a set $\T \subseteq \S = [0:d_1]\times[0:d_2]\times\dots\times[0:d_m]$, such that ${\P\limits_{\chi\sim\Omega}[\chi\in\T] \geq 1 - \l^{-(\log\l)/4}}$, and for any $\bs\in\T$ it holds that
 \begin{equation}
     \label{multi_sum_lem_eq}
 \l H(W) - \sum_{i=1}^{m}d_ih\left(\frac{s_i}{d_i}\right) \leq 8\,\l^{1/2}\log^3\l.
 \end{equation}
 ($\Omega = \text{Binom}(d_1, p_1)\times\text{Binom}(d_2, p_2)\times\dots\times\text{Binom}(d_m, p_m)$ above)
 \end{lem}
 \begin{proof}
 Rearrange the above summation as follows:
 \begin{equation}
 \begin{aligned}
      \l H(W) - \sum_{i=1}^{m}d_ih\left(\frac{s_i}{d_i}\right) =& \sum_{i=1}^m\left(\l q_ih(p_i)  - d_ih\left(\frac{s_i}{d_i}\right)\right) \\
      =& \sum_{i=1}^m\big(\l q_i - d_i\big)h(p_i) + \sum_{i=1}^md_i\left(h(p_i) - h\left(\frac{s_i}{d_i}\right)\right).
      \end{aligned}
 \end{equation}
 
 Now recall that we took typical $\by$ for now, so by inequality \eqref{close_entropy_sum} from the definition of the typicality of $\by$ we already have that the first part of the above sum is bounded by $\l^{1/2}\log^3\l$.

To deal with the second part, which is $\sum_{i=1}^md_i\left(h(p_i) - h\left(\frac{s_i}{d_i}\right)\right)$, we use a separate Lemma~\ref{two_concentrations_lem}, since the proof will be almost exactly similar for another concentration inequality we will need later. Lemma~\ref{two_concentrations_lem} claims that $\sum_{i=1}^md_i\left(h(p_i) - h\left(\frac{\chi_i}{d_i}\right)\right) \leq 7\l^{1/2}\log^3\l$ with probability at least $ 1 - \l^{-(\log\l)/4}$  over $\chi\sim \O$. Then the result of the current lemma follows by taking $\T$ to be the subset of $\S$ where this inequality holds.
 \end{proof}

 Fix now a set $\T\subseteq\S$ as in Lemma~\ref{multi_sum_lem}. Then using the arguments above we conclude that the RHS in \eqref{upper_bound}, and therefore \eqref{chebyshev_multi}, is bounded above by $2^{-3\l^{1/2}\log^3\l}$ for any $\bs\in\T$. Thus we can apply union bound over $\bs\in\T$ for \eqref{chebyshev_multi}, since $|\T|\leq |\S| = \prod_{i=1}^m(d_i+1) \leq \left(2\sqrt{\l}\right)^{\sqrt{\l}} \leq 2^{\l^{1/2}\log^3\l}$ for $\l \geq 4$, similarly to \eqref{prod_of_sqrt}. Therefore, we derive
 \begin{cor}
 \label{substantial_cor}
 With probability at least $1 - 2^{-2\l^{1/2}\log^3\l}$ (over the random kernel $g\sim G$) it holds \textit{simultaneously} for all $\bs\in\T$ that 
 \begin{equation}
 \label{substantial_part_final}
 \Big\lvert B_g(\bs,\by) - \E[B_g(\bs,\by)] \Big\lvert \leq \l^{-2\log\l}\E[B_g(\bs,\by)]. 
 \end{equation}
 \end{cor}
 Moreover, the set $\N = \S\setminus\T$ satisfies $\P_{\chi\sim\O}[\chi\in\N] \leq \l^{-(\log\l)/4}$, which we will use next section to bound the second part of \eqref{split_sum}.

\paragraph{Negligible part}
\label{neglig_sect}
Denote for convenience $Z_g(\by) = \sum_{s\in\N}B_g(\bs,\by)\prod_{i=1}^mq_i^{d_i}p_i^{s_i}(1-p_i)^{d_i-s_i}$, the second part of the RHS of \eqref{split_sum}. Recall the value of $\E_{g\sim G}[B_g(\bs,\bY)]$ from~\eqref{eq:ExpBgsy} and notation of $E$ in~\eqref{exp_of_sum}. Then for the expectation of $Z_g(\by)$ derive
\begin{equation} 
\begin{aligned} 
\E_{g\sim G}\left[Z_g(\by)\right] &= \left(\prod_{i=1}^mq_i^{d_i}\right)\cdot \sum_{\bs\in\N}\left(\E_{g\sim G}\big[B_g(\bs, \by)\big] \prod_{i=1}^m p_i^{s_i}(1-p_i)^{d_i-s_i}\right) \\
&= \dfrac{2^k-1}{2^{\l}}\left(\prod_{i=1}^mq_i^{d_i}\right)\cdot\sum_{\bs\in\N}\prod_{i=1}^m\binom{d_i}{s_i}p_i^{s_i}(1-p_i)^{d_i-s_i} \\
&= E\cdot\P_{\chi\sim\Omega}\big[\chi\in\N\big] \\
&\leq E\cdot\l^{-(\log\l)/4}.
\end{aligned}
\end{equation}
Thus Markov's inequality implies
\begin{cor}
\label{negligible_cor}
With probability at least $1 - \l^{-(\log\l)/8}$ (over the random kernel $g\sim G$) it holds 
\begin{equation}
\label{negligible_part_final}
Z_g(\by) \leq \l^{(\log\l)/8}\E[Z_g(\by)] \leq E\cdot \l^{-(\log\l)/8}.
\end{equation}
\end{cor}

\paragraph{Putting it together}
Combining the Corollaries \ref{substantial_cor} and \ref{negligible_cor} together and using the union bound, we derive
\begin{cor}
\label{together_cor}
With probability at least $1 - \l^{-(\log\l)/8} - 2^{-2\l^{1/2}\log^3\l} \geq 1 - 2\l^{-(\log\l)/8}$ over the randomness of the kernel $g\sim G$ it simultaneously holds
\begin{equation}
\label{together_simul}
    \begin{aligned}
         &\Big\lvert B_g(\bs,\by) - \E[B_g(\bs,\by)]\Big\lvert \leq \l^{-2\log\l}\E[B_g(\bs,\by)], \qquad\qquad \text{for all } \bs \in \T,\\
         &\sum_{s\in\N}B_g(\bs,\by)\prod_{i=1}^mq_i^{d_i}p_i^{s_i}(1-p_i)^{d_i-s_i} \leq E\cdot \l^{-(\log\l)/8}.
    \end{aligned}
\end{equation}
\end{cor}

We are finally ready to formulate the concentration result we need. The following lemma is an analogue of Lemma \ref{binary_main_concentration_lem} from the BSC case:

\begin{lem}
\label{mult_main_concentration_lem}
With probability at least $1 - 2\l^{-(\log\l)/8}$ over the choice of $g\sim G$ it holds
\begin{equation}
    \label{mult_weight_concentration}
    \left|\sum_{\bs\in\S} B_g(\bs, \by)\prod_{i=1}^mq_i^{d_i}p_i^{s_i}(1-p_i)^{d_i-s_i} - E \right| \leq 2\l^{-(\log\l)/8}\cdot E.
\end{equation}
\end{lem}
\begin{proof}
Let us consider a kernel $g$ such that the conditions \eqref{together_simul} hold, which happens with probability at least $1 - 2\l^{-(\log\l)/8}$ according to Corollary \ref{together_cor}. Then
\begin{equation}
\begin{aligned}
     \sum_{\bs\in\S} B_g(\bs, \by)\prod_{i=1}^mq_i^{d_i}p_i^{s_i}(1-p_i)^{d_i-s_i} &\geq \sum_{\bs\in\T} B_g(\bs, \by)\prod_{i=1}^mq_i^{d_i}p_i^{s_i}(1-p_i)^{d_i-s_i}\\
     &\overset{\eqref{together_simul}}{\geq} \sum_{\bs\in\T} \left(1 - \l^{-2\log\l}\right)\E[B_g(\bs, \by)]\prod_{i=1}^mq_i^{d_i}p_i^{s_i}(1-p_i)^{d_i-s_i}\\
     & \overset{\eqref{eq:ExpBgsy}}{=} \left(1 - \l^{-2\log\l}\right)\dfrac{2^k-1}{2^{\l}}\prod_{i=1}^mq_i^{d_i}\cdot\sum_{\bs\in\T}\prod_{i=1}^m\binom{d_i}{s_i}p_i^{s_i}(1-p_i)^{d_i-s_i}\\
     & = \left(1 - \l^{-2\log\l}\right)\cdot E\cdot\P_{\chi\sim\O}\big[\chi\in\T\big]\\
     & \geq \left(1 - \l^{-2\log\l}\right)\left(1 - \l^{-(\log\l)/8}\right)E\\
     & \geq \left(1 - 2\l^{-(\log\l)/8}\right)E.
\end{aligned}
\end{equation}
For the other direction, we derive for such $g$
\begin{align}
     \sum_{\bs\in\S} B_g(\bs, \by)&\prod_{i=1}^mq_i^{d_i}p_i^{s_i}(1-p_i)^{d_i-s_i} = 
     \left(\sum_{\bs\in\T} + \sum_{\bs\in\N}\right) B_g(\bs, \by)\prod_{i=1}^mq_i^{d_i}p_i^{s_i}(1-p_i)^{d_i-s_i} \\
     &\overset{\eqref{together_simul}}{\leq} \sum_{\bs\in\T} \left(1 + \l^{-2\log\l}\right)\E[B_g(\bs, \by)]\prod_{i=1}^mq_i^{d_i}p_i^{s_i}(1-p_i)^{d_i-s_i} + E\cdot \l^{-(\log\l)/8}\\
     &\leq \left(1 + \l^{-2\log\l}\right)\underbrace{\sum_{\bs\in\S} \E[B_g(\bs, \by)]\prod_{i=1}^mq_i^{d_i}p_i^{s_i}(1-p_i)^{d_i-s_i}}_{E} + E\cdot \l^{-(\log\l)/8}\\
     &= \left(1 + \l^{-2\log\l} + \l^{-(\log\l)/8}\right)E\\
     &\leq \left(1 + 2\l^{-(\log\l)/8}\right)E. \qedhere
\end{align} 
\end{proof}

\subsubsection{Concentration of entropy}
We can now get a tight concentration for $\P\nolimits^{(g)}[\bY = \by]$ using the relation~\eqref{prob_of_y}. We already showed that the sum in RHS of \eqref{prob_of_y} is tightly concentrated around its expectation, so it only remains to show that $\P[\bY = \by|\bv=\mathbi{0}]$ is tiny compared to $E$. Here we will use that we picked $\by$ to be ``typical" from the start so that~\eqref{close_entropy_sum} and~\eqref{close_coordinates_sum} hold, and that we consider here the above-capacity regime. Recall~\eqref{def_prob_subchannels}, as well the the conditions~\eqref{close_entropy_sum} and~\eqref{close_coordinates_sum} on $\by$ being typical. We derive
\begingroup
\allowdisplaybreaks
\begin{align}
\P[\bY = \by|\bV = \mathbi{0}] &= \prod_{i=1}^mq_i^{d_i}p_i^{t_i}(1-p_i)^{d_i-t_i} = \prod_{i=1}^m\left[q_i^{d_i}\cdot p_i^{d_ip_i}(1-p_i)^{d_i(1-p_i)}\cdot\left(\dfrac{1-p_i}{p_i}\right)^{d_ip_i-t_i}\right] \\
&=   \prod_{i=1}^mq_i^{d_i} \cdot \prod_{i=1}^m2^{-d_ih(p_i)}\cdot\prod_{i=1}^m2^{(d_ip_i-t_i)\log\left(\frac{1-p_i}{p_i}\right)}\\
&= \prod_{i=1}^mq_i^{d_i} \cdot 2^{\sum_{i=1}^m\left(-\l q_i h(p_i) + (\l q_i-d_i)h(p_i)\right)}\cdot 2^{\sum_{i=1}^m(d_ip_i - t_i)\log\left(\frac{1-p_i}{p_i}\right)} \\
&\hspace{-12pt}\overset{\eqref{close_entropy_sum}, \eqref{close_coordinates_sum}}{\leq} \prod_{i=1}^mq_i^{d_i}\cdot2^{-\l H(W) + 2\l^{1/2}\log\l + 3\l^{1/2}\log^2\l} \\
& \leq \prod_{i=1}^mq_i^{d_i}\cdot \dfrac{2^{k}-1}{2^{\l}}\cdot\l^{-\log\l} \,\, = \,\, E\cdot \l^{-\log\l}, \label{invis_eq:712} \noeqref{invis_eq:712}
\end{align}
\endgroup
where the last inequality follows from $k\geq \l(1-H(W)) +13\l^{1/2}\log^3\l$.

Now, combining this with Lemma \ref{mult_main_concentration_lem}, we obtain a concentration for \eqref{prob_of_y}:
\begin{cor}
\label{mult_final_conc_cor}
With probability at least $1 - 2\l^{-(\log\l)/8}$ over the choice of kernel $g\sim G$ and for any typical $\by$ 
\begin{equation}
\label{cor_main_concentration}
\left|2^k\cdot\P\nolimits^{(g)}[\bY=\by] - E\right| \leq 3\l^{-(\log\l)/8}\cdot E,
\end{equation}
where $E = \dfrac{2^k-1}{2^{\l}}\prod\limits_{i=1}^mq_i^{d_i}$.
\end{cor}

Next, completely analogously we derive the concentration for $\P\nolimits^{(g)}[\bY=\by|V_1=0]$, which is the numerator inside the entropy in \eqref{mult_entropy_ratio}. The only thing that changes is that we will have dimension $k-1$ instead of $k$ for this case. We can state

\begin{corbis}{mult_final_conc_cor}
\label{mult_final_conc_cor_prime}
With probability at least $1 - 2\l^{-(\log\l)/8}$ over the choice of kernel $g\sim G$ and for any typical $\by$ 
\begin{equation}
\label{cor_main_concentration_prime}
\left|2^k\cdot\P\nolimits^{(g)}[V_1=0, \bY=\by] - \widetilde{E}\right| \leq 3\l^{-(\log\l)/8}\cdot \widetilde{E},
\end{equation}
where $\widetilde{E} = \dfrac{2^{k-1}-1}{2^{\l}}\prod\limits_{i=1}^mq_i^{d_i}$.
\end{corbis}

Combining these two together and skipping the simple math, completely analogical to that of the BSC case in~\eqref{eq:BSC_ratio_1}--\eqref{eq:BSC_ratio_2}, we derive
\begin{cor}
With probability at least $1 - 4\l^{-(\log\l)/8}$ over the choice of kernel $g\sim G$ and for any typical $\by$,
\begin{equation}
    \left|\dfrac{\P\nolimits^{(g)}[V_1=0, \bY=\by]}{\P\nolimits^{(g)}[\bY=\by]} - \dfrac12 \right| \leq \l^{-(\log\l)/9} .
\end{equation}
\end{cor}

Since $h(1/2 + x) \geq 1 - 4x^2$ for any $x\in [-1/2,1/2]$ (\cite[Theorem~1.2]{Topsoe}), we then derive for a typical $\by$:
\begin{equation}
\begin{aligned}
\E_g\big[ H^{(g)}(V_1|\bY=\by)\big] = \E_g\left[h\left(\dfrac{\P\nolimits^{(g)}[V_1=0,\bY=\by]}{\P\nolimits^{(g)}[\bY=\by]}\right)\right] &\geq  (1 - 4\l^{-(\log\l)/8})\cdot (1- 4\l^{-(\log\l)/9})\\
&\geq 1 - 8\l^{-(\log\l)/9}. 
\end{aligned}
\end{equation}
Then in \eqref{mult_entropy_formula} we have
\begin{equation}
\label{mult_bound_exp_entropy}
    \begin{aligned}
 \E_g\big[H^{(g)}(V_1|\bY) \big] &= \sum_{\by\in\Y^{\l}}\P[\bY=\by|\bV=\mathbi{0}]\E_g\big[ H^{(g)}(V_1|\bY=\by)\big]\\
 &\geq \sum_{\by \text{ typical}}\P[\bY=\by|\bV=\mathbi{0}]\E_g\big[ H^{(g)}(V_1|\bY=\by)\big]\\
 &\geq (1 - \l^{-\log\l})\cdot(1 - 8\l^{-(\log\l)/9}).\\
&\geq 1 - 9\l^{-(\log\l)/9} \geq  1 - \l^{-(\log\l)/10},
\end{aligned}
\end{equation}
where we used that the probability to get a typical output on a zero input is at least $1 -  \l^{-\log\l}$ by Lemma \ref{typical_lem}, as well as the condition $\log \l \geq 20$.

Finally, using the fact that $H^{(g)}(V_1|\bY) \leq 1$, Markov's inequality, and \eqref{mult_bound_exp_entropy}, we get
\begin{align*}
\P\limits_{g\sim G}\Big[H^{(g)}(V_1|\bY) \leq 1 - \l^{-\frac{\log\l}{20}} \Big]
 = \P\big[ 1 - H^{(g)}(V_1|\bY) 
 \geq \l^{-\frac{\log\l}{20}}  \big]  \leq \dfrac{ \E\big[1 - H^{(g)}(V_1|\bY)\big]}{\l^{-(\log\l)/20}} \leq \l^{-(\log\l)/20}.
\end{align*}
This completes the proof of Theorem~\ref{thm:converse_Shannon_BMS} for the case of BMS channel with bounded output alphabet size, asuming the typicality Lemma~\ref{typical_lem} and concentration Lemma~\ref{two_concentrations_lem} which we used in Lemma~\ref{multi_sum_lem}. We now proceed to proving these.

\subsubsection{Proof that the typical set is indeed typical}
\label{typical_proof_sect}
 \begin{proof}[Proof of Lemma \ref{typical_lem}]
 We start with proving that  \eqref{close_entropy_sum} is satisfied with high probability (over the randomness of the channel). Notice that $(d_1, d_2, \dots, d_m)$ are multinomially distributed by construction, since for each of the $\l$ bits transitioned, we choose independently the subchannel $W^{(i)}$ to use with probability $q_i$, for $i = 1, 2,\dots, m$, and $d_i$ represents the number of times the channel $W^{(i)}$ was chosen. So indeed $(d_1, d_2, \dots, d_m) \sim \text{Mult}(\l, q_1, q_2, \dots, q_m)$. The crucial property of multinomial random variables we are going to use is \textit{negative association} (\cite{NA_stat}, \cite{NA2}). The (simplified version of the) fact we are going to use about negatively associated random variables can be formulated as follows:
\begin{lem}[\cite{NA_stat}, Property P$_2$]\label{NA_lemma} Let $X_1, X_2,\dots, X_m$ be negatively associated random variables. Then, for every set of $m$ positive monotone non-decreasing functions $f_1,\dots, f_m$ it holds 
\begin{equation}
    \label{NA_property}
    \E\left[\prod_{i=1}^mf_i(X_i)\right] \leq \prod_{i=1}^m \E[f_i(X_i)].
\end{equation}
\end{lem}

We also use the fact that since $(d_1, d_2, \dots, d_m)$ are negatively associated, then applying decreasing functions $g_i(x) = \l q_i - x$ coordinate-wise to these random variables, we will also obtain negatively associated random variables (\cite{NA2}, Proposition 7). In other words, we argue that $(\l q_1 - d_1, \l q_2 - d_2, \dots, \l q_m - d_m)$ are also negatively associated, thus we can apply Lemma \ref{NA_lemma} to these random variables.
\begin{sloppypar}

Let us now denote for convenience $\alpha_i = h(p_i)$ for $i=1, 2, \dots, m$, and so we have $0 \leq \alpha_i \leq 1$. Let also $X = \sum_{i=1}^m(\l\cdot q_i - d_i)\alpha_i$, and we now can start with simple exponentiation and Markov's inequality: for any $a$ and any $t > 0$
\begin{equation}
\label{Chernoff-type1}
    \P[X \geq a] = \P[e^{tX} \geq e^{ta}] \leq e^{-ta}\E\left[e^{tX}\right] = e^{-ta}\E\left[\prod_{i=1}^me^{t\cdot\alpha_i(\l q_i - d_i)}\right] \leq e^{-ta}\prod_{i=1}^m\E\left[e^{t\cdot\alpha_i(\l q_i - d_i)}\right],
\end{equation}
where in the last inequality we applied Lemma \ref{NA_lemma} for negatively associated random variables ${(\l q_1 - d_1, \l q_2 - d_2, \dots, \l q_m - d_m)}$, as discussed above, and positive non-decreasing functions ${f_i(x) = e^{t\cdot\alpha_i\cdot x}}$, since $\alpha_i, t \geq 0$. 
\end{sloppypar}

Next, consider the following claim, which follows from standard Chernoff-type arguments:
\begin{claim}
\label{cl:Chernoff_binom}
Let $Z\sim \text{Binom}(n, p)$, and let $b > 0$. Then $\E[e^{-b\cdot Z}] \leq e^{np\cdot(e^{-b}-1)}$.
\end{claim}
\begin{proof} We can write $Z = \sum\limits_{j=1}^n Z_j$, where $Z_j \sim\text{Bern}(p)$ are independent Bernoulli random variables. Then
\begin{equation}
\label{eq:Chernoff_binom}
\begin{aligned}
    \E\left[e^{-b\cdot Z}\right] &= 
    \E\left[\prod_{j=1}^{n}e^{-b\cdot Z_j}\right]=
    \prod_{j=1}^{n}\E\left[e^{-b\cdot Z_j}\right] =  
    \left((1-p) + p\cdot e^{-b}\right)^n  \leq  e^{np(e^{-b} - 1)},
\end{aligned}
\end{equation}
where the only inequality uses the fact that $1 + x\leq e^x$ for any $x$.
\end{proof}

Turning back to~\eqref{Chernoff-type1}, we are going to bound the terms $\E\left[e^{t\cdot\alpha_i(\l q_i - d_i)}\right]$ individually. It is clear that the marginal distribution of $d_i$ is just $\text{Binom}(\l, q_i)$, so we are able to use Claim~\ref{cl:Chernoff_binom} for it. We derive:
\begin{equation}
\label{eq:Chernoff_binom2}\noeqref{eq:Chernoff_binom2}
\begin{aligned}
    \E\left[e^{t\cdot\alpha_i(\l q_i - d_i)}\right] = e^{t\alpha_i\l q_i} \cdot \E\left[e^{-t\alpha_i\cdot d_i}\right] \overset{\eqref{eq:Chernoff_binom}}{\leq} e^{t\cdot\alpha_i\l q_i}\cdot e^{\l q_i\left(e^{-t\alpha_i}-1\right)}= e^{\l q_i\left(t\alpha_i + e^{-t\alpha_i}-1\right)} \leq e^{\l q_i\left(t + e^{-t} - 1 \right)},
\end{aligned}
\end{equation}
where the last inequality uses that $x + e^{-x}$ is increasing for $x\geq 0$ together with $0 \leq t\alpha_i \leq t$, as $t > 0$ and $0\leq \alpha_i \leq 1$. Plugging the above into \eqref{Chernoff-type1} and using $\sum_{i=1}^m q_i = 1$, we obtain
\begin{equation}
\label{eq:Chernoff_binom3}
    \P[X\geq a] \leq e^{-ta}\prod_{i=1}^me^{\l q_i\left(t + e^{-t}-1\right)}= e^{-ta}\cdot e^{\l\left(t + e^{-t}-1\right)} \leq e^{-ta + \l\frac{t^2}{2}},
\end{equation}
where we used $x + e^{-x} - 1 \leq \frac{x^2}{2}$ for any $x\geq 0$. Finally, by taking $a = 2\sqrt{\l}\log\l$, setting $t = a/\l$, and recalling what we denoted by $X$ and $\alpha_i$ above, we immediately deduce
\begin{equation}
    \P\left[\sum_{i=1}^m(\l\cdot q_i - d_i)h(p_i) \geq 2\sqrt{\l}\log\l\right] \leq e^{-\frac{a^2}{2\l}} = e^{-2\log^2\l} \leq \l^{-2\log\l}.
\end{equation}
This means that the first typicality requirement~\eqref{close_entropy_sum} holds with very high probability (over the randomness of the channel).

\vspace{0.3cm}
Let us now prove that the second typicality condition~\eqref{close_coordinates_sum} holds with high probability. For that, we condition on the values of $d_1, d_2, \dots, d_m$. We will see that \eqref{close_coordinates_sum} holds with high probability for all values of $d_1, d_2, \dots, d_m$, and then it is clear that is will imply that it also holds with high probability overall.

So, fix the values of $d_1, d_2, \dots, d_m$. Denote a random variable $Y = \sum_{i=1}^m (p_id_i - t_i)\log\left(\frac{1-p_i}{p_i}\right)$, and then our goal it to show that $Y$ is bounded above by $O(\sqrt{\l}\log^2\l)$ with high probability (over the randomness of $t_i$'s). Given the conditioning on $d_1, d_2, \dots, d_m$, it is clear that $t_i \sim \text{Binom}(d_i, p_i)$ for all $i=1, 2,\dots,m$, and they are all independent (recall that $d_i$ corresponds to the number of times subchannel $W^{(i)}$ is chosen, while $t_i$ corresponds to the number of ``flips" within this subchannel).

We split the summation in $Y$ into two parts: let $T_1 = \{i\,:\, p_i \leq \frac1{\l} \}$ and $T_2 = [m] \setminus T_1$. Then for any realization of $t_i$'s, we have $\sum\limits_{i\in T_1}(p_id_i - t_i)\log\left(\frac{1-p_i}{p_i}\right) \leq \sum\limits_{i\in T_1}p_id_i\log\left(\frac{1}{p_i}\right) \leq \sum\limits_{i\in T_1}\frac{d_i\log\l}{\l} \leq \log\l$.

Denote the second part of the summation as $Y_2 = \sum_{i\in T_2} (p_id_i - t_i)\log\left(\frac{1-p_i}{p_i}\right)$. Notice that $\log\left(\frac{1-p_i}{p_i}\right) \leq \log\left(\frac{1}{p_i}\right) \leq \log\l$ for $i\in T_2$. Denote then $\gamma_i = \log\left(\frac{1-p_i}{p_i}\right)/\log\l$, and so $0\leq \gamma_i \leq 1$ for $i\in T_2$. Finally, let $\wt{Y_2} = Y_2/\log\l = \sum_{i\in T_2}(p_id_i - t_i)\cdot\gamma_i$. 

We now prove the concentration on $\wt{Y_2}$ in almost exactly the same way as we did for $X$ above. Similarly to~\eqref{Chernoff-type1} we obtain
\begin{equation}
    \label{eq:Chernoff_again0}
    \begin{aligned}
        \P\left[\wt{Y_2} >\hspace{-1.5pt} a\right] &= \P\left[e^{t\wt{Y_2}} > e^{ta} \right] \leq e^{-ta}\E\left[e^{t\wt{Y_2}}\right]  = e^{-ta}\cdot\E\left[\prod_{i=1}^m e^{t\cdot\gamma_i\cdot (p_id_i - t_i)}\right] = e^{-ta}\cdot\prod_{i=1}^m\E\left[ e^{t\cdot\gamma_i\cdot (p_id_i - t_i)}\right],
    \end{aligned}
\end{equation}
where the last equality holds because we conditioned on $d_1, d_2, \dots, d_m$, and so $t_1, t_2, \dots, t_m$ are independent, as discussed above. Next, Claim~\ref{cl:Chernoff_binom} applied for $t_i \sim\text{Binom}(d_i, p_i)$ and $t\cdot\gamma_i > 0$ for any $t > 0$ gives $\E\left[e^{-t\gamma_i\cdot t_i}\right] \leq e^{d_ip_i(e^{-t\gamma_i}-1)}$, and so similarly to~\eqref{eq:Chernoff_binom}--\eqref{eq:Chernoff_binom3} derive from~\eqref{eq:Chernoff_again0}
\begin{equation}
    \label{eq:Chernoff_again}
    \begin{aligned}
        \P\left[\wt{Y_2} >\hspace{-1.5pt} a\right] \leq e^{-ta}\cdot\hspace{-2.5pt}\prod_{i\in T_2}e^{p_id_i\left(t\gamma_i + e^{-t\gamma_i} -1\right)} \leq e^{-ta}\cdot\hspace{-2.5pt}\prod_{i\in T_2}e^{p_id_i\left(t + e^{-t} -1\right)} \leq e^{-ta + \sum_{i\in T_2}p_id_i\cdot t^2/2} \leq e^{-ta + \l t^2/2}
    \end{aligned}
\end{equation}
for any $t > 0$, where we used $0 \leq \gamma_i \leq 1$ for $i\in T_2$, $p_i < 1$, and $\sum_{i\in T_2}d_i \leq \l$. Therefore, by taking again $a = 2\sqrt{\l}\log\l$ and $t = a/\l$, obtain
\begin{align}
    \P\left[Y_2 \geq 2\sqrt{\l}\log^2\l\right] = \P\left[\wt{Y_2} \geq 2\sqrt{\l}\log\l\right] \leq \l^{-2\log\l}.
\end{align}
Since $Y \leq \log\l + Y_2$, we conclude that $Y \leq 3\sqrt{\l}\log^2\l$ with probability at least $1 - \l^{-2\log\l}$ over the randomness of the channel. 

\vspace{0.3cm}
Since both~\eqref{close_entropy_sum} and~\eqref{close_coordinates_sum} hold with probability at least $1 - \l^{-2\log\l}$, the union bound implies that these two conditions hold simultaneously with probability at least $1 - 2\l^{-2\log\l} \geq 1 - \l^{-\log\l}$.
\end{proof}

\subsubsection{Concentration Lemma}
\label{sec:concentration}
\begin{lem}
\label{two_concentrations_lem}
Let $\chi \sim\O = \text{Binom}(d_1, p_1)\times\text{Binom}(d_2, p_2)\times\dots\times\text{Binom}(d_m, p_m)$, where $d_i$'s are nonnegative integers for $i\in[m]$, $p_i \leq 1/2$, $\sum_{i=1}^md_i = \l$, and $m \leq \sqrt{\l}$. Let also $\l$ be large so that $\log\l\geq 8$.
Then the following holds with probability at least $1 - \l^{-(\log\l)/4}$:
    \begin{align}
         \sum_{i=1}^md_i\left(h(p_i) - h\left(\frac{\chi_i}{d_i}\right)\right) &\leq 7\l^{1/2}\log^3\l. \label{conc1}
    \end{align}
 \end{lem}
 \begin{proof}
First, notice that we can disregard all the indices $i$ for which $d_i = 0$, as they do not contribute anything to the LHS of~\eqref{conc1}. So from now on, we assume for simplicity that $d_i \geq 1$ for all $i = 1, 2, \dots, m$.

Next, we split the interval $[1:m]$ into two parts. In the first part the value of $d_i\cdot p_i$ is going to be small, and the sum of $d_ih(p_i)$ will also be small. For the second part, when $d_i\cdot p_i$ is large enough, we will be able to apply some concentration arguments. Denote:
 \begin{equation}
 \label{split_intervals_lem}
     \begin{aligned}
          F_1 &\coloneqq \left\{i\,:\ p_i \leq \frac{4\log^2\l}{d_i}\right\},\\
          F_2 &\coloneqq \{1, 2,\dots,m\} \setminus F_1.
     \end{aligned}
 \end{equation}
 Then
 \begin{align}
 \sum_{i=1}^md_i\left(h(p_i) - h\left(\frac{\chi_i}{d_i}\right)\right) &\leq \sum_{i\in F_1}d_ih(p_i) + \sum_{i\in F_2}d_i\left(h(p_i) - h\left(\frac{\chi_i}{d_i}\right)\right). \label{split_sum1}
 \end{align}

Let us deal with the summation over $F_1$ first. Split this set even further: $F_1^{(1)} = \{i\in F_1 \,:\, d_i \geq 8\log^2\l\}$, and $F_1^{(2)} = F_1\setminus F_1^{(1)}$. Then for any $i \in F_1^{(1)}$ we use $h(p_i) \leq 2p_i\log\frac{1}{p_i}$ from Proposition~\ref{prop:entropy_half}, since $p_i \leq 1/2$. For any $i\in F_1^{(2)}$ we just use $h(p_i) \leq 1$. Combining these, obtain
\begin{align}
    \sum_{i\in F_1}d_ih(p_i) \leq \sum_{i\in F_1^{(1)}} 2d_ip_i\log\dfrac{1}{p_i} + \sum_{i\in F_1^{(2)}}d_i &\leq \sum_{i\in F_1^{(1)}} 8\log^2\l\cdot\log\left(\dfrac{d_i}{4\log^2\l}\right)  + \left\lvert F_1^{(2)}\right\lvert\cdot 8\log^2\l\\
    &\leq  8\log^2\l \cdot \sum_{i\in F_1^{(1)}} \log d_i  +  \left\lvert F_1^{(2)}\right\lvert\cdot 8\log^2\l.
    \label{split2_part2}
\end{align}
For the second summand in the RHS above, we will just use $\left\lvert F_1^{(2)}\right\lvert \leq m \leq \l^{1/2}$. For the first summand, we use Jensen's inequality, the fact that $\sum\limits_{i=1}^m d_i = \l$, and $\left\lvert F_1^{(1)}\right\lvert \leq m \leq \l^{1/2}$ to derive
\[ \sum_{i\in F_1^{(1)}} \log d_i \leq \left\lvert F_1^{(1)}\right\lvert\cdot\log\left(\dfrac{\sum_{i\in F_1^{(1)}} d_i}{\left\lvert F_1^{(1)}\right\lvert}\right) \leq \left\lvert F_1^{(1)}\right\lvert\cdot\log\left(\dfrac{\l}{\left\lvert F_1^{(1)}\right\lvert}\right) \leq \l^{1/2}\log\left(\l^{1/2}\right) = \frac12\l^{1/2}\log\l,\]
where the last inequality uses that $x\log(\l/x)$ is increasing while $x \leq \l/e$. Therefore, in~\eqref{split2_part2} obtain
\begin{align}
    \sum_{i\in F_1}d_ih(p_i) \leq  8\log^2\l \cdot \sum_{i\in F_1^{(1)}} \log d_i  +  \left\lvert F_1^{(2)}\right\lvert\cdot 8\log^2\l \leq 5\l^{1/2}\log^3\l,
    \label{split2_part3}
\end{align}
where we also used $8 \leq \log\l$.
 
Therefore, the first part of the RHS of~\eqref{split_sum1} is always bounded by $5\l^{1/2}\log^3\l$. We will now deal with the remaining summations over $i\in F_2$.

For any $i\in F_2$, we know that $d_ip_i\geq4\log^2\l$. Now we apply the multiplicative Chernoff bound~\eqref{new_eq:multiplicative_Chernoff} for $\chi_i\sim\text{Binom}(d_i, p_i)$ and 
$\delta = \frac{\log\l}{\sqrt{d_ip_i}}$ to get
\begin{equation}
\label{indiv_Chernoff}
    \P_{\chi_i}\big[|\chi_i - d_ip_i| \geq \sqrt{d_ip_i}\log\l\big] \leq 2e^{-\log^2 \l/3} 
    \le \l^{-(\log\l)/3} \qquad \text{if}\ \log\l \leq \sqrt{d_ip_i},
\end{equation}
 where the last inequality holds for $\log\l > 3$ because the $\log$ in the exponent is to base $2$.
 The condition ${\log\l \leq \sqrt{d_ip_i}}$ is needed in order to have $\delta \leq 1$ for the multiplicative Chernoff bound~\eqref{new_eq:multiplicative_Chernoff} to hold.
 
 Then, by the union bound, we derive
 \begin{equation}
 \label{mult_Chernoff}
     \P_{\chi\sim\Omega}\left[|\chi_i - d_ip_i| \geq \sqrt{d_ip_i}\log\l \text{ for some } i \in F_2 \right] \leq |F_2|\cdot\l^{-(\log\l)/3}  \leq \l^{-(\log\l)/3 + 1/2}.
 \end{equation}
 
 Define the sets $\T_1^{(i)}$ for all $i=1,2,\dots,m$ as follows:
 \begin{equation}
 \label{chernoff_intervals}
 \begin{aligned}
     &\T_1^{(i)}\coloneqq \left\{ s_i \in [0:d_i]\, :\ |s_i - d_ip_i| \leq \sqrt{d_ip_i}\log\l\right\}, \quad &\text{for } i \in F_2;\\
     &\T_1^{(i)}\coloneqq [0:d_i], \quad &\text{for } i \notin F_2.
 \end{aligned}
 \end{equation}
 and let 
 \begin{equation}
 \label{theta_def}
      \theta_i \coloneqq \P[\chi_i\in\T_1^{(i)}].
 \end{equation}
 Then by \eqref{indiv_Chernoff} we have
\begin{equation}
\label{chernoff_probs}
    \begin{aligned}
         &\theta_i  \geq 1 -\l^{-(\log\l)/3},\qquad &\text{for } i \in F_2;\\
     &\theta_i = 1, \qquad &\text{for } i \notin F_2.
    \end{aligned}
\end{equation}

 Finally, define
 \begin{equation}
 \label{mult_chernoff_prob}
     \theta \coloneqq \prod_{i=1}^m\theta_i = \prod_{i\in F_2}\theta_i = \prod_{i\in F_2}\P[\chi_i\in\T_1^{(i)}] = \P_{\chi\sim\O}[\chi_i \in \T_1^{(i)} \text{ for all } i \in F_2] \geq 1 - \l^{-(\log\l)/3 + 1/2},
 \end{equation}
 where the last inequality is a direct implication of \eqref{mult_Chernoff}.
 
We will now define a set of new probability distributions $\D_i$ for all $i=1,2,\dots,m$, as binomial distributions $\text{Binom}(d_i, p_i)$ restricted to intervals $\T_1^{(i)}$. Formally, let us write
 \begin{equation}
 \label{truncated_distr}
 \begin{aligned}
      \P_{\eta_i\sim\D_i}\big[\eta_i=x\big] = \begin{cases} 0, \quad &\text{if } x \notin \T_1^{(i)};\\
      \P_{\chi_i\sim\text{Binom}(d_i, p_i)}\big[\chi_i=x\big]\cdot\theta_i^{-1}, \quad &\text{if } x \in \T_1^{(i)}.
      \end{cases}
 \end{aligned}
 \end{equation}
 (So to get $\D_i$ we just took a distribution Binom$(d_i, p_i)$, truncated it so it does not have any mass outside of $\T_1^{(i)}$, and rescaled appropriately.)
 
 Next, define a product distribution $\D \coloneqq \bigtimes_{i=1}^m\D_i$ on the set $\T_1 \coloneqq \bigtimes_{i=1}^m\T_1^{(i)}$. Notice now that it is trivial that for any subset $\Ra\subseteq \T_1$ it holds
 \begin{equation}
 \label{distributions_almost_equal}
      \P_{\chi\sim\O}[\chi\in\Ra] = \P_{\eta\sim\D}[\eta\in\Ra] \cdot\theta.
 \end{equation}
 Since $\theta$ is very close to $1$, it suffices to prove the claims for $\D$ instead of $\O$.
 
\vspace{0.25cm}

Recall that our goal was to show that $\sum_{i\in F_2}d_i\left(h(p_i) - h\left(\frac{\chi_i}{d_i}\right)\right)$ (the second part from \eqref{split_sum1}) is bounded above by $O(\l^{1/2}\log^3\l)$ with high probability, when $\chi\sim\O$. Instead now let us show that this summation is small with high probability when $\chi\sim\D$, and then use the arguments above to see that there is not much of a difference when $\chi\sim\O$.

\begin{claim}
\label{entropy_distortion_claim}
Let $i\in F_2$ and $\chi_i \sim \D_i$. Then 
\begin{align}
\left|d_i\left(h(p_i) - h\left(\frac{\chi_i}{d_i}\right)\right)\right| &\leq \sqrt{d_ip_i}\log^2\l . \label{claim_ent1}
\end{align}

\end{claim}
\begin{proof} First,  $\left|\frac{\chi_i}{d_i} - p_i\right| \leq \sqrt{\frac{p_i}{d_i}}\log\l$ for $\chi_i\sim\D_i$ by definition of the distribution $\D_i$. Now, for $i\in F_2$, $p_i \geq \frac{4\log^2\l}{d_i}$, from which it follows that $\frac{p_i}2 \geq \sqrt{\frac{p_i}{d_i}}\log\l$, and therefore $\frac{p_i}{2} \leq \frac{\chi_i}{d_i} \leq \frac{3p_i}{2}$. We then use the concavity of the binary entropy function on $[0, 1]$. For a concave differentiable function $f$ on an interval $[a, b]$, one has $|f(b) - f(a)| \leq |b-a|\cdot \max\left\{|f'(a)|, |f'(b)|\right\}$, which follows from a standard inequality $f(y) \leq f(x)  + f'(x)(y - x)$ applied for $(a, b)$ or $(b, a)$, depending on which of $f(a)$ and $f(b)$ is larger. We apply this for the binary entropy function $h(\cdot)$ and one of the intervals $\left[\frac{\chi_i}{d_i}, p_i\right]$ and $\left[p_i, \frac{\chi_i}{d_i}\right]$, depending on which of $\frac{\chi_i}{d_i}$ and $p_i$ is smaller:
\begin{equation*}
     \left|h\left(\frac{\chi_i}{d_i}\right) - h(p_i)\right| \leq \left|\dfrac{\chi_i}{d_i} - p_i\right|\cdot \text{max}\left\{\left\lvert\dfrac{dh}{dx}(p_i)\right\lvert, \left\lvert\dfrac{dh}{dx}\left(\frac{\chi_i}{d_i}\right)\right\lvert \right\}.
\end{equation*}
Now, both $p_i$ and $\frac{\chi_i}{d_i}$ lie in the interval $\left[\frac{p_i}{2}, \frac{3p_i}{2}\right]$, which is contained in $\left[\frac{p_i}{2}, 1-\frac{p_i}{2}\right]$, as $p_i < 1/2$. Out of symmetry of $h$ around $1/2$, it follows that the maximal value of $\left\lvert\frac{dh}{dx}(\cdot)\right\lvert$ on the interval $\left[\frac{p_i}{2}, 1 -\frac{p_i}{2}\right]$ is attained at $\frac{p_i}{2}$. Therefore, we have

\begin{equation*}
    \begin{aligned}
 \left|h\left(\frac{\chi_i}{d_i}\right) - h(p_i)\right| &\leq \left|\dfrac{\chi_i}{d_i} - p_i\right|\cdot \text{max}\left\{\left\lvert\dfrac{dh}{dx}(p_i)\right\lvert, \left\lvert\dfrac{dh}{dx}\left(\frac{\chi_i}{d_i}\right)\right\lvert \right\} \\
 &\leq \sqrt{\frac{p_i}{d_i}}\log\l\cdot \left\lvert\dfrac{dh}{dx}\left(\frac{p_i}2\right)\right\lvert = \sqrt{\frac{p_i}{d_i}}\log \l \cdot \log\dfrac{1-p_i/2}{p_i/2}\\
 &\leq \sqrt{\frac{p_i}{d_i}}\log \l \cdot \log\frac{2}{p_i} \leq \sqrt{\frac{p_i}{d_i}}\log\l\cdot\log\left(\dfrac{d_i}{2\log^2\l}\right) \leq \sqrt{\frac{p_i}{d_i}}\log^2 \l,
 \end{aligned}
 \end{equation*}
where the penultimate inequality follows from $p_i \geq \frac{4\log^2\l}{d_i}$ for $i \in F_2$, and the last inequality uses $\frac{d_i}{2\log^2\l} \leq \ell$, as $\sum_{i=1}^md_i = \ell$ and $d_i$'s are nonnegative. Therefore, \eqref{claim_ent1} follows. 
\end{proof}
\vspace{0.3cm}

Let $\chi\sim \D$ here and further. Define for convenience new random variables $X_i = d_i\left(h(p_i) - h\left(\frac{\chi_i}{d_i}\right)\right)$ for all $i \in F_2$, and let also $X = \sum_{i\in F_2}X_i = \sum_{i\in F_2}d_i\left(h(p_i) - h\left(\frac{\chi_i}{d_i}\right)\right)$.

\begin{claim}
\label{sum_Hoeffding}
With probability at least $1 - \l^{-\log\l}$ it holds that
 \[ X - \E[X] \leq \l^{1/2}\log^3\l \]
\end{claim}
\begin{proof}
Obviously all the $X_i$'s are independent, and also $X_i \in \left[- \sqrt{d_ip_i}\log^2\l, \sqrt{d_ip_i}\log^2\l \right]$ by Claim~\ref{entropy_distortion_claim}. Then we can apply Hoeffding's inequality for the sum of bounded independent random variables (\cite[Theorem 2]{Hoeffding1963}), and obtain
\begin{equation}
\begin{aligned}
     \P_{\chi\sim\D}\Big[X - \E[X] \geq \l^{1/2}\log^3\l\Big] &\leq \text{exp}\left(-\dfrac{2\l\log^6\l}{\sum\limits_{i\in F_2}(2\sqrt{d_ip_i}\log^2\l)^2}\right)  \\
     &=  \text{exp}\left(-\dfrac{2\l\log^6\l}{\log^4\l\cdot\sum\limits_{i\in F_2}(4d_ip_i)}\right) \leq  \text{exp}\left(-\dfrac{\l\log^2\l}{\sum\limits_{i\in F_2}d_i}\right)\\
     &\leq e^{-\log^2\l} \leq \l^{-\log\l},
     \end{aligned}
\end{equation}
 where we used $p_i\leq 1/2$ and $\sum\limits_{i\in F_2}d_i \leq \sum\limits_{i=1}^md_i = \l$ in the second and third inequalities, respectively.
\end{proof}

 So by now we proved that $X = \sum_{i\in F_2}d_i\left(h(p_i) - h\left(\frac{\chi_i}{d_i}\right)\right)$ does not deviate much from its expectation. What we are left to show now is that $\E[X]$ is not very large by itself. 
 
 The following two claims show that the first moment and mean absolute deviation of the distribution $\D_i$ are close to those of $\O_i$. This easily follows from the definition~\eqref{truncated_distr} of $\D_i$, and the proofs are deferred to Appendix~\ref{app:moments}

 \begin{claim}
 \label{cl:exp_of_Di}
 Let $i\in F_2$. Then $\left|\E\limits_{\chi_i \sim \D_i}\left[\frac{\chi_i}{d_i}\right] - p_i\right| \leq \frac1{d_i}$.
 \end{claim}
\begin{claim} 
\label{cl:mean_absolute}
Let $\chi_i \sim \D_i$ and $\eta_i \sim \O_i$ for $i\in F_2$. Then $\E\Big\lvert\chi_i - \E[\chi_i]\Big\lvert \leq \E\Big\lvert\eta_i - \E\left[\eta_i\right]\Big\lvert + 1$.
\end{claim}

These observations allow us we prove the following
     
\begin{claim}
\label{cl:hE-Eh}
 Let $i\in F_2$, and $\chi_i \sim \D_i$. Then $h\left(\E\left[\frac{\chi_i}{d_i}\right]\right) - \E\left[h\left(\frac{\chi_i}{d_i}\right)\right] \leq \frac{5\log\l}{d_i}$.
 \end{claim}
 \begin{proof}
  Unfortunately, Jensen's inequality works in the opposite direction for us here. However, we use some form of converse Jensen's from \cite{Dragomir}, which says the following:
 \begin{lem}[Converse Jensen's inequality, \cite{Dragomir}, Corollary 1.8]\label{converse_Jensen} Let $f$ be a concave differentiable function on an interval $[a, b]$, and let $Z$ be a (discrete) random variable, taking values in $[a, b]$. Then
 \[ 0 \leq f(\E[Z]) - \E[f(Z)] \leq \frac12\left(f'(a) - f'(b)\right)\cdot\E\left|Z - \E[Z]\right|.   \]
 \end{lem}
 
 We apply it here for the concave binary entropy function $h$, and  random variable $Z = \frac{\chi_i}{d_i}$ for $\chi_i\sim\D_i$, which takes values in $[a,b] \coloneqq \left[p_i - \sqrt{\frac{p_i}{d_i}}\log\l, p_i + \sqrt{\frac{p_i}{d_i}}\log\l \right]$. Recall also that for $i\in F_2$, $p_i \geq \frac{4\log^2\l}{d_i}$ and then $\frac{p_i}2 \geq \sqrt{\frac{p_i}{d_i}}\log\l$, therefore $a = p_i - \sqrt{\frac{p_i}{d_i}}\log\l \geq \frac{p_i}{2}$, and also ${b = p_i + \sqrt{\frac{p_i}{d_i}}\log\l \leq \frac{3p_i}{2}}$. Using the mean value theorem, for some $c\in [a,b] \subseteq \left[\frac{p_i}{2}, \frac{3p_i}{2}\right]$ we have 
 \[ h'(a) - h'(b) = (b-a)\cdot(-h''(c)) \leq 2\sqrt{\frac{p_i}{d_i}}\log\l\cdot (-h''(c)).\] 
 Now we look at $(-h''(c)) = \frac{1}{c(1-c)\ln 2}$ for some $c \in \left[\frac{p_i}{2}, \frac{3p_i}{2}\right]$. As $p_i < 1/2$, it follows $\left[\frac{p_i}{2}, \frac{3p_i}{2}\right] \subseteq \left[\frac{p_i}{2}, 1-\frac{p_i}{2}\right]$. Using the symmetry of a function $x(1-x)$ around $1/2$, we conclude that its minimal value over the interval $\left[\frac{p_i}{2}, \frac{3p_i}{2}\right]$ is attained at $p_i/2$. Thus derive $c(1-c) \geq \frac{p_i}{2}\left(1 - \frac{p_i}{2}\right) \geq \frac{3p_i}{8}$, since $p_i < 1/2$. And so $(-h''(c)) = \frac{1}{c(1-c)\ln 2} \leq \frac{8}{p_i\cdot 3\ln 2} \leq \frac{4}{p_i}$. Therefore
 \[ h'(a) - h'(b) \leq \dfrac{8\log\l}{\sqrt{d_ip_i}}.  \]
 Finally, Claim~\ref{cl:mean_absolute} gives $\E\left|Z - \E[Z]\right| \leq \E\left|\frac{Z_2}{d_i} - \E\left[\frac{Z_2}{d_i}\right]\right| + \frac1{d_i}$ for $Z_2\sim\text{Binom}(d_i, p_i)$, and so
 \[ \E\left|Z - \E[Z]\right| \leq \frac1{d_i}\E\left|Z_2 - \E[Z_2]\right|+\frac1{d_i} \leq \frac1{d_i}\sqrt{\E[(Z_2 - \E[Z_2])^2]}+\frac1{d_i} = \sqrt{\frac{p_i(1-p_i)}{d_i}}+\frac1{d_i} \leq \sqrt{\frac{p_i}{d_i}} + \frac1{d_i}.   \] 
 Putting all this together, Lemma \ref{converse_Jensen} gives us
 \[ 0 \leq h\left(\E\left[\frac{\chi_i}{d_i}\right]\right) - \E\left[h\left(\frac{\chi_i}{d_i}\right)\right] \leq \dfrac12\cdot \dfrac{8\log\l}{\sqrt{d_ip_i}} \cdot \left(\sqrt{\frac{p_i}{d_i}} + \frac1{d_i}\right) = \dfrac{4\log\l}{d_i} + \dfrac{4\log\l}{d_i\sqrt{d_ip_i}} \leq \dfrac{5\log\l}{d_i},  \]
where the last step uses $\sqrt{p_id_i}\geq 2\log\l$ for $i \in F_2$.
 \end{proof}
 
 We can now use the above claims and Proposition~\ref{prop:entropy_differ} to bound the expectation of $X$:
\begin{equation}
 \begin{aligned}
 \label{eq:X_exp}
      \E[X] = \sum_{i\in F_2} d_i\left(h(p_i) - \E\left[h\left(\frac{\chi_i}{d_i}\right)\right]\right) &\leq  \sum_{i\in F_2} d_i\left(h(p_i) - h\left(\E\left[\frac{\chi_i}{d_i}\right]\right) + \dfrac{5\log\l}{d_i}\right) \\
      &\leq  \sum_{i\in F_2} d_i\left(h\left(\left\lvert p_i - \E\left[\frac{\chi_i}{d_i}\right]\right\lvert\right) + \dfrac{5\log\l}{d_i}\right) \\
      &\leq \sum_{i\in F_2}d_i\left(h\left(\frac1{d_i}\right) + \frac{5\log\l}{d_i}\right) \\
      &\leq \sum_{i\in F_2}d_i\left(\frac{2}{d_i}\log d_i + \frac{5\log\l}{d_i}\right) \leq 7\l^{1/2}\log\l \leq \l^{1/2}\log^3\l,
 \end{aligned}
\end{equation}
where the first inequality is from Claim~\ref{cl:hE-Eh}, the second is by Proposition~\ref{prop:entropy_differ}, the third one follows from Claim~\ref{cl:exp_of_Di}, the fourth inequality is from Proposition~\ref{prop:entropy_half}, and the next ones follow from $d_i \leq \l$, $|F_2| \leq m \leq \l^{1/2}$, and $\log \l > 8$ by the conditions for this Lemma~\ref{two_concentrations_lem}.

 So we showed in Claim~\ref{sum_Hoeffding} that $X$ does not exceed its expectations by more than $\l^{1/2}\log^3\l$ with high probability (over $\chi\sim\D$), and also that $E[X]$ is bounded by $\l^{1/2}\log^3\l$ in~\eqref{eq:X_exp}, and therefore $X$ does not exceed $2\l^{1/2}\log^3\l$ with high probability. Specifically, it means that there exists $\T\subseteq \T_1$, such that $\P_{\chi\sim\D}[\chi\in\T] \geq 1 - \l^{-\log\l}$, and that for any $\bs\in\T$ it holds $\sum\limits_{i\in F_2}d_i\left(h(p_i) - h\left(\frac{s_i}{d_i}\right)\right) \leq 2\l^{1/2}\log^3\l$. Recall that $\sum\limits_{i\in F_1}d_i h(p_i) \leq 5\l^{1/2}\log^3\l$ as we showed in~\eqref{split2_part3}. Thus, by summing these two inequalities, we conclude from~\eqref{split_sum1} that $\sum_{i=1}^{m} d_i\left(h(p_i) - h\left(\frac{s_i}{d_i}\right)\right) \leq 7\l^{1/2}\log^3\l$ for any $\bs\in\T$. 
 
 Finally, the last step is to return back from the product of ``truncated binomials" $\D$ to the original product of binomials $\O$. As we defined the set $\T$ above, we have $\P_{\chi\sim\D}[\chi\in\T] \geq 1 - \l^{-\log\l}$. But by \eqref{distributions_almost_equal} the distributions $\O$ and $\D$ are very close to each other, and therefore we obtain:
 \begin{equation}
\label{prob_of_tau_is_big}
 \P_{\chi\sim\O}[\chi\in\T] = \P_{\chi\sim\D}[\chi\in\T]\cdot\theta \geq \left(1 - \l^{-\log\l}\right)\left(1 - \l^{-(\log\l)/3 + 1/2}\right) \geq 1 - \l^{-(\log\l)/4},
 \end{equation}
where we used the bound~\eqref{mult_chernoff_prob} on $\theta$ for the first inequality and $\log\l \geq 8$ for the second one.
\end{proof}

\subsection{Arbitrary alphabet size}
\label{sec:BMS_any_alphabet}
In this section we finish the proof of Theorem~\ref{thm:converse_Shannon_BMS} for the general BMS channel using the results from the previous section. 

For BMS channels with large output alphabet size we will use binning of the output, however we will do it in a way that \textit{upgrades} the channel, rather then degrades it (recall Definition~\ref{def:degrad}).
Specifically, we will employ the following statement:

\begin{prop}
\label{prop:binned_upgraded_channel}
Let $W$ be any BMS channel. Then there exists another BMS channel~$\wW$ with the following properties:
\begin{enumerate}[label=(\roman*)]
    \item Output alphabet size of $\wW$ is at most $2\sqrt{\l}$;
    \item $\wW$ is \textit{upgraded} with respect to $W$, i.e. $W \preceq \wW$;
    \item $H(\wW) \geq H(W) - \dfrac{\log\l}{\l^{1/2}}$.
\end{enumerate}
\end{prop}

Before proving this proposition, we first show how we can finish a proof of Theorem~\ref{thm:converse_Shannon_BMS} using it. So, consider any BMS channel $W$ with output alphabet size larger than $2\sqrt{\l}$, and consider the channel $\wW$ which satisfies properties (i)-(iii) from Proposition~\ref{prop:binned_upgraded_channel} with respect to $W$. 
First of all, notice that $k \geq \l(1 - H(W)) + 14\l^{1/2}\log^3\l \geq \l\left(1 - H(\wW) - \frac{\log\l}{\l^{1/2}}\right) + 14\l^{1/2}\log^3\l$, and thus ${k \geq \l(1 - H(\wW)) + 13\l^{1/2}\log^3\l}$. Taking the property (i) into consideration, it follows that the channel $\wW$ satisfies all the conditions for the arguments in the Section \ref{sec:BMS_large_alphabet} to be applied, i.e. the statement of Theorem~\ref{thm:converse_Shannon_BMS} holds for $\wW$. Therefore, we can argue that with probability at least $1 - \l^{-(\log\l)/20}$ over a random kernel $G$ it holds $H(V_1\,|\,\wt{\bY}) \geq 1 - \l^{-(\log\l)/20}$, where $\wt{\bY} = \wW^{\l}(\bV\cdot G)$ is the output vector if one would use the channel $\wW$ instead of $W$, for $\bV \sim \{0,1\}^k$. 

Now, let $W_1$ be the channel which ``proves" that $\wW$ is upgraded with respect to $W$, i.e. $W_1\left(\wW(x)\right)$ and $W(x)$ are identically distributed for any $x\in\bit$. Trivially then, $W_1^{\l}\left(\wW^{\l}(X)\right)$ and $W^{\l}(X)$ are identically distributed for any random variable $X$ supported on $\bit^{\l}$. 

Next, observe that the following forms a Markov chain
\[ V_1 \to \bV \to \bV\cdot G \to \wW^{\l}(\bV G) \to W_1^{\l}\left(\wW^{\l}(\bV G)\right),   \]
where $\bV$ is distributed uniformly over $\bit^{k}$. But then the data-processing inequality gives
\[ I\left(V_1\,;\,W_1^{\l}\left(\wW^{\l}(\bV G)\right) \right) \leq I\left(V_1\,;\,\wW^{\l}(\bV G)\right).  \]
However, as we discussed above, $W_1^{\l}\left(\wW^{\l}(\bV G)\right)$ and $W^{\l}(\bV G)$ are identically distributed, and so 
\[ I(V_1\,;\, \bY) = I\left(V_1\,;\,W^{\l}(\bV G) \right) =   I\left(V_1\,;\,W_1^{\l}\left(\wW^{\l}(\bV G)\right) \right) \leq  I\left(V_1\,;\,\wW^{\l}(\bV G)\right) = I(V_1\,;\,\wt{\bY}).  \]
Therefore using $H(X|Y) = H(X) - I(X;Y)$ we derive that
\[ H(V_1\,|\,\bY) \geq H(V_1\,|\,\wt{\bY}) \geq 1 - \l^{-(\log\l)/20}  \]
with probability at least $1 - \l^{-(\log\l)/20}$. This concludes the proof of Theorem~\ref{thm:converse_Shannon_BMS}.
\end{proof}

\vspace{0.8cm}

\begin{proof}[Proof of Proposition \ref{prop:binned_upgraded_channel}] We are going to describe how to construct such an upgraded channel $\wW$. We again are going to look at $W$ as a convex combination of BSCs, as we discussed in Section \ref{sec:BMS_large_alphabet}: let $W$ consist of $m$ underlying BSC subchannels $W^{(1)}, W^{(2)}\dots, W^{(m)}$, each has probability $q_j$ to be chosen. The subchannel $W^{(j)}$ has crossover probability $p_j$, and $0\leq p_1 \leq\dots\leq p_m \leq \frac12$. The subchannel $W^{(j)}$ can output $z^{(0)}_j$ or $z^{(1)}_j$, and the whole output alphabet is then $\Y = \{z^{(0)}_1, z^{(1)}_1, z^{(0)}_2, z^{(1)}_2, \dots, z^{(0)}_m, z^{(1)}_m\}$, $|\Y| = 2m$. It will be convenient to write the transmission probabilities of $W$ explicitly: for any $k \in [m]$, $c, x \in \bit$:
\begin{align}
\label{eq:W_def}
W\left(z^{(c)}_k\;\Big\lvert\;x\right) =  \begin{cases} q_k\cdot (1-p_k), \qquad &x = c,\\
q_k\cdot p_k, \qquad &x \not= c.\\
\end{cases}
\end{align}
The key ideas behind the construction of $\wW$ are the following:
\begin{itemize}[label={--}]
    \item decreasing a crossover probability in any BSC (sub)channel always upgrades the channel, i.e. $\text{BSC}_{p_1} \preceq \text{BSC}_{p_2}$ for any $0 \leq p_2 \leq p_1 \leq \frac12$ (\cite[Lemma 9]{Tal_Vardy}). Indeed, one can simulate a flip of coin with bias $p_1$ by first flipping a coin with bias $p_2$, and then flipping the result one more time with probability $q = \frac{p_1-p_2}{1 - 2p_2}$. In other words, $\text{BSC}_{p_1}(x)$ and $\text{BSC}_{q}\left(\text{BSC}_{p_2}(x)\right)$ are identically distributed for $x\in\bit$.
    \item ``binning" two BSC subchannels with the same crossover probability doesn't change the channel
(\cite[Corollary 10]{Tal_Vardy}).
\end{itemize}

Let us finally describe how to construct $\wW$. Split the interval $[0, 1/2]$ into $\sqrt{\l}$ parts evenly, i.e. let ${\theta_j = \frac{j-1}{2\sqrt{\l}}}$ for $j = 1, 2, \dots, \sqrt{\l}+1$, and consider intevals $[\theta_j, \theta_{j+1})$ for $j = 1, 2, \dots, \sqrt{\l}$ (include $1/2$ into the last interval). Now, to get $\wW$, we first slightly decrease the crossover probabilities in all the BSC subchannels $W^{(1)}, W^{(2)}\dots, W^{(m)}$ so that they all become one of $\theta_1, \theta_2, \dots, \theta_{\sqrt{\l}}$. After that we bin together the subchannels with the same crossover probabilities and let the resulting channel be $\wW$. Formally, we define
\begin{align*}
     T_j &\coloneqq \bigg\{ i \in [m]\ :\ p_i \in \big[\theta_j, \theta_{j+1}\big) \bigg\}, \quad\qquad j = 1, 2, \dots, \sqrt{\l}-1, \\
     T_{\sqrt{\l}} &\coloneqq \bigg\{ i \in [m]\ :\ p_i \in \big[\theta_{\sqrt{\l}}, \theta_{\sqrt{\l}+1}\big] \bigg\}.
\end{align*}

So, $T_j$ is going to be the set of indices of subchannels of $W$ for which we decrease the crossover probability to be equal to $\theta_j$. Then the probability distribution over the new, binned, BSC subchannels $\wt{W^{(1)}}, \wt{W^{(2)}}\dots, \wt{W^{(\sqrt{\l})}}$ in the channel $\wW$ is going to be $(\wt{q_1}, \wt{q_2}, \dots, \wt{q_{\sqrt{\l}}})$, where  $\widetilde{q_j} \coloneqq \sum\limits_{i\in T_j} q_i$. The subchannel $\wt{W^{(j)}}$ has crossover probability $\theta_j$, and it can output one of two new symbols $\wt{z^{(0)}_j}$ or $\wt{z^{(1)}_j}$. The whole output alphabet is then $\wt{\Y} = \{\wt{z^{(0)}_1}, \wt{z^{(1)}_1}, \wt{z^{(0)}_2}, \wt{z^{(1)}_2}, \dots, \wt{z^{(0)}_{\sqrt{\l}}}, \wt{z^{(1)}_{\sqrt{\l}}}\}$, $|\wt{\Y}| = 2\sqrt{\l}$. To be more specific, we describe $\wW\, :\,\bit\to\wt{\Y}$, as follows: for any $j\in[\sqrt{\l}]$ and any $b, x\in \bit$
\begin{align}
\label{eq:wW_def}
\wW\left(\wt{z^{(b)}_j}\;\Big\lvert\;x\right) =  \begin{cases} \sum\limits_{i\in T_j}q_i\cdot (1-\theta_j), \qquad &x = b,\\
\sum\limits_{i\in T_j}q_i\cdot \theta_j, \qquad &x \not= b.\\
\end{cases}
\end{align}

Property (i) on the output alphabet size for $\wW$ then holds immediately. Let us verify (ii) by showing that $\wW$ is indeed upgraded with respect to $W$.

One can imitate the usage of $W$ using $\wW$ as follows: on input $x\in\bit$, feed it through $\wW$ to get output $\wt{z_j^{(b)}}$ for some $b\in\bit$ and $j \in [\sqrt{\l}]$. We then know that the subchannel $\wt{W^{(j)}}$ was used, which by construction corresponds to the usage of a subchannel $W^{(i)}$ for some $i\in T_j$. Then we randomly choose an index $k$ from $T_j$ with probability of $i\in T_j$ being chosen equal to $\dfrac{q_i}{\wt{q_j}}$. This determines that we are going to use the subchannel $W^{(k)}$ while imitating the usage of $W$. By now we flipped the input with probability $\theta_j$ (since we used the subchannel $\wt{W^{(j)}}$), while we want it to be flipped with probability $p_k \geq \theta_j$ overall, since we decided to use $W^{(k)}$. So the only thing we need to do it to ``flip" $b$ to $(1-b)$ with probability $\frac{p_k - \theta_j}{1-2\theta_j}$, and then output $z_k^{(b)}$ or $z_k^{(1-b)}$ correspondingly. 

Formally, we just describe the channel $W_1\,:\wt{\Y} \to \Y$ which proves that $\wW$ is upgraded with respect to $W$ by all of its transmission probabilities: for all ${k \in [m]}$, ${j \in [\sqrt{\l}]}$, ${b,c\in\bit}$ set
\begin{align}
\label{eq:W1_def}
     W_1\left(z_k^{c} \;\Big\lvert\;\wt{z_j^{(b)}}\right) = \begin{cases}
0, \qquad &k \notin T_j \\
\dfrac{q_k}{\sum\limits_{i\in T_j}q_i}\cdot\left(1 - \dfrac{p_k - \theta_j}{1-2\theta_j} \right), \qquad &k\in T_j,\ b = c,\\
\dfrac{q_k}{\sum\limits_{i\in T_j}q_i}\cdot\left(\dfrac{p_k - \theta_j}{1-2\theta_j} \right), \qquad &k\in T_j,\ b \not= c.\\
\end{cases}
\end{align}

It is easy to check that $W_1$ is a valid channel, and that it holds for any $k\in[m]$ and $c,x \in \bit$
 \begin{equation}
 \label{eq:upgrad_calculation}
       \sum_{j\in[\sqrt{\l}],\;b\in\bit}\wW\left(\wt{z_j^{(b)}}\,\Big\lvert\,x\right)W_1\left(z_k^{(c)} \;\Big\lvert\;\wt{z_j^{(b)}}\right) = W\left(z^{(c)}_k\;\Big\lvert\;x\right),
 \end{equation}
which proves that $\wW$ is indeed upgraded to $W$. We prove the above equality in Appendix~\ref{app:upgraded_calc}.

 It only remains to check that the property (iii) also holds, i.e. that the entropy did not decrease too much after we upgrade the channel $W$ to $\wW$. We have
 \[ H\left(\wW\right) = \sum_{j\in[\sqrt{\l}]}\wt{q_j}h(\theta_j) = 
 \sum_{j\in[\sqrt{\l}]}\left(\sum_{i \in T_j}q_i\right)h(\theta_j) = \sum_{k\in[m]} q_k h(\theta_{j_k}), \]
 where we again denoted by $j_k$ the index from $[\sqrt{\l}]$ for which $k \in T_{j_k}$. Therefore
 \[ H(W) - H\left(\wW\right) = \sum_{k\in [m]}q_k\big(h(p_k) - h(\theta_{j_k})\big) \leq  \sum_{k\in [m]}q_k\big(h(\theta_{j_k+1}) - h(\theta_{j_k})\big),  \]
 since $p_k \in [\theta_{j_k}, \theta_{j_k+1}]$ as $k\in T_{j_k}$. Finally, since $\theta_{j+1} - \theta_j = \frac1{2\sqrt{\l}}$, Proposition~\ref{prop:entropy_differ} gives
\begin{align} 
H(W) - H\left(\wW\right) \leq 
 \sum_{k\in [m]}q_k\big(h(\theta_{j_k+1}) - h(\theta_{j_k})\big) \leq
h\left(\dfrac1{2\sqrt{\l}}\right) \leq 2\cdot\dfrac1{2\sqrt{\l}}\log \left(2\sqrt{\l}\right) \leq \dfrac{\log\l}{\sqrt{\l}}.  &\qedhere
\end{align}
\end{proof}

\section{Suction at the ends}
\label{sec:suctions}
In this section we present the proof for Theorem~\ref{thm:kernel_seacrh_correct} in the case the standard Ar{\i}kans kernel was chosen in Algorithm~\ref{algo:kernel_search} -- the so-called suction at the ends regime. Recall that, as we discussed in section~\ref{sect:local}, this regime applies when the entropy of the channel $W$ falls into the interval ${(\l^{-4}, 1-\l^{-4})}$, and the algorithm directly takes a kernel $K = A_2^{\otimes \log\l}$,  where $A_2 =  \left( \begin{smallmatrix}1 & 0\\ 1 & 1\end{smallmatrix} \right)$ is the kernel of Ar{\i}kan's original polarizing transform, instead of trying out all the possible matrices. Note that multiplying by such a kernel $K$ is equivalent to just applying the Ar{\i}kan's $2 \times 2$  transform recursively $\log\l$ times. Suppose we have a BMS channel $W$ with $H(W)$ very close to $0$ or $1$. For Ar\i kan's basic transform, by working with the channel Bhattacharyya parameter $Z(W)$ instead of the entropy $H(W)$, it is well known that one of the two Ar\i kan bit-channels has $Z$ value getting much closer (quadratically closer) to the boundary of the interval $(0,1)$~\cite{arikan-polar,Korada_thesis}.
Using these ideas, we prove in this section that basic transform decreases the average of the potential function $\g(\cdot)$ of entropy at least by a factor of $\l^{-1/2}$ after $\log\l$ iterations for large enough $\l$. 

The basic Ar{\i}kan's transform takes one channel $W$ and splits it into a slightly worse channel $W^-$ and a slightly better channel $W^+$. Then the transform is applied recursively to $W^-$ and $W^+$, creating channels $W^{--}, W^{-+}, W^{+-},$ and $W^{++}$. One can think of the process as of a complete binary tree of depth $\log\l$, with the root node $W$, and any node at the level $i$ is of form $W^{B_i}$ for some $B_i \in \{-, +\}^i$, with two children $W^{B_i-}$ and $W^{B_i+}$. Denote $r  = \log\l$, then the channels at the leaves $\{W^{B_r}\}$, for all $B_r\in\{-,+\}^r$ are exactly the Ar{\i}kan's subchannels of $W$ with respect to the kernel $K = A_2^{\otimes \log\l}$. We are going to prove the following result
\begin{lem}
\label{lem:suction_evolution}
Let $W$ be a BMS channel with $H(W) \notin (\l^{-4}, 1 - \l^{-4})$, and $\a \in \left(0, \frac1{12}\right)$ be some constant. Let $\l$ be a power of two and denote $r = \log\l$. Then for $\l$ large enough such that $r\geq \max\left\{\dfrac1{\a}, 128\right\}$
\begin{equation}
    \label{eq:suction_lem}
    \sum_{B\in\{-,+\}^r}g_{\a}\left(H\left(W^{B}\right)\right) \leq \l^{1/2}g_{\a}\left(H(W)\right),
\end{equation}
where $\g(\cdot)$ is the potential function defined in~\eqref{eq:potential}.
\end{lem}

Clearly, the above lemma will imply the suction at the end case of Theorem~\ref{thm:kernel_seacrh_correct}, as the inequality $\log\l \geq \frac1{\a}$ holds by the conditions of this theorem.

For the analysis below, apart from the entropy of the channel, we will also use Bhattacharrya parameter $Z(W)$:
\[ Z(W) = \sum_{y \in \Y}\sqrt{W(y\,|\,0)W(y\,|\,1)}, \]
together with the inequalities which connect it to the entropy:
\begin{align}
\label{eq:Z-H}
    Z(W)^2 \leq H(W) \leq Z(W),
\end{align}
for any BMS channel $W$ (\cite[Lemma 1.5]{Korada_thesis}, \cite[Proposition 2]{Arikan_Source}). The reason we use this parameter is because of the following relations, which show how the Bhattacharrya parameter changes after the basic transform (\cite[Proposition 5]{arikan-polar} \cite{ModernCoding}, \cite[eq (13)]{ hassani-finite-scaling-paper-journal}):
\begin{align}
    Z(W^+) &= Z(W)^2, \label{eq:Z+}\\
    Z(W) \sqrt{2 - Z(W)^2} \leq Z(W^-) &\leq 2Z(W). \label{eq:Z-}
\end{align}

We will also use the conservation of conditional entropy on application of Ar{\i}kan's transform
\begin{equation}
\label{eq:entropy_martingale}
    H(W^+) + H(W^-) = 2H(W).
\end{equation}

\begin{proof}[Proof of Lemma~\ref{lem:suction_evolution}] The proof is presented in the next two sections, as it is divided into two parts: the case when $H(W) \leq \l^{-4}$ (suction at the lower end), and when $H(W) \geq 1 - \l^{-4}$ (suction at the upper end). 

\subsection{Suction at the lower end}
\label{sec:lower_suction}
Suppose $H(W) \leq \l^{-4}$ for this case, thus $Z(W) \leq \l^{-2} = 2^{-2r}$. 

First, recursive application of~\eqref{eq:entropy_martingale} gives
\begin{equation}
\label{eq:lower_eq0}
   \sum_{B\in\{-,+\}^r}H\left(W^{B}\right) = 2^rH(W),
\end{equation}
and since entropy is always nonnegative, this implies for any $B\in\{-, +\}^r$
\begin{equation}
\label{eq:lower_eq1}
    H\left(W^B\right) \leq 2^rH(W).
\end{equation}

Denote now $k = \left\lceil\log\frac1{\a}\right\rceil$, and notice that $\log r \geq k-1$ since $r \geq \frac1{\a}$. For $B \in \{-, +\}^r$, define $wt_+(B)$ to be number of $+$'s in $B$. We will split the summation in~\eqref{eq:suction_lem} into two parts: the part with $wt_+(B) < k$, and when $wt_+(B) \geq k$.

\medskip\noindent\textbf{First part.} 
Out of~\eqref{eq:lower_eq1} derive
\begin{equation}
\label{eq:lower_end_lower_sum}
    \begin{aligned}
         \sum_{wt_+(B) < k } g_{\a}\left(H\left(W^{B}\right)\right) 
         \leq \sum_{j=0}^{k-1} \binom{r}{j}g_{\a}\left(2^{r}H(W)\right) 
         \leq \log r\cdot\binom{r}{\log r}\cdot 2^{r\a}H(W)^{\a} \leq 2^{\log^2 r + r\a}\cdot H(W)^{\a},
    \end{aligned}
\end{equation}
where we used $\binom{r}{\log r} \leq \frac{r^{\log r}}{(\log r)!}$; the fact the $g_{\a}$ is increasing on $\left(0, \frac12\right)$ together with ${2^{r}H(W) \leq \l^{-3} < \frac12}$, and that $g_\a(x) \leq x^{\a}$ for $x \in (0,1)$.

\medskip\noindent\textbf{Second part.} We are going to use the following observation, which was already established in \cite[Lemma 1]{arikan-telatar} and can be proved by induction based on~\eqref{eq:Z+} and~\eqref{eq:Z-}:
\begin{claim}
\label{cl:Z_evolution}
Let $B \in \{-,+\}^r$, such that number of $+$'s in $B$ is equal to $s$. Then
\[  Z\left(W^B\right) \leq  \left(2^{r-s}\cdot Z(W)\right)^{2^s}.   \]
This corresponds to first using the upper bound in~\eqref{eq:Z-} $(r-s)$ times, and after that using~\eqref{eq:Z+} $s$ times while walking \textbf{down} the recursive binary tree of channels.
\end{claim}

 Then, using Claim~\ref{cl:Z_evolution} along with~\eqref{eq:Z-H} and the fact that $Z(W) \leq \l^{-2} =  2^{-2r}$, we obtain the following for any $B \in \{-, +\}^r$ with ${wt_+(B) = s\geq k}$:

\begin{equation}
\label{eq:lower_eq3}
\begin{aligned}
    H\left(W^B\right) \leq Z\left(W^B\right) \leq \left(2^{r-s}\cdot Z(W)\right)^{2^s} &\leq 2^{(r-s)2^s}\cdot Z(W)^{2^s-2}\cdot H(W) \\
    &\leq 2^{(r-s)2^s}\cdot 2^{ - 2r\cdot 2^s + 4r}\cdot H(W) \\
    &= 2^{-r2^s - s2^s + 4r}\cdot H(W) \\
    &\leq 2^{-r2^k - k2^k + 4r}\cdot H(W).
\end{aligned}
\end{equation}
Therefore
\begin{equation}
\label{eq:lower_end_upper_sum_prelim}
    \sum_{wt_+(B) \geq k}g_{\a}\left(H\left(W^{B}\right)\right) \leq \sum_{wt_+(B) \geq k}H\left(W^{B}\right)^{\a} \leq 2^r\cdot 2^{\a(-r2^k - k2^k + 4r)}\cdot H(W)^{\a}.
\end{equation}
Observe now the following chain of inequalities
\begin{align*}
    \dfrac{r}{2} + 4r\a + 2 \leq r \leq r\cdot2^k\a \leq r\cdot2^k\a + k\cdot2^k\a,
\end{align*}
which trivially holds for $\a \leq \dfrac1{12}$. Therefore
\begin{equation}
    r + \a(-r2^k-k2^k+4r) \leq \dfrac{r}2 - 2,
\end{equation}
and thus in~\eqref{eq:lower_end_upper_sum_prelim} obtain
\begin{equation}
\label{eq:lower_end_upper_sum}
    \sum_{wt_+(B) \geq k}g_{\a}\left(H\left(W^{B}\right)\right) \leq 2^{r/2-2}\cdot H(W)^{\a}.
\end{equation}

\medskip\noindent\textbf{Overall bound.} Combining~\eqref{eq:lower_end_lower_sum} and~\eqref{eq:lower_end_upper_sum} we derive
\begin{equation}
    \label{eq:lower_overall}
    \begin{aligned}
    \sum_{B\in\{-,+\}^r} g_{\a}\left(H\left(W^B\right)\right) 
    &\leq \left(2^{\log^2r  + r\a}  + 2^{r/2 - 2}\right)\cdot H(W)^{\a}\\
    &\leq 2^{r/2}\cdot \dfrac{H(W)^{\a}}{2} \\
    &\leq \l^{1/2}g_{\a}(H(W)),
    \end{aligned}
\end{equation}
where we used $\log^2r + r\a \leq \frac{r}{2} - 2$ for $r \geq 128$, and $\frac12 \leq (1-x)^{\a}$ for any $x \leq \frac12$. This proves Lemma~\ref{lem:suction_evolution} for the lower end case $H(W) \leq \l^{-4}$.

\subsection{Suction at the upper end}
\label{upper_suction}
Now consider the case $H(W) \geq 1 - \l^{-4}$. The proof is quite similar to the previous case, but we are going to track the distance from $H(W)$ (and $Z(W)$) to $1$ now. Specifically, denote
\begin{equation}
    \label{eq:upper_I-S}
    \begin{aligned}
    I(W) &= 1 - H(W),\\
    S(W) &= 1 - Z(W),
    \end{aligned}
\end{equation}
where $I(W)$ is actually the (symmetric) capacity of the channel, and $S(W)$ is just a notation we use in this proof. Notice that $g_{\a}(x) = g_{\a}(1-x)$, therefore it suffices to prove~\eqref{eq:suction_lem} with capacities of the channels instead of entropies in the inequality. Also notice that $I(W) \leq \l^{-4}$ for the current case of suction at the upper end.

Let us now derive the relations between $I(W)$, $S(W)$, as well as evolution of $S(\cdot)$ for $W^+$ and $W^-$, similar to~\eqref{eq:Z-H}, \eqref{eq:Z+}, \eqref{eq:Z-}, and~\eqref{eq:entropy_martingale}. Inequalities in~\eqref{eq:Z-H} imply
\begin{equation}
\begin{aligned}
    S(W) = 1 - Z(W) &\leq 1 - H(W) = I(W),\\
    I(W) = 1 - H(W) &\leq 1 - Z(W)^2 \leq 2(1-Z(W)) = 2S(W),
\end{aligned}
\end{equation}
so let us combine this to write
\begin{equation}
    \label{eq:S-I}
    S(W) \leq I(W) \leq 2S(W).
\end{equation}

Next,~\eqref{eq:Z+} and~\eqref{eq:Z-} give
\begin{align}
    \label{eq:S+}
    S(W^+)  &= 1 - Z(W)^2 \leq 2(1-Z(W)) \leq 2S(W),\\
    \label{eq:S-}
    S(W^-) &\leq 1 - Z(W)\sqrt{2 - Z(W)^2} \leq 2(1 - Z(W))^2 = 2S(W)^2,
\end{align}
where we used $1 - x\sqrt{2-x^2} \leq 2(1-x)^2$ for any $x\in(0,1)$, which can be proven easily by showing that equality holds at $x = 1$ and that the derivative of RHS minus LHS is negative on $(0, 1)$.

Finally, it easily follows from~\eqref{eq:lower_eq0} that
\begin{equation}
\label{eq:upper_eq0}
   \sum_{B\in\{-,+\}^r}I\left(W^{B}\right) = 2^rI(W),
\end{equation}
and since capacity is nonnegative as well, we also obtain for any $B\in\{-, +\}^r$
\begin{equation}
\label{eq:upper_cap_bound}
    I\left(W^B\right) \leq 2^rI(W).
\end{equation}

We now proceed with a very similar approach to the suction at the lower end case in Section~\ref{sec:lower_suction}: denote $k = \left\lceil\log\frac1{\a}\right\rceil$, and notice that $\log r \geq k-1$ since $r \geq \frac1{\a}$. For $B \in \{-, +\}^r$, define $wt_-(B)$ to be number of $-$'s in $B$. We will split the summation in~\eqref{eq:suction_lem} (but with capacities of channels instead of entropies) into two parts: the part with $wt_-(B) < k$, and when $wt_-(B) \geq k$.

\medskip\noindent\textbf{First part.} 
Out of~\eqref{eq:upper_cap_bound} derive, similarly to~\eqref{eq:lower_end_lower_sum}
\begin{equation}
\label{eq:upper_end_lower_sum}
    \begin{aligned}
         \sum_{wt_-(B) < k } g_{\a}\left(I\left(W^{B}\right)\right) 
         \leq \sum_{j=0}^{k-1} \binom{r}{j}g_{\a}\left(2^{r}I(W)\right) 
        \leq \log r\cdot\binom{r}{\log r}\cdot 2^{r\a}I(W)^{\a} \leq 2^{\log^2 r + r\a}\cdot I(W)^{\a}.
    \end{aligned}
\end{equation}

\medskip\noindent\textbf{Second part.} Similarly to Claim~\ref{cl:Z_evolution}, one can show via induction using~\eqref{eq:S+} and~\eqref{eq:S-} the following

\begin{claim}
\label{cl:S_evolution}
Let $B \in \{-,+\}^r$, such that number of $-$'s in $B$ is equal to $s$. Then
\[  S\left(W^B\right) \leq  2^{2^s-1}\left(2^{r-s}\cdot S(W)\right)^{2^s}.   \]
This corresponds to first using equality~\eqref{eq:S+} $(r-s)$ times, and after that using bound~\eqref{eq:S-} $s$ times while walking \textbf{down} the recursive binary tree of channels.
\end{claim}

 Using this claim with~\eqref{eq:S-I} and the fact that $S(W) \leq Z(W) \leq \l^{-4}\leq 2^{-4r}$ obtain for any $B \in \{-, +\}^r$ with ${wt_-(B) = s\geq k}$

\begin{equation}
\label{eq:upper_eq3}
\begin{aligned}
    I\left(W^B\right) \leq 2S\left(W^B\right) &\leq 2^{2^s}\cdot\left(2^{r-s}\cdot S(W)\right)^{2^s} &&\leq 2^{(r-s+1)2^s}\cdot S(W)^{2^s-1}\cdot I(W) \\
    &\leq 2^{(r-s+1)2^s - 4r2^s + 4r}\cdot I(W)
    &&= 2^{-2^s(3r+s-1) + 4r}\cdot I(W) \\
    &\leq 2^{-2^k(3r+k-1) + 4r}\cdot I(W) &&\leq 2^{-r2^k}\cdot I(W),
\end{aligned}
\end{equation}
where the last inequality uses $4r \leq 2^k(2t + k - 1)$, which holds trivially for $k\geq 1$. Therefore
\begin{equation}
\label{eq:upper_end_upper_sum}
    \sum_{wt_-(B) \geq k}g_{\a}\left(I\left(W^{B}\right)\right) \leq \sum_{wt_-(B) \geq k}I\left(W^{B}\right)^{\a} \leq 2^r\cdot 2^{-\a r2^k}\cdot I(W)^{\a} \leq I(W)^{\a},
\end{equation}
since $\a\cdot2^k \geq 1$ by the choice of $k$. 

\medskip\noindent\textbf{Overall bound.} The bounds~\eqref{eq:upper_end_lower_sum} and~\eqref{eq:upper_end_upper_sum} give us
\begin{equation}
    \label{eq:upper_overall}
    \begin{aligned}
    \sum_{B\in\{-,+\}^r} g_{\a}\left(H\left(W^B\right)\right) = \sum_{B\in\{-,+\}^r} g_{\a}\left(I\left(W^B\right)\right) 
    \leq \left(2^{\log^2r  + r\a}  + 1\right)\cdot I(W)^{\a} \leq \l^{1/2}g_{\a}(H(W))
    \end{aligned}
\end{equation}
for large enough $r$ when $H(W) \geq 1 - \l^{-4}$. This completes the proof of Lemma~\ref{lem:suction_evolution}.
\end{proof}

\section{Code construction, encoding and decoding procedures} \label{sect:cons}
Before presenting our code construction and encoding/decoding procedures, we first distinguish the difference between the code construction and the encoding procedure. The objectives of code construction for polar-type codes are two-fold: First, find the $N\times N$ encoding matrix; second, find the set of noiseless bits under the successive cancellation decoder, which will carry the message bits.
On the other hand, by encoding we simply mean the procedure of obtaining the codeword $\bX_{[1:N]}$ by multiplying the information vector $\bU_{[1:N]}$ with the encoding matrix, where we only put information in the noiseless bits in $\bU_{[1:N]}$ and set all the frozen bits to be $0$.
As we will see at the end of this section, while the code construction has complexity polynomial in $N$, the encoding procedure only has complexity $O_{\l}(N\log N)$.

For polar codes with a fixed invertible kernel $K\in \bit^{\l\times\l}$, the polarization process works as follows: We start with some BMS channel $W$. After applying the polar transform to $W$ using kernel $K$, we obtain $\ell$ bit-channels $\{W_i:i\in[\ell]\}$ as defined in \eqref{Arikan_subchannels}.
Next we apply the polar transform using kernel $K$ to each of these $\ell$ bit-channels, and we write the polar transform of $W_i$ as $\{W_{ij}:j\in[\ell]\}$.
Then we apply the polar transform to each of the $\ell^2$ bit channels $\{W_{i_1,i_2}:i_1,i_2\in[\ell]\}$
and obtain $\{W_{i_1,i_2,i_3}:i_1,i_2,i_3\in[\ell]\}$, so on and so forth.
After $t$ rounds of polar transforms, we obtain $\ell^t$ bit-channels
$\{W_{i_1,\dots,i_t}:i_1,\dots,i_t\in[\ell]\}$, and one can show that these are the bit-channels seen by the successive  cancellation decoder when decoding the corresponding polar codes constructed from kernel $K$.

\begin{algorithm}[t]
 \caption{Degraded binning algorithm}
\label{algo:bin}
\DontPrintSemicolon
\SetAlgoLined
  \KwInput{$W:\{0,1\}\to\Y$, bound $\mathsf{Q}$ on the output alphabet size after binning}
  \KwOutput{$\widetilde{W}:\{0,1\}\to\widetilde{\Y}$, where $|\widetilde{\Y}|\le \mathsf{Q}$}
  Initialize the new channel $\widetilde{W}$ with output symbols $\tilde{y}_1,\tilde{y}_2,\dots,\tilde{y}_{\mathsf{Q}}$ by setting $\widetilde{W}(\tilde{y}_i|x)=0$ for all $i\in[\mathsf{Q}]$ and $x\in\{0,1\}$
  
  \For{$y\in\Y$}
  {$p(0|y)\gets \frac{W(y|0)}{W(y|0)+W(y|1)}$
  
  $i\gets \lceil \Q \cdot p(0|y) \rceil$
  
  \If{$i=0$}
  {$i\gets 1$ 
  \tcp*{$i=0$ if and only if $p(0|y)=0$; we merge this single point into the next bin} }

  $\widetilde{W}(\tilde{y}_i|0)\gets \widetilde{W}(\tilde{y}_i|0)+W(y|0)$
  
  $\widetilde{W}(\tilde{y}_i|1)\gets \widetilde{W}(\tilde{y}_i|1)+W(y|1)$}
  
  \Return $\widetilde{W}$
\end{algorithm}

For our purpose, we need to use polar codes with mixed kernels, and we need to search for a ``good" kernel at each step of polarization.
We will also introduce new notation for the bit-channels in order to indicate the usage of different kernels for different bit-channels.
As mentioned in Sections~\ref{sect:encdec} and~\ref{sect:local}, we need to use a binning algorithm (Algorithm~\ref{algo:bin}) to quantize all the bit-channels we obtain in the code construction procedure. As long as we choose the parameter $\Q$ in Algorithm~\ref{algo:bin} to be a large enough polynomial of $N$, the quantized channel can be used as a very good approximation of the original channel. This is made precise by \cite[Proposition 13]{GX15}: For $W$ and $\widetilde{W}$ in Algorithm~\ref{algo:bin}, we have\footnote{Note that the binning algorithm (Algorithm 2) in \cite{GX15} has one minor difference from the binning algorithm (Algorithm~\ref{algo:bin}) in this paper: In \cite{GX15}, the binning algorithm outputs a channel with $\Q+1$ outputs in contrast to $\Q$ outputs in this paper. More precisely, line 5-7 in Algorithm~\ref{algo:bin} of this paper is not included in the algorithm in \cite{GX15}, but one can easily check that this minor difference does not affect the proof at all.}
\begin{equation}  \label{eq:ttt}
H(W) \le H(\widetilde{W})
\le H(W) + \frac{2\log \Q}{\Q}.
\end{equation}
Given a BMS channel $W$, our code construction works as follows:
\begin{enumerate}
    \item {\bf Step 0:} 
    We first use Algorithm~\ref{algo:bin} to quantize/bin the output alphabet of $W$ such that the resulting (degraded) channel has at most $N^3$ outputs, i.e., we set $\Q=N^3$ in Algorithm~\ref{algo:bin}.
    Note that the parameter $\Q$ can be chosen as any polynomial of $N$. By changing the value of $\Q$, we obtain a tradeoff between the decoding error probability and the gap to capacity; see Theorem~\ref{thm:main1} at the end of this section.
    Here we choose the special case of $\Q=N^3$ to give a concrete example of code construction.
    Next we use Algorithm~\ref{algo:kernel_search} in Section~\ref{sect:KC} to find a good kernel\footnote{We will prove in Proposition~\ref{prop:approx_accumulation} that the error parameter $\Delta$ in Algorithm~\ref{algo:kernel_search} can be chosen as $\Delta=\frac{6\l \log N}{N^2}$ when we set $\Q=N^3$.} for the quantized channel and denote it as $K_1^{(0)}$. 
    Recall from Section~\ref{sect:mix} that a kernel is good if all but a $\tilde{O}(\l^{-1/2})$ fraction of the bit-channels obtained after polar transform by this kernel have entropy $\l^{-\Omega(\log\l)}$-close to either $0$ or $1$.
    The superscript $(0)$ in $K_1^{(0)}$ indicates that this is the kernel used in Step 0 of polarization. In this case, we use $\{W_i(B,K_1^{(0)}):i\in[\ell]\}$ to denote the $\ell$ bit-channels resulting from the polar transform of the quantized version of $W$ using kernel $K_1^{(0)}$.
    Here $B$ stands for the binning operation, and the arguments in the brackets are the operations to obtain the bit-channel $W_i(B,K_1^{(0)})$ from $W$: first bin the outputs of $W$ and then perform the polar transform using kernel $K_1^{(0)}$.
    For each $i\in[\ell]$, we again use Algorithm~\ref{algo:bin} to quantize/bin the output alphabet of $W_i(B,K_1^{(0)})$ such that the resulting (degraded) bit-channel 
    $W_i(B,K_1^{(0)},B)$ has at most $N^3$ outputs.
    
    \item {\bf Step 1:} For each $i_1\in[\ell]$, we use Algorithm~\ref{algo:kernel_search} to find a good kernel for the quantized bit-channel $W_{i_1}(B,K_1^{(0)},B)$ and denote it as $K_{i_1}^{(1)}$.
    The $\ell$ bit-channels resulting from the polar transform of $W_{i_1}(B,K_1^{(0)},B)$ using kernel $K_{i_1}^{(1)}$ are denoted as $\{W_{i_1,i_2}(B,K_1^{(0)},B,K_{i_1}^{(1)}):i_2\in[\ell]\}$. In this step, we will obtain $\ell^2$ bit-channels $\{W_{i_1,i_2}(B,K_1^{(0)},B,K_{i_1}^{(1)}):i_1,i_2\in[\ell]\}$. For each of them, we use Algorithm~\ref{algo:bin} to quantize/bin its output alphabet such that the resulting (degraded) bit-channels 
    $\{W_{i_1,i_2}(B,K_1^{(0)},B,K_{i_1}^{(1)},B):i_1,i_2\in[\ell]\}$ has at most $N^3$ outputs. See Fig.~\ref{fig:tree} for an illustration of this procedure for the special case of $\l=3$.

    \item We repeat the polar transforms and binning operations at each step of the code construction. More precisely, at {\bf Step $j$} we have $\ell^j$ bit-channels 
    $$
    \{W_{i_1,i_2,\dots,i_j}(B,K_1^{(0)},B,K_{i_1}^{(1)},B,\dots,K_{i_1,\dots,i_{j-1}}^{(j-1)},B):i_1,i_2,\dots,i_j\in[\ell]\}.
    $$
    This notation is a bit messy, so we introduce some simplified notation for the bit-channels obtained with and without binning operations: We still use
    $$
    W_{i_1,i_2,\dots,i_j}(K_1^{(0)},K_{i_1}^{(1)},\dots,K_{i_1,\dots,i_{j-1}}^{(j-1)})
    $$ 
    to denote the bit-channel obtained without the binning operations at all, and we use 
    $$
    W_{i_1,i_2,\dots,i_j}^{\bin}(K_1^{(0)},K_{i_1}^{(1)},\dots,K_{i_1,\dots,i_{j-1}}^{(j-1)})
    $$ 
    to denote the bit-channel obtained with binning operations performed at every step from Step $0$ to Step $j-1$, i.e.,
    \begin{align*}
        W_{i_1,i_2,\dots,i_j}^{\bin}(K_1^{(0)},K_{i_1}^{(1)},\dots,K_{i_1,\dots,i_{j-1}}^{(j-1)})
        := W_{i_1,i_2,\dots,i_j}(B,K_1^{(0)},B,K_{i_1}^{(1)},B,\dots,K_{i_1,\dots,i_{j-1}}^{(j-1)},B).
    \end{align*}
    Moreover, we use $W_{i_1,i_2,\dots,i_j}^{\bin *}(K_1^{(0)},K_{i_1}^{(1)},\dots,K_{i_1,\dots,i_{j-1}}^{(j-1)})$ to denote the bit-channel obtained with binning operations performed at every step except for the last step, i.e.,
    \begin{align*}
        W_{i_1,i_2,\dots,i_j}^{\bin*}(K_1^{(0)},K_{i_1}^{(1)},\dots,K_{i_1,\dots,i_{j-1}}^{(j-1)})
        := W_{i_1,i_2,\dots,i_j}(B,K_1^{(0)},B,K_{i_1}^{(1)},B,\dots,B,K_{i_1,\dots,i_{j-1}}^{(j-1)}).
    \end{align*}
    
    Next we use Algorithm~\ref{algo:kernel_search} to find a good kernel for each of them and denote the kernel as $K_{i_1,\dots,i_j}^{(j)}$. After applying polar transforms using these kernels, we obtain $\ell^{j+1}$ bit-channels
    $$
    \{W_{i_1,\dots,i_{j+1}}^{\bin*}(K_1^{(0)},K_{i_1}^{(1)},\dots,K_{i_1,\dots,i_j}^{(j)}):i_1,\dots,i_{j+1}\in[\ell]\}.
    $$
    Then we quantize/bin the output alphabets of these bit-channels using Algorithm~\ref{algo:bin} and obtain the following $\ell^{j+1}$ quantized bit-channels
    $$
    \{W_{i_1,\dots,i_{j+1}}^{\bin}(K_1^{(0)},K_{i_1}^{(1)},
    \dots,K_{i_1,\dots,i_j}^{(j)}):i_1,\dots,i_{j+1}\in[\ell]\}.
    $$
    \item After {\bf step $t-1$}, we obtain $N=\ell^t$ quantized bit-channels
    $$
    \{W_{i_1,\dots,i_t}^{\bin}(K_1^{(0)},K_{i_1}^{(1)},\dots,K_{i_1,\dots,i_{t-1}}^{(t-1)}):i_1,i_2,\dots,i_j\in[\ell]\},
    $$
    and we have also obtained all the kernels in each step of polarization. More precisely, we have $\ell^i$ kernels in step $i$, so from step $0$ to step $t-1$, we have $1+\ell+\dots+\ell^{t-1}=\frac{N-1}{\ell-1}$ kernels in total.
    
    \item Find the set of good (noiseless) indices. More precisely, we use the shorthand notation\footnote{We omit the reference to the kernels in the notation $H_{i_1,\dots,i_t}(W)$ and $H_{i_1,\dots,i_t}^{\bin}(W)$.}
    \begin{equation} \label{eq:defHi}
    \begin{aligned}
    H_{i_1,\dots,i_t}(W) &:=
    H(W_{i_1,\dots,i_t}(K_1^{(0)},K_{i_1}^{(1)},\dots,K_{i_1,\dots,i_{t-1}}^{(t-1)}))   \\
    H_{i_1,\dots,i_t}^{\bin}(W) &:=
    H(W_{i_1,\dots,i_t}^{\bin}(K_1^{(0)},K_{i_1}^{(1)},\dots,K_{i_1,\dots,i_{t-1}}^{(t-1)}))
    \end{aligned}
    \end{equation}
    and define the set of good indices as
    \begin{equation}  \label{eq:Sgood}
    \S_{\good}:=\left\{(i_1,i_2,\dots,i_t)\in[\ell]^t: H_{i_1,\dots,i_t}^{\bin}(W) \le \frac{7\ell \log N}{N^2} \right\}.
    \end{equation}
    
    \item Finally, we need to construct the encoding matrix from these $\frac{N-1}{\ell-1}$ kernels. The kernels we obtained in step $j$ are
    $$
    \{K_{i_1,\dots,i_j}^{(j)}:i_1,\dots,i_j\in[\ell]\}.
    $$
    For an integer $i\in[\ell^j]$, we write the $j$-digit $\ell$-ary expansion of $i-1$ as $(\tilde{i}_1,\tilde{i}_2,\dots,\tilde{i}_j)$, where $\tilde{i}_j$ is the least significant digit and $\tilde{i}_1$ is the most significant digit, and each digit takes value in $\{0,1,\dots,\ell-1\}$. Let $(i_1,i_2,\dots,i_j):=(\tilde{i}_1+1,\tilde{i}_2+1,\dots,\tilde{i}_j+1)$, and define the mapping $\tau_j:[\ell^j]\to[\ell]^j$ as 
    \begin{equation} \label{eq:deftau}
    \tau_j(i):=(i_1,i_2,\dots,i_j) 
    \text{~~for~} i\in[\ell^j]   .
    \end{equation}
    This is a one-to-one mapping between $[\ell^j]$ and $[\ell]^j$, and we use the shorthand notation $K_i^{(j)}$ to denote $K_{\tau_j(i)}^{(j)}$ for $i\in[\ell^j]$.
    For each $j\in\{0,1,\dots,t-1\}$, we define the block diagonal matrices $\overline{D}^{(j)}$ with size $\ell^{j+1}\times \ell^{j+1}$ and $D^{(j)}$ with size $N\times N$ as
    \begin{equation} \label{eq:defD}
        \overline{D}^{(j)}:=
        \Diag(K_1^{(j)},K_2^{(j)},\dots,K_{\ell^j}^{(j)}) , \quad\quad\quad
        D^{(j)}:=\underbrace{\{\overline{D}^{(j)},\overline{D}^{(j)},\dots,\overline{D}^{(j)}\}}_{
        \text{number of }\overline{D}^{(j)} \text{ is } \ell^{t-j-1}}  .
    \end{equation}
    For $i\in[\ell^t]$, we have $\tau_t(i)=(i_1,\dots,i_t)$.
    For $j\in[t-1]$, we
    define the permutation $\pi^{(j)}$ on the set $[\ell^t]$ as
    \begin{equation}\label{eq:defpi}
    \pi^{(j)}(i):=\tau_t^{-1}(i_1,\dots,i_{t-j-1},i_t,i_{t-j},i_{t-j+1},\dots,i_{t-1})
    \quad \forall i\in[\ell^t].
    \end{equation}
    By this definition, $\pi^{(j)}$ simply keeps the first $t-j-1$ digits of $i$ to be the same and performs a cyclic shift on the last $j+1$ digits. Here we give some concrete examples:
    \begin{align*}
    \pi^{(1)}(i) & =\tau_t^{-1}(i_1,\dots,i_{t-2},i_t,i_{t-1}),   \\
    \pi^{(2)}(i) & =\tau_t^{-1}(i_1,\dots,i_{t-3},i_t,i_{t-2},i_{t-1}),  \\
    \pi^{(3)}(i) & =\tau_t^{-1}(i_1,\dots,i_{t-4},i_t,i_{t-3},i_{t-2},i_{t-1}) , \\
    \pi^{(t-1)}(i) & =\tau_t^{-1}(i_t,i_1,i_2,\dots,i_{t-1}) .
    \end{align*}
    For each $j\in[t-1]$, let $Q^{(j)}$ be the $\ell^t \times \ell^t$ permutation matrix corresponding to the permutation $\pi^{(j)}$, i.e., $Q^{(j)}$ is the permutation matrix such that
    \begin{equation}\label{eq:defQ}
    (U_1,U_2,\dots,U_{\ell^t})Q^{(j)}
    =(U_{\pi^{(j)}(1)},U_{\pi^{(j)}(2)},\dots,U_{\pi^{(j)}(\ell^t)}) .
    \end{equation}
    Finally, for each $j\in[t]$, we define the $N\times N$ matrix
    \begin{equation} \label{eq:defMj}
    M^{(j)}:=D^{(j-1)}Q^{(j-1)}D^{(j-2)}Q^{(j-2)}\dots D^{(1)}Q^{(1)} D^{(0)} .
    \end{equation}
    Therefore, $M^{(j)},j\in[t]$ satisfy the following recursive relation:
    $$
    M^{(1)}=D^{(0)}, \quad\quad
    M^{(j+1)}=D^{(j)}Q^{(j)}M^{(j)} .
    $$
    Our encoding matrix for code length $N=\ell^t$ is the submatrix of $M^{(t)}$ consisting of all the row vectors with indices belonging to the set $\S_{\good}$ defined in \eqref{eq:Sgood};
    see the next paragraph for a detailed description of the encoding procedure.
\end{enumerate}

\begin{figure}
    \centering
    \begin{tikzpicture}
\node at (7, 9.4) (w0) {$W$};
\node [block, align=center] at (7, 7.6) (q0) 
{Bin, then\\find $K_1^{(0)}$};
\node at (2.5, 5.8) (w1) {$W_1$};
\node at (7, 5.8) (w2) {$W_2$};
\node at (11.5, 5.8) (w3) {$W_3$};
\node [block, align=center] at (2.5, 4) (q1) 
{Bin, then\\find $K_1^{(1)}$};
\node [block, align=center] at (7, 4) (q2) {Bin, then\\find $K_2^{(1)}$};
\node [block, align=center] at (11.5, 4) (q3) 
{Bin, then\\find $K_3^{(1)}$};
\node at (1, 2) (w11) {$W_{1,1}$};
\node at (2.5, 2) (w12) {$W_{1,2}$};
\node at (4, 2) (w13) {$W_{1,3}$};
\node at (5.5, 2) (w21) {$W_{2,1}$};
\node at (7, 2) (w22) {$W_{2,2}$};
\node at (8.5, 2) (w23) {$W_{2,3}$};
\node at (10, 2) (w31) {$W_{3,1}$};
\node at (11.5, 2) (w32) {$W_{3,2}$};
\node at (13, 2) (w33) {$W_{3,3}$};

\node at (1, 1.5)  {$\vdots$};
\node at (2.5, 1.5)  {$\vdots$};
\node at (4, 1.5)  {$\vdots$};
\node at (5.5, 1.5)  {$\vdots$};
\node at (7, 1.5)  {$\vdots$};
\node at (8.5, 1.5)  {$\vdots$};
\node at (10, 1.5)  {$\vdots$};
\node at (11.5, 1.5) {$\vdots$};
\node at (13, 1.5)  {$\vdots$};

\draw[->, thick] (w0)--(q0);
\draw[->, thick] (q0)--(w1);
\draw[->, thick] (q0)--(w2);
\draw[->, thick] (q0)--(w3);
\draw[->, thick] (w1)--(q1);
\draw[->, thick] (w2)--(q2);
\draw[->, thick] (w3)--(q3);
\draw[->, thick] (q1)--(w11);
\draw[->, thick] (q1)--(w12);
\draw[->, thick] (q1)--(w13);
\draw[->, thick] (q2)--(w21);
\draw[->, thick] (q2)--(w22);
\draw[->, thick] (q2)--(w23);
\draw[->, thick] (q3)--(w31);
\draw[->, thick] (q3)--(w32);
\draw[->, thick] (q3)--(w33);
\end{tikzpicture} 
    \caption{Illustration of code construction for the special case of $\l=3$.}
    \label{fig:tree}
\end{figure}

\begin{figure}
\begin{tikzpicture}

\node at (15, 9) (y9) {$Y_1$};
\node at (15, 8) (y8) {$Y_2$};
\node at (15, 7) (y7) {$Y_3$};
\node at (15, 6) (y6) {$Y_4$};
\node at (15, 5) (y5) {$Y_5$};
\node at (15, 4) (y4) {$Y_6$};
\node at (15, 3) (y3) {$Y_7$};
\node at (15, 2) (y2) {$Y_8$};
\node at (15, 1) (y1) {$Y_9$};

\node [block] at (13.5, 9) (w9) {$W$};
\node [block] at (13.5, 8) (w8) {$W$};
\node [block] at (13.5, 7) (w7) {$W$};
\node [block] at (13.5, 6) (w6) {$W$};
\node [block] at (13.5, 5) (w5) {$W$};
\node [block] at (13.5, 4) (w4) {$W$};
\node [block] at (13.5, 3) (w3) {$W$};
\node [block] at (13.5, 2) (w2) {$W$};
\node [block] at (13.5, 1) (w1) {$W$};

\node at (12, 9) (x9) {$X_1$};
\node at (12, 8) (x8) {$X_2$};
\node at (12, 7) (x7) {$X_3$};
\node at (12, 6) (x6) {$X_4$};
\node at (12, 5) (x5) {$X_5$};
\node at (12, 4) (x4) {$X_6$};
\node at (12, 3) (x3) {$X_7$};
\node at (12, 2) (x2) {$X_8$};
\node at (12, 1) (x1) {$X_9$};

\node [sblock] at (9.5, 8) (K3) {$K_1^{(0)}$};
\node [sblock] at (9.5, 5) (K2) {$K_1^{(0)}$};
\node [sblock] at (9.5, 2) (K1) {$K_1^{(0)}$};

\node at (7, 9) (u9) {$U_1^{(1)}$};
\node at (7, 8) (u8) {$U_2^{(1)}$};
\node at (7, 7) (u7) {$U_3^{(1)}$};
\node at (7, 6) (u6) {$U_4^{(1)}$};
\node at (7, 5) (u5) {$U_5^{(1)}$};
\node at (7, 4) (u4) {$U_6^{(1)}$};
\node at (7, 3) (u3) {$U_7^{(1)}$};
\node at (7, 2) (u2) {$U_8^{(1)}$};
\node at (7, 1) (u1) {$U_9^{(1)}$};

\node at (5, 9) (v9) {$V_1^{(1)}$};
\node at (5, 8) (v8) {$V_2^{(1)}$};
\node at (5, 7) (v7) {$V_3^{(1)}$};
\node at (5, 6) (v6) {$V_4^{(1)}$};
\node at (5, 5) (v5) {$V_5^{(1)}$};
\node at (5, 4) (v4) {$V_6^{(1)}$};
\node at (5, 3) (v3) {$V_7^{(1)}$};
\node at (5, 2) (v2) {$V_8^{(1)}$};
\node at (5, 1) (v1) {$V_9^{(1)}$};

\node [sblock] at (2.5, 8) (KK3) {$K_1^{(1)}$};
\node [sblock] at (2.5, 5) (KK2) {$K_2^{(1)}$};
\node [sblock] at (2.5, 2) (KK1) {$K_3^{(1)}$};

\node at (0, 9) (uu9) {$U_1$};
\node at (0, 8) (uu8) {$U_2$};
\node at (0, 7) (uu7) {$U_3$};
\node at (0, 6) (uu6) {$U_4$};
\node at (0, 5) (uu5) {$U_5$};
\node at (0, 4) (uu4) {$U_6$};
\node at (0, 3) (uu3) {$U_7$};
\node at (0, 2) (uu2) {$U_8$};
\node at (0, 1) (uu1) {$U_9$};

\draw[->, thick] (uu9)--(1.2,9);
\draw[->, thick] (uu8)--(1.2,8);
\draw[->, thick] (uu7)--(1.2,7);
\draw[->, thick] (uu6)--(1.2,6);
\draw[->, thick] (uu5)--(1.2,5);
\draw[->, thick] (uu4)--(1.2,4);
\draw[->, thick] (uu3)--(1.2,3);
\draw[->, thick] (uu2)--(1.2,2);
\draw[->, thick] (uu1)--(1.2,1);

\draw[->, thick] (3.8,9)--(v9);
\draw[->, thick] (3.8,8)--(v8);
\draw[->, thick] (3.8,7)--(v7);
\draw[->, thick] (3.8,6)--(v6);
\draw[->, thick] (3.8,5)--(v5);
\draw[->, thick] (3.8,4)--(v4);
\draw[->, thick] (3.8,3)--(v3);
\draw[->, thick] (3.8,2)--(v2);
\draw[->, thick] (3.8,1)--(v1);

\draw[->, thick] (v9.east)--(u9.west);
\draw[->, thick] (v8.east)--(u6.west);
\draw[->, thick] (v7.east)--(u3.west);
\draw[->, thick] (v6.east)--(u8.west);
\draw[->, thick] (v5.east)--(u5.west);
\draw[->, thick] (v4.east)--(u2.west);
\draw[->, thick] (v3.east)--(u7.west);
\draw[->, thick] (v2.east)--(u4.west);
\draw[->, thick] (v1.east)--(u1.west);

\draw[->, thick] (u9)--(8.2,9);
\draw[->, thick] (u8)--(8.2,8);
\draw[->, thick] (u7)--(8.2,7);
\draw[->, thick] (u6)--(8.2,6);
\draw[->, thick] (u5)--(8.2,5);
\draw[->, thick] (u4)--(8.2,4);
\draw[->, thick] (u3)--(8.2,3);
\draw[->, thick] (u2)--(8.2,2);
\draw[->, thick] (u1)--(8.2,1);

\draw[->, thick] (10.8,9)--(x9);
\draw[->, thick] (10.8,8)--(x8);
\draw[->, thick] (10.8,7)--(x7);
\draw[->, thick] (10.8,6)--(x6);
\draw[->, thick] (10.8,5)--(x5);
\draw[->, thick] (10.8,4)--(x4);
\draw[->, thick] (10.8,3)--(x3);
\draw[->, thick] (10.8,2)--(x2);
\draw[->, thick] (10.8,1)--(x1);

\draw[->, thick] (x9)--(w9);
\draw[->, thick] (x8)--(w8);
\draw[->, thick] (x7)--(w7);
\draw[->, thick] (x6)--(w6);
\draw[->, thick] (x5)--(w5);
\draw[->, thick] (x4)--(w4);
\draw[->, thick] (x3)--(w3);
\draw[->, thick] (x2)--(w2);
\draw[->, thick] (x1)--(w1);

\draw[->, thick] (w9)--(y9);
\draw[->, thick] (w8)--(y8);
\draw[->, thick] (w7)--(y7);
\draw[->, thick] (w6)--(y6);
\draw[->, thick] (w5)--(y5);
\draw[->, thick] (w4)--(y4);
\draw[->, thick] (w3)--(y3);
\draw[->, thick] (w2)--(y2);
\draw[->, thick] (w1)--(y1);

\end{tikzpicture}
\caption{Illustration of the encoding process $\bX_{[1:N]}=\bU_{[1:N]} M^{(t)}$ for the special case of $\ell=3$ and $t=2$. Here $\bX_{[1:N]}$ and $\bU_{[1:N]}$ are row vectors. All four kernels in this figure $K_1^{(0)},K_1^{(1)},K_2^{(1)},K_3^{(1)}$ have size $3\times 3$, and the outputs of each kernel is obtained by multiplying the inputs with the kernel, e.g. $\bV_{[1:3]}^{(1)}=\bU_{[1:3]} K_1^{(1)}$.}
\label{fig:ill32}
\end{figure}

\begin{figure}
\centering
\begin{tikzpicture}

\node at (5, 9) (v9) {$V_1^{(1)}$};
\node at (5, 8) (v8) {$V_2^{(1)}$};
\node at (5, 7) (v7) {$V_3^{(1)}$};

\node [sblock] at (2.5, 8) (KK3) {$K_1^{(1)}$};

\node at (0, 9) (uu9) {$U_1$};
\node at (0, 8) (uu8) {$U_2$};
\node at (0, 7) (uu7) {$U_3$};

\node [block] at (7.5, 9)   (w9) {$W_1(K_1^{(0)})$};
\node [block] at (7.5, 8)   (w8) {$W_1(K_1^{(0)})$};
\node [block] at (7.5, 7)   (w7) {$W_1(K_1^{(0)})$};

\node at (10.5, 9) (y9) {$(Y_1,Y_2,Y_3)$};
\node at (10.5, 8) (y8) {$(Y_4,Y_5,Y_6)$};
\node at (10.5, 7) (y7) {$(Y_7,Y_8,Y_9)$};

\draw[->, thick] (uu9)--(1.2,9);
\draw[->, thick] (uu8)--(1.2,8);
\draw[->, thick] (uu7)--(1.2,7);

\draw[->, thick] (3.8,9)--(v9);
\draw[->, thick] (3.8,8)--(v8);
\draw[->, thick] (3.8,7)--(v7);

\draw[->, thick] (v9)--(w9);
\draw[->, thick] (v8)--(w8);
\draw[->, thick] (v7)--(w7);

\draw[->, thick] (w9)--(y9);
\draw[->, thick] (w8)--(y8);
\draw[->, thick] (w7)--(y7);

\end{tikzpicture}
\caption{The (stochastic) mapping from $\bU_{[1:3]}$ to $\bY_{[1:9]}$}
\label{fig:top3}
\end{figure}

\begin{figure}
\centering
\begin{tikzpicture}

\node at (5, 9) (v9) {$V_4^{(1)}$};
\node at (5, 8) (v8) {$V_5^{(1)}$};
\node at (5, 7) (v7) {$V_6^{(1)}$};

\node [sblock] at (2.5, 8) (KK3) {$K_2^{(1)}$};

\node at (0, 9) (uu9) {$U_4$};
\node at (0, 8) (uu8) {$U_5$};
\node at (0, 7) (uu7) {$U_6$};

\node [block] at (7.5, 9)   (w9) {$W_2(K_1^{(0)})$};
\node [block] at (7.5, 8)   (w8) {$W_2(K_1^{(0)})$};
\node [block] at (7.5, 7)   (w7) {$W_2(K_1^{(0)})$};

\node at (10.5, 9) (y9) {$(V_1^{(1)},Y_1,Y_2,Y_3)$};
\node at (10.5, 8) (y8) {$(V_2^{(1)},Y_4,Y_5,Y_6)$};
\node at (10.5, 7) (y7) {$(V_3^{(1)},Y_7,Y_8,Y_9)$};

\draw[->, thick] (uu9)--(1.2,9);
\draw[->, thick] (uu8)--(1.2,8);
\draw[->, thick] (uu7)--(1.2,7);

\draw[->, thick] (3.8,9)--(v9);
\draw[->, thick] (3.8,8)--(v8);
\draw[->, thick] (3.8,7)--(v7);

\draw[->, thick] (v9)--(w9);
\draw[->, thick] (v8)--(w8);
\draw[->, thick] (v7)--(w7);

\draw[->, thick] (w9)--(y9);
\draw[->, thick] (w8)--(y8);
\draw[->, thick] (w7)--(y7);

\end{tikzpicture}
\caption{The (stochastic) mapping from $\bU_{[4:6]}$ to $(\bV_{[1:3]}^{(1)},\bY_{[1:9]})$}
\label{fig:mid3}
\end{figure}

\begin{figure}[t!]
\centering
\begin{tikzpicture}

\node at (5, 9) (v9) {$V_7^{(1)}$};
\node at (5, 8) (v8) {$V_8^{(1)}$};
\node at (5, 7) (v7) {$V_9^{(1)}$};

\node [sblock] at (2.5, 8) (KK3) {$K_3^{(1)}$};

\node at (0, 9) (uu9) {$U_7$};
\node at (0, 8) (uu8) {$U_8$};
\node at (0, 7) (uu7) {$U_9$};

\node [block] at (7.5, 9)   (w9) {$W_3(K_1^{(0)})$};
\node [block] at (7.5, 8)   (w8) {$W_3(K_1^{(0)})$};
\node [block] at (7.5, 7)   (w7) {$W_3(K_1^{(0)})$};

\node at (11.5, 9) (y9) {$(V_1^{(1)},V_4^{(1)},Y_1,Y_2,Y_3)$};
\node at (11.5, 8) (y8) {$(V_2^{(1)},V_5^{(1)},Y_4,Y_5,Y_6)$};
\node at (11.5, 7) (y7) {$(V_3^{(1)},V_6^{(1)},Y_7,Y_8,Y_9)$};

\draw[->, thick] (uu9)--(1.2,9);
\draw[->, thick] (uu8)--(1.2,8);
\draw[->, thick] (uu7)--(1.2,7);

\draw[->, thick] (3.8,9)--(v9);
\draw[->, thick] (3.8,8)--(v8);
\draw[->, thick] (3.8,7)--(v7);

\draw[->, thick] (v9)--(w9);
\draw[->, thick] (v8)--(w8);
\draw[->, thick] (v7)--(w7);

\draw[->, thick] (w9)--(y9);
\draw[->, thick] (w8)--(y8);
\draw[->, thick] (w7)--(y7);

\end{tikzpicture}
\caption{The (stochastic) mapping from $\bU_{[7:9]}$ to $(\bV_{[1:6]}^{(1)},\bY_{[1:9]})$}
\label{fig:bot3}
\end{figure}

Once we obtain the matrix $M^{(t)}$ and the set $\S_{\good}$ in the code construction, the encoding procedure is standard; it is essentially the same as the original polar codes \cite{arikan-polar}.
Let $\bU_{[1:N]}$ be a random vector consisting of $N$ i.i.d. Bernoulli-$1/2$ random variables, and let $\bX_{[1:N]}=\bU_{[1:N]} M^{(t)}$.
Recall that we use $\{W_i(M^{(t)}):i\in[\ell^t]\}$ to denote the $\ell^t$ bit-channels resulting from the polar transform of $W$ using matrix $M^{(t)}$.
If we transmit the random vector $\bX_{[1:N]}$ through $N$ independent copies of $W$ and denote the channel outputs as $\bY_{[1:N]}$, then by definition, the bit-channel mapping from $U_i$ to $(\bU_{[1:i-1]},\bY_{[1:N]})$ is exactly $W_i(M^{(t)})$.
Therefore, if we use a successive  cancellation decoder to decode the input vector $\bU_{[1:N]}$ bit by bit from all the channel outputs $\bY_{[1:N]}$ and all the previous input bits $\bU_{[1:i-1]}$, then $W_i(M^{(t)})$ is the channel seen by the successive  cancellation decoder when it decodes $U_i$.
Clearly, $H(W_i(M^{(t)}))\approx 0$ means that the successive  cancellation decoder can decode $U_i$ correctly with high probability.
For every $i\in\ell^t$, we write $\tau_t(i)=(i_1,i_2,\dots,i_t)$.
In Proposition~\ref{prop:eqv} below, we will show that $H(W_i(M^{(t)}))=H_{i_1,\dots,i_t}(W)$.
Then in Proposition~\ref{prop:approx_accumulation}, we further show that $H_{i_1,\dots,i_t}(W)\approx H_{i_1,\dots,i_t}^{\bin}(W)$. Therefore, 
$H(W_i(M^{(t)}))\approx H_{i_1,\dots,i_t}^{\bin}(W)$.
By definition \eqref{eq:Sgood}, the set $\S_{\good}$ contains all the indices $(i_1,\dots,i_t)$ for which $H_{i_1,\dots,i_t}^{\bin}(W)\approx 0$, so for all $i$ such that $\tau_t(i)\in\S_{\good}$, we also have $H(W_i(M^{(t)}))\approx 0$, meaning that the successive  cancellation decoder can decode all the bits $\{U_i:\tau_t(i)\in\S_{\good}\}$ correctly with high probability.
In the encoding procedure, we put all the information in the set of good bits $\{U_i:\tau_t(i)\in\S_{\good}\}$, and we set all the other bits to be some pre-determined value, e.g., set all of them to be $0$. It is clear that the generator matrix of this code is the submatrix of $M^{(t)}$ consisting of all the row vectors with indices belonging to the set $\S_{\good}$.

\subsection{Analysis of bit-channels} \label{sect:bit}
We say that two channels $W_1:\{0,1\}\to\mathcal{Y}_1$ and $W_2:\{0,1\}\to\mathcal{Y}_2$ are equivalent if there is a one-to-one mapping $\pi$ between $\mathcal{Y}_1$ and $\mathcal{Y}_2$ such that $W_1(y_1|x)=W_2(\pi(y_1)|x)$ for all $y_1\in\mathcal{Y}_1$ and $x\in\{0,1\}$. Denote this equivalence relation as $W_1\equiv W_2$.
Then we have the following result.
\begin{prop} \label{prop:eqv}
For every $i\in\ell^t$, we write $\tau_t(i)=(i_1,i_2,\dots,i_t)$. Then we always have
$$
W_i(M^{(t)}) \equiv
W_{i_1,\dots,i_t}(K_1^{(0)},K_{i_1}^{(1)},\dots,K_{i_1,\dots,i_{t-1}}^{(t-1)}).
$$
\end{prop}
Before formally proving this proposition, we first use the special case of $t=2$ and $\ell=3$ to illustrate the main idea behind the proof.
In this case, we obtained one kernel $K_1^{(0)}$ in step $0$ and three kernels $K_1^{(1)},K_2^{(1)},K_3^{(1)}$ in step $1$.
See Fig.~\ref{fig:ill32} for an illustration of the encoding process $\bX_{[1:9]}=\bU_{[1:9]} M^{(2)}$.
 In particular, we can see that
$$
\bV_{[1:9]}^{(1)}
=\bU_{[1:9]} D^{(1)}, \quad\quad
\bU_{[1:9]}^{(1)}
=\bV_{[1:9]}^{(1)} Q^{(1)}, \quad\quad
X_{[1:9]}
=\bU_{[1:9]}^{(1)} D^{(0)} .
$$
Therefore, we indeed have $\bX_{[1:9]}=\bU_{[1:9]} D^{(1)}Q^{(1)}D^{(0)}=\bU_{[1:9]} M^{(2)}$.
Assume that $\bU_{[1:9]}$ consists of $9$ i.i.d. Bernoulli-$1/2$ random variables. Since $D^{(1)},Q^{(1)},D^{(0)}$ are all invertible matrices, the random vectors $\bV_{[1:9]}^{(1)},\bU_{[1:9]}^{(1)}$ and $\bX_{[1:9]}$ also consist of i.i.d. Bernoulli-$1/2$ random variables.

In order to analyze the bit-channels, we view Fig.~\ref{fig:ill32} from the right side to the left side. 
First observe that the following three vectors 
$$
(U_1^{(1)},U_2^{(1)},U_3^{(1)},Y_1,Y_2,Y_3),\quad\quad
(U_4^{(1)},U_5^{(1)},U_6^{(1)},Y_4,Y_5,Y_6),\quad\quad
(U_7^{(1)},U_8^{(1)},U_9^{(1)},Y_7,Y_8,Y_9)
$$
are independent and identically distributed (i.i.d.).

Given a channel $W_1:\mathcal{X}\to\mathcal{Y}$ and a pair of random variables $(X,Y)$ that take values in $\mathcal{X}$ and $\mathcal{Y}$ respectively, we write
$$
\P(X\to Y)\equiv W_1
$$
if $\P(Y=y|X=x)=W(y|x)$ for all $x\in\mathcal{X}$ and $y\in\mathcal{Y}$, where $\P(X\to Y)$ means the channel that takes $X$ as input and gives $Y$ as output.
By this definition, we have
$$
\P(U_1^{(1)}\to \bY_{[1:3]}) \equiv
\P(U_4^{(1)}\to \bY_{[4:6]}) \equiv
\P(U_7^{(1)}\to \bY_{[7:9]}) \equiv
W_1(K_1^{(0)}) .
$$
Since $V_1^{(1)}=U_1^{(1)},V_2^{(1)}=U_4^{(1)},V_3^{(1)}=U_7^{(1)}$, we also have
$$
\P(V_1^{(1)}\to \bY_{[1:3]}) \equiv
\P(V_2^{(1)}\to \bY_{[4:6]}) \equiv
\P(V_3^{(1)}\to \bY_{[7:9]}) \equiv
W_1(K_1^{(0)}) .
$$
Moreover, the following three vectors
$$
(V_1^{(1)},\bY_{[1:3]}), \quad\quad
(V_2^{(1)},\bY_{[4:6]}), \quad\quad
(V_3^{(1)},\bY_{[7:9]})
$$
are independent. Therefore, the (stochastic) mapping from $U_{[1:3]}$ to $Y_{[1:9]}$ in Fig.~\ref{fig:ill32} can be represented in a more compact form in Fig.~\ref{fig:top3}.
From Fig.~\ref{fig:top3}, we can see that
\begin{align*}
& W_1(M^{(2)}) \equiv
\P(U_1 \to \bY_{[1:9]}) \equiv
W_{1,1}(K_1^{(0)},K_1^{(1)}) ,  \\
& W_2(M^{(2)}) \equiv
\P(U_2 \to (U_1,\bY_{[1:9]})) \equiv
W_{1,2}(K_1^{(0)},K_1^{(1)}) ,  \\
& W_3(M^{(2)}) \equiv
\P(U_3 \to (U_1,U_2,\bY_{[1:9]})) \equiv
W_{1,3}(K_1^{(0)},K_1^{(1)}) .
\end{align*}

Next we investigate $W_4(M^{(2)}),W_5(M^{(2)}),W_6(M^{(2)})$. Observe that
$$
 \P(U_2^{(1)}\to(U_1^{(1)},\bY_{[1:3]} )) 
\equiv 
\P(U_5^{(1)}\to(U_4^{(1)},\bY_{[4:6]} )) 
\equiv
 \P(U_8^{(1)}\to(U_7^{(1)},\bY_{[7:9]} )) \equiv
W_2(K_1^{(0)}).
$$
Therefore, 
$$
 \P(V_4^{(1)}\to(V_1^{(1)},\bY_{[1:3]} )) 
\equiv 
\P(V_5^{(1)}\to(V_2^{(1)},\bY_{[4:6]} )) 
\equiv
 \P(V_6^{(1)}\to(V_3^{(1)},\bY_{[7:9]} )) \equiv
W_2(K_1^{(0)}).
$$
Moreover, since
$$
(V_1^{(1)},V_4^{(1)},\bY_{[1:3]} ), \quad\quad
(V_2^{(1)},V_5^{(1)},\bY_{[4:6]} ), \quad\quad
(V_3^{(1)},V_6^{(1)},\bY_{[7:9]} )
$$
are independent,
the (stochastic) mapping from $\bU_{[4:6]}$ to $(\bV_{[1:3]}^{(1)} , \bY_{[1:9]})$ in Fig.~\ref{fig:ill32} can be represented in a more compact form in Fig.~\ref{fig:mid3}.
Notice that there is a bijection between $\bU_{[1:3]}$ and $\bV_{[1:3]}^{(1)}$. Thus we can conclude from Fig.~\ref{fig:mid3} that
\begin{align*}
  & W_4(M^{(2)}) \equiv
\P(U_4 \to (\bU_{[1:3]},\bY_{[1:9]})) \equiv
\P(U_4 \to (\bV_{[1:3]}^{(1)},\bY_{[1:9]})) \equiv
W_{2,1}(K_1^{(0)},K_2^{(1)}) ,   \\
& W_5(M^{(2)}) \equiv
\P(U_5 \to (\bU_{[1:4]},\bY_{[1:9]})) \equiv
\P(U_5 \to (U_4,\bV_{[1:3]}^{(1)},\bY_{[1:9]})) \equiv
W_{2,2}(K_1^{(0)},K_2^{(1)}) ,    \\
& W_6(M^{(2)}) \equiv
\P(U_6 \to (\bU_{[1:5]},\bY_{[1:9]})) \equiv
\P(U_6 \to (U_4,U_5,\bV_{[1:3]}^{(1)},\bY_{[1:9]})) \equiv
W_{2,3}(K_1^{(0)},K_2^{(1)}) .
\end{align*}
Finally, we can use the same method to show that
\begin{align*}
& \P(V_7^{(1)}\to(V_1^{(1)},V_4^{(1)},\bY_{[1:3]} )) 
\equiv 
\P(V_8^{(1)}\to(V_2^{(1)},V_5^{(1)},\bY_{[4:6]} )) \\
\equiv
& \P(V_9^{(1)}\to(V_3^{(1)},V_6^{(1)},\bY_{[7:9]} )) \equiv
W_3(K_1^{(0)}).
\end{align*}
Therefore, the (stochastic) mapping from $\bU_{[7:9]}$ to $(\bV_{[1:6]}^{(1)} , \bY_{[1:9]})$ in Fig.~\ref{fig:ill32} can be represented in a more compact form in Fig.~\ref{fig:bot3}.
Notice that there is a bijection between $\bU_{[1:6]}$ and $\bV_{[1:6]}^{(1)}$. Thus we can conclude from Fig.~\ref{fig:bot3} that
\begin{align*}
  & W_7(M^{(2)}) \equiv
\P(U_7 \to (\bU_{[1:6]},\bY_{[1:9]})) \equiv
\P(U_7 \to (\bV_{[1:6]}^{(1)},\bY_{[1:9]})) \equiv
W_{3,1}(K_1^{(0)},K_3^{(1)}) ,   \\
& W_8(M^{(2)}) \equiv
\P(U_8 \to (\bU_{[1:7]},\bY_{[1:9]})) \equiv
\P(U_8 \to (U_7,\bV_{[1:6]}^{(1)},\bY_{[1:9]})) \equiv
W_{3,2}(K_1^{(0)},K_3^{(1)}) ,    \\
& W_9(M^{(2)}) \equiv
\P(U_9 \to (\bU_{[1:8]},\bY_{[1:9]})) \equiv
\P(U_9 \to (U_7,U_8,\bV_{[1:6]}^{(1)},\bY_{[1:9]})) \equiv
W_{3,3}(K_1^{(0)},K_3^{(1)}) .
\end{align*}
Now we have proved Proposition~\ref{prop:eqv} for the special case of $\ell=3$ and $t=2$. The proof for the general case follows the same idea, and we defer it to Appendix~\ref{app:proofeqv}.

\subsection{Complexity of code construction, encoding and decoding}
\begin{prop} \label{prop:complex}
The code construction has $N^{O_\l(1)}$ complexity. Both the encoding and successive decoding procedures have $O_{\l}(N\log N)$ complexity.
\end{prop}
\begin{proof}
The key in our proof is that we consider $\ell$ as a (possibly very large) constant.
We start with the code construction and we first show that both Algorithm~\ref{algo:kernel_search} and Algorithm~\ref{algo:bin} have $\poly(N)$ time complexity.
 In the worst case, we need to check all $2^{\ell^2}$ possible kernels in Algorithm~\ref{algo:kernel_search}, and for each kernel we need to calculate the conditional entropy of the $\ell$ subchannels. Since we always work with the quantized channel with output size upper bounded by $N^3$, each subchannel of the quantized channels has no more than $2^{\ell}N^{3\ell}$ outputs. Therefore, the conditional entropy of these subchannels can be calculated in $\poly(N)$ time, so Algorithm~\ref{algo:kernel_search} also has $\poly(N)$ complexity. After finding the good kernels, we need to use Algorithm~\ref{algo:bin} to quantize/bin the output alphabet of the subchannels produced by these good kernels. As mentioned above, the original alphabet size of these subchannels is no more than $2^{\ell}N^{3\ell}$. Therefore, Algorithm~\ref{algo:bin} also has $\poly(N)$ complexity.
 At Step $i$, we use Algorithm~\ref{algo:kernel_search} $\ell^i$ times to find  good kernels, and then we use Algorithm~\ref{algo:bin} $\ell^{i+1}$ times to quantize the bit-channels produced by these kernels, so in total we use Algorithm~\ref{algo:kernel_search} $\frac{N-1}{\ell-1}$ times and we use Algorithm~\ref{algo:bin} $\frac{\ell(N-1)}{\ell-1}$ times. Finally, finding the set $\S_{\good}$ only requires calculating the conditional entropy of the bit-channels in the last step, so this can also be done in polynomial time. Thus we conclude that the code construction has $\poly(N)$ complexity, albeit the degree in $\poly(N)$ complexity depends on $\l$.

In the encoding procedure,
we first form the vector $\bU_{[1:N]}$ by putting all the information in the bits $\{U_i:\tau_t(i)\in\S_{\good}\}$ and setting all the other bits $\{U_i:\tau_t(i)\notin\S_{\good}\}$ to be $0$. Then
we multiply $\bU_{[1:N]}$ with the encoding matrix $M^{(t)}$ and obtain the codeword $\bX_{[1:N]}=\bU_{[1:N]}M^{(t)}$.
Since the matrix $M^{(t)}$ has size $N\times N$, a naive implementation of the encoding procedure would require $O(N^2)$ operations.
Fortunately, we can use \eqref{eq:defMj} to accelerate the encoding procedure. Namely, we first multiply $\bU_{[1:N]}$ with $D^{(t-1)}$, then multiply the result with $Q^{(t-1)}$, then multiply by $D^{(t-2)}$, so on and so forth. As mentioned above, for $j=0,1,\dots,t-1$, each $D^{(j)}$ is a block diagonal matrix with $N/\ell$ blocks on the diagonal, where each block has size $\ell\times\ell$. Therefore, multiplication with $D^{(j)}$ only requires $N\ell$ operations. By definition, $Q^{(j)}, j\in[t-1]$ are permutation matrices, so multiplication with them only requires $N$ operations. In total, we multiply with $2t-1=2\log_{\ell}N-1$ matrices. Therefore, the encoding procedure can be computed in $O_{\l}(N\log N)$ time, where $O_{\l}$ means that the constant in big-$O$ depends on $\l$.

The decoding algorithm uses exactly the same idea as the algorithm in Ar{\i}kan's original paper \cite[Section~VIII-B]{arikan-polar}. Here we only use the special case of $\ell=3$ and $t=2$ in Fig.~\ref{fig:ill32} to explain how Ar{\i}kan's decoding algorithm works for large (and mixed) kernels, and we omit the proof for general parameters.
We start with the decoding of $U_1,U_2,U_3$ in Fig.~\ref{fig:ill32}. It is clear that decoding $U_1,U_2,U_3$ is equivalent to decoding $U_1^{(1)},U_4^{(1)},U_7^{(1)}$. Then the log-likelihood ratio (LLR) of each of these three bits can be calculated locally from only three output symbols. More precisely, the LLR of $U_1^{(1)}$ can be computed from $\bY_{[1:3]}$, the LLR of $U_4^{(1)}$ can be computed from $\bY_{[4:6]}$, and the LLR of $U_7^{(1)}$ can be computed from $\bY_{[7:9]}$. Therefore, the complexity of calculating each LLR only depends on the value of $\ell$. Since $\ell$ is considered as a constant, the calculation of each LLR also has constant time complexity (although the complexity is exponential in $\ell$).
The next step is to decode $\bU_{[4:6]}$ from $\bY_{[1:9]}$ together with $\bU_{[1:3]}$. This is equivalent to calculating the LLRs of $U_2^{(1)},U_5^{(1)},U_8^{(1)}$ given $\bY_{[1:9]}$ and $U_1^{(1)},U_4^{(1)},U_7^{(1)}$. This again can be done locally: To compute the LLR of $U_2^{(1)}$, we only need the values of $\bY_{[1:3]}$ and $U_1^{(1)}$; to compute the LLR of $U_5^{(1)}$, we only need the values of $\bY_{[4:6]}$ and $U_4^{(1)}$; to compute the LLR of $U_8^{(1)}$, we only need the values of $\bY_{[7:9]}$ and $U_7^{(1)}$. Finally, the decoding of $\bU_{[7:9]}$ from $\bY_{[1:9]}$ and $\bU_{[1:6]}$ can be decomposed into local computations in a similar way. 
Using this idea, one can show that for general values of $\ell$ and $t$, the decoding can also be decomposed into $t=\log_{\ell}N$ stages, and in each stage, the decoding can further be decomposed into $N/\ell$ local tasks, each of which has constant time complexity (although the complexity is exponential in $\ell$). Therefore, the decoding complexity at each stage is $O_{\l}(N)$ and the overall decoding complexity is $O_{\l}(N\log N)$.
As a final remark, we mention that after calculating the LLRs of all $U_i$'s, we will only use the LLRs of the bits $\{U_i:\tau_t(i)\in\S_{\good}\}$. For these bits, we decode $U_i$ as $0$ if its LLR is larger than $0$ and decode it $1$ otherwise. Recall that in the encoding procedure, we have set all the other bits $\{U_i:\tau_t(i)\notin\S_{\good}\}$ to be $0$, so for these bits we simply decode them as $0$.
\end{proof}

\subsection{Code rate and decoding error probability}
In \eqref{eq:defHi}, we have defined the conditional entropy for all the bit-channels obtained in the last step (Step $t-1$). Here we also define the conditional entropy for the bit-channels obtained in the previous steps. More precisely, for every $j\in[t]$ and every $(i_1,i_2,\dots,i_j)\in[\ell]^j$, we use the following short-hand notation:
    \begin{align*}
    H_{i_1,\dots,i_j}(W) &:=
    H(W_{i_1,\dots,i_j}(K_1^{(0)},K_{i_1}^{(1)},\dots,K_{i_1,\dots,i_{j-1}}^{(j-1)}))   \\
    H_{i_1,\dots,i_j}^{\bin}(W) &:=
    H(W_{i_1,\dots,i_j}^{\bin}(K_1^{(0)},K_{i_1}^{(1)},\dots,K_{i_1,\dots,i_{j-1}}^{(j-1)}))   \\
    H_{i_1,\dots,i_j}^{\bin*}(W) &:=
    H(W_{i_1,\dots,i_j}^{\bin*}(K_1^{(0)},K_{i_1}^{(1)},\dots,K_{i_1,\dots,i_{j-1}}^{(j-1)}))  .
    \end{align*}
According to \eqref{eq:ttt}, we have
\begin{equation}  \label{eq:gpH}
H_{i_1,\dots,i_j}^{\bin*}(W) \le H_{i_1,\dots,i_j}^{\bin}(W)
\le H_{i_1,\dots,i_j}^{\bin*}(W) + \frac{6\log N}{N^3}
\end{equation}
for every $j\in[t]$ and every $(i_1,i_2,\dots,i_j)\in[\ell]^j$.

\begin{prop}
\label{prop:approx_accumulation}
For every $j\in[t]$ and $(i_1,i_2,\dots,i_j)\in[\ell]^j$, the conditional entropy $H_{i_1,\dots,i_j}(W)$ and $H_{i_1,\dots,i_j}^{\bin}(W)$ satisfy the following inequality
\begin{equation} \label{eq:obg}
H_{i_1,\dots,i_j}(W) \le
    H_{i_1,\dots,i_j}^{\bin}(W) \le 
    H_{i_1,\dots,i_j}(W) + \frac{6\ell \log N}{N^2}
\end{equation}
\end{prop}
\begin{proof}

Since the binning algorithm (Algorithm~\ref{algo:bin}) always produces a channel that is degraded with respect to the original channel,
the first inequality in \eqref{eq:obg} follows immediately by applying Proposition~\ref{prop:degrad_subchannel} recursively in our $t$-step code construction.

Now we prove the second inequality in \eqref{eq:obg}.
We will prove the following inequality by induction on $j$:
\begin{equation}  \label{eq:glg}
H_{i_1,\dots,i_j}^{\bin}(W) \le H_{i_1,\dots,i_j}(W)
+ \frac{6\log N}{N^3}(1+\ell+\ell^2+\dots+\ell^j)
\quad\quad \forall (i_1,i_2,\dots,i_j)\in[\ell]^j .
\end{equation}
The base case of $j=0$ is trivial. Now assume that this inequality holds for $j$ and we prove it for $j+1$.
By chain rule, we know that
$$
\sum_{i_{j+1}=1}^{\ell}H_{i_1,\dots,i_j,i_{j+1}}^{\bin*}(W)
=\ell H_{i_1,\dots,i_j}^{\bin}(W), 
\quad\quad
\sum_{i_{j+1}=1}^{\ell}H_{i_1,\dots,i_j,i_{j+1}}(W)
=\ell H_{i_1,\dots,i_j}(W).
$$
Therefore, 
$$
\sum_{i_{j+1}=1}^{\ell}\Big(H_{i_1,\dots,i_j,i_{j+1}}^{\bin*}(W) -H_{i_1,\dots,i_j,i_{j+1}}(W)  \Big)
=\ell \Big( H_{i_1,\dots,i_j}^{\bin}(W)
- H_{i_1,\dots,i_j}(W) \Big).
$$
Since every summand on the left-hand side is non-negative, we have
$$
H_{i_1,\dots,i_j,i_{j+1}}^{\bin*}(W) -H_{i_1,\dots,i_j,i_{j+1}}(W)  
\le \ell \Big( H_{i_1,\dots,i_j}^{\bin}(W)
- H_{i_1,\dots,i_j}(W) \Big)
\le \frac{6\log N}{N^3} (\ell+\ell^2+\dots+\ell^{j+1}),
$$
where the second inequality follows from the induction hypothesis.
Combining this with \eqref{eq:gpH}, we obtain that 
$$
H_{i_1,\dots,i_j,i_{j+1}}^{\bin}(W) \le H_{i_1,\dots,i_j,i_{j+1}}(W)
+ \frac{6\log N}{N^3}(1+\ell+\ell^2+\dots+\ell^{j+1}) .
$$
This establishes the inductive step and completes the proof of \eqref{eq:glg}.
The inequality \eqref{eq:obg} then follows directly from \eqref{eq:glg} by using the fact that $1+\ell+\dots+\ell^j<\ell N$ for all $j\le t$.
\end{proof}

Recall that in Remark~\ref{rmk:rmk} we denoted by $\l\ge\exp(\Omega(\a^{-1.01}))$ the conditions on $\ell$ to be large enough so that $\log\l \geq \frac{11}{\a}$ and $\frac{\log\l}{\log\log\l + 2} \geq \frac{3}{\a}$. In the theorems below, even though the statements hold for any $\a \in (0, 1/12)$, we modify the intervals of $\a$ so that the rate appears positive in the formulations. This is also why in the formulation of the Theorem~\ref{thm:intro-main} we take $\a$ from $(0, 1/36)$.

We now can formulate

\begin{thm} \label{thm:m1}
For arbitrarily small $\a \in \left(0, \frac1{14}\right)$, if we choose a large enough constant $\l\ge\exp(\Omega(\a^{-1.01}))$ to be a power of $2$ and let $t=\log_{\l} N$ grow, then
the codes constructed from the above procedure have decoding error probability $O_{\a}(\log N/N)$ under successive decoding and code rate $I(W)-N^{-1/2+7\a}$, where $N=\l^t$ is the code length. 
\end{thm}
\begin{proof}
By \eqref{eq:obg} and the definition of $\S_{\good}$ in \eqref{eq:Sgood}, we know that for every $(i_1,\dots,i_t)\in\S_{\good}$, we have $H_{i_1,\dots,i_t}(W)\le H_{i_1,\dots,i_t}^{\bin}(W)\le \frac{7\ell \log N}{N^2}$.
Then by Lemma 2.2 in \cite{Blasiok18}, we know that the ML decoding error probability of the bit-channel $W_{i_1,\dots,i_t}(K_1^{(0)},K_{i_1}^{(1)},\dots,K_{i_1,\dots,i_{t-1}}^{(t-1)})$ is also upper bounded by $\frac{7\ell \log N}{N^2}$. Since the cardinality of $\S_{\good}$ is at most $N$, we can conclude that the overall decoding error probability under the successive  cancellation decoder is $O_{\a}(\log N/N)$ using the union bound.

Notice that $|\S_{\good}|$ is the code dimension. Therefore, we only need to lower bound $|\S_{\good}|$ in order to get the lower bound on the code rate.
Define another set 
\begin{equation}\label{eq:sgprime}
 \S_{\good}':=\left\{(i_1,i_2,\dots,i_t)\in[\ell]^t: H_{i_1,\dots,i_t}(W) \le \frac{\ell \log N}{N^2} \right\}.
\end{equation}
According to \eqref{eq:obg}, if $H_{i_1,\dots,i_t}(W)<\frac{\ell\log N}{N^2}$, then
$H_{i_1,\dots,i_t}^{\bin}(W)\le \frac{7\ell \log N}{N^2}$.
Therefore, $\S_{\good}'\subseteq \S_{\good}$, so $|\S_{\good}|\ge |\S_{\good}'|$.
In Lemma~\ref{lm:acd} below, we will prove that $|\S_{\good}'|\ge N(I(W)-N^{-1/2+7\alpha})$. Therefore, $|\S_{\good}|\ge N(I(W)-N^{-1/2+7\alpha})$.
This completes the proof of the theorem.
\end{proof}

\begin{lem}  \label{lm:acd}
If $\a \in \left(0, \frac1{14}\right)$ and $\l$ is large enough so that $\log\l \geq \frac{11}{\a}$ and $\frac{\log\l}{\log\log\l + 2} \geq \frac{3}{\a}$, then the set $\S_{\good}'$ defined in \eqref{eq:sgprime} satisfies the following inequality
$$
\left\lvert\S_{\good}'\right\lvert 
\geq N\left(I(W) - N^{-\frac12 + 7\a } \right) .
$$
\end{lem}

\begin{proof}
The proof is the same as in~\cite[Claim A.2]{Blasiok18}. Recall that we proved in~\eqref{eq:strong_Markov}--\eqref{eq:Ega_exponential}
$$
    \P\left[H^{(t)} \in \left(\frac{\l\log N}{N^2}, 1 - \frac{\l\log N}{N^2} \right)\right] \leq 2 \frac{N^{2\a}}{(\l\log N)^{\a}}\cdot\la^t,
$$
where $H^{(t)}$ is (marginally) the entropy of the random channel at the last level of construction, i.e. $H^{(t)}$ is uniformly distributed over  $H_{i_1,\dots,i_t}(W)$ for all possible $(i_1,i_2,\dots,i_t)\in[\ell]^t$, and $\la$ is such that~\eqref{mult_decrease} holds for any channel $W'$ throughout the construction. 
By Proposition~\ref{prop:approx_accumulation}, we can choose the error parameter $\Delta$ in Algorithm~\ref{algo:kernel_search} to be $\Delta=\frac{6\l\log N}{N^2}$, which satisfies the condition $\Delta\le \l^{-\log \l}$ in Theorem~\ref{thm:kernel_seacrh_correct}.
Then Theorem~\ref{thm:kernel_seacrh_correct} and Remark~\ref{rmk:rmk} tell us that as long as the conditions on $\l$ and $\a$ specified in this lemma hold, Algorithm~\ref{algo:kernel_search} allows us to choose kernels such that $\la\leq \l^{-1/2+5\a}$, which gives
\begin{equation} \label{eq:wus}
    \P\left[H^{(t)} \in \left(\frac{\l\log N}{N^2}, 1 - \frac{\l\log N}{N^2} \right)\right] \leq \frac{2 N^{-1/2 + 7\a}}{(\l\log N)^{\a}}.
\end{equation}
On the other hand, conservation of entropy throughout the process implies $E\left[H^{(t)}\right] = H(W)$, therefore by Markov's inequality
\begin{equation}
    \P\left[H^{(t)} \geq 1 - \frac{\l\log N}{N^2} \right] \leq \frac{H(W)}{1 - \frac{\l\log N}{N^2}} \leq H(W) + \frac{2\l\log N}{N^2}.
\end{equation}
Since $H(W) = 1 - I(W)$ for symmetric channels and $\left\lvert\S_{\good}'\right\lvert = N\cdot\P\left[H^{(t)} \leq \frac{\l\log N}{N^2} \right]$, we have
\begin{align*}
   \left\lvert\S_{\good}'\right\lvert &\geq N\left(1 - \frac{2 N^{-1/2 + 7\a}}{(\l\log N)^{\a}} - H(W) - \frac{2\l\log N}{N^2} \right) \\
   &\geq N\left(I(W) - \frac{3 N^{-1/2 + 7\a}}{(\l\log N)^{\a}}\right) \\
   &\geq N\Big(I(W) - N^{-1/2 + 7\a}\Big). \qedhere
\end{align*}
\end{proof}

\subsection{Main theorem: Putting everything together}
\label{sect:main_together}
As we mentioned at the beginning of this section, the code construction presented above only takes the special case of $\Q=N^3$ as a concrete example, where $\Q$ is the upper bound on the output alphabet size after binning; see Algorithm~\ref{algo:bin}.
In fact, we can change the value of $\Q$ to be any polynomial of $N$, and this will allow us to obtain a trade-off between the decoding error probability and the gap to capacity while maintaining the polynomial-time code construction as well as the $O_{\a}(N\log N)$ encoding and decoding complexity.
More precisely, we have the following theorem.

\begin{thm} \label{thm:main1}
For any BMS channel $W$, any $c>0$ and arbitrarily small $\a \in \left(0, \frac1{12 + 2c}\right)$, if we choose a large constant $\l$ to be a power of $2$ which satisfies $\log\l \geq \frac{11}{\a}$ and $\frac{\log\l}{\log\log\l + 2} \geq \frac{3}{\a}$, and set $\Q=N^{c+2}$ in the above code construction procedure, then
we can construct a code $\mathcal{C}$ with code length $N=\l^t$ such that the following four properties hold when $t$ grows: (1) the code construction has $N^{O_\a(1)}$ complexity; (2) both encoding and decoding have $O_{\a}(N\log N)$ complexity; (3) rate of $\mathcal{C}$ is $I(W)-O(N^{-1/2+(c+6)\a})$; (4) decoding error probability of $\mathcal{C}$ is $O_{\a}(\log N/N^c)$ under successive  decoding when $\mathcal{C}$ is used for channel coding over $W$.
\end{thm}
\begin{proof}
The proof of properties (1) and (2) is exactly the same as Proposition~\ref{prop:complex}.
Here we only briefly explain how to adjust the proof of Theorem~\ref{thm:m1} to show properties (3) and (4).
First, we change the definitions of $\S_{\good}$ and $\S_{\good}'$ to
\begin{align*}
 \S_{\good} & :=\left\{(i_1,i_2,\dots,i_t)\in[\ell]^t: H_{i_1,\dots,i_t}^{\bin}(W) \le \frac{(2c+3)\ell \log N}{N^{c+1}} \right\},   \\
     \S_{\good}' & :=\left\{(i_1,i_2,\dots,i_t)\in[\ell]^t: H_{i_1,\dots,i_t}(W) \le \frac{\ell \log N}{N^{c+1}} \right\}.
\end{align*}
The definition of $\S_{\good}$ immediately implies property (4). Next we prove property (3). 
Since we change $\Q$ from $N^3$ to $N^{c+2}$,
inequality \eqref{eq:gpH} becomes
$$
H_{i_1,\dots,i_j}^{\bin*}(W) \le H_{i_1,\dots,i_j}^{\bin}(W)
\le H_{i_1,\dots,i_j}^{\bin*}(W) + \frac{2(c+2)\log N}{N^{c+2}} .
$$
As a consequence, inequality \eqref{eq:obg} in Proposition~\ref{prop:approx_accumulation} becomes
$$
H_{i_1,\dots,i_j}(W) \le
    H_{i_1,\dots,i_j}^{\bin}(W) \le 
    H_{i_1,\dots,i_j}(W) + \frac{2(c+2)\l\log N}{N^{c+1}} .
$$
This inequality tells us that $\S_{\good}'\subseteq \S_{\good}$, so $|\S_{\good}|\ge |\S_{\good}'|$.
Then we follow Lemma~\ref{lm:acd} to lower bound $|\S_{\good}'|$.
Inequality \eqref{eq:wus} now becomes
$$
    \P\left[H^{(t)} \in \left(\frac{\l\log N}{N^{c+1}}, 1 - \frac{\l\log N}{N^{c+1}} \right)\right] \leq \frac{2 N^{-1/2 + (c+6)\a}}{(\l\log N)^{\a}}.
$$
Therefore, we obtain that 
$$
|\S_{\good}|\ge |\S_{\good}'|
\ge N\Big(I(W)-N^{-1/2 + (c+6)\a} \Big).
$$
This completes the proof of the theorem.
\end{proof}


\section{Inverse sub-exponential decoding error probability}
\label{sec:exponential-decoding}
In this section we finish proving our main result (Theorem~\ref{thm:intro-main}), by showing how to obtain inverse sub-exponential $\exp(-N^{\a})$ probability of error decoding within our construction of polar codes, while still having $\poly(N)$ time complexity of construction.
Note that up to this point we only claimed inverse polynomial decoding error probability in Theorem~\ref{thm:main1}. This restriction came from the fact that we need to approximate the channels we see in the tree during the construction phase (recall the discussion at the beginning of Sections~\ref{sect:local} and~\ref{sect:cons}), and to get a polynomial-time construction we need the binning parameter $\Q$ to be $\poly(N)$ itself. But this means that we are only able to track the parameters (entropies, for instance) of the bit-channels approximately, with an additive error which is inverse polynomial in $N$, see~\eqref{eq:gpH}. Since the decoding error probability relates directly to the upper bound on the entropies of the ``good" bit-channels we choose, this leads to only being able to claim inverse polynomial decoding error probability.

It was proved in a recent work~\cite{Wang-Duursma} that it is possible to achieve a fast scaling of polar codes (good scaling exponent) and good decoding error probability (inverse sub-exponential instead of inverse polynomial in $N$) simultaneously, using the idea of multiple (dynamic) kernels in the construction. Specifically, for any constants $\pi, \mu > 0$ such that $\pi + 2\mu < 1$, it is shown that one can construct a polar code with rate $N^{-\mu}$  close to capacity of the channel (which corresponds to scaling exponent $\mu$) and decoding error probability $\exp(-N^{\pi})$, as $N \to \infty$. Moreover, it is shown that this is an optimal scaling of these two parameters one can obtain for \emph{any} (not just polar) codes. However, the construction phase in~\cite{Wang-Duursma} tracked the \emph{true} bit-channels that are obtained in the $\l$-ary tree of channels, which makes the construction intractable. This is because (most of) the true bit-channels cannot even be described in a tractable way, since they have exponential size of output alphabet.

In what follows we combine our approach of using Ar\i kan's kernels for polarized bit-channels with a stronger analysis of polarization from~\cite{Wang-Duursma} to overcome this issue of intractable construction. Specifically, we show that even though we only track \emph{approximations} (binned versions) of the bit-channels in the tree, if we use Ar\i kan's channels for suction at the end regime, then we are still able to prove very strong polarization, as in~\cite{Wang-Duursma}. This comes from the fact that we know very well how Ar\i kan's basic $2\times 2$ kernel evolves the parameters of the bit-channels. 
This allows us to get very strong bounds on the parameters of the \emph{true} bit-channels (which leads to good decoding error probability), while still only tracking their \emph{approximations} (which keeps the construction time polynomial). Somewhat surprisingly, the phase of the construction where the local kernels are chosen is exactly the same as it was before in Section~\ref{sect:cons}, and the difference lies in a much tighter analysis of how to choose a set of ``good" indices to actually construct a polar code. 

\subsection*{Notations}
We fix a small positive parameter $\a > 0$ from the statement of Theorem~\ref{thm:intro-main}, which corresponds to how close the scaling exponent will be to $1/2$. Specifically, we will have the scaling exponent $\mu = 2 + O(\a)$. As before, the size of the kernel is denoted by $\l = 2^s$, where $\l$ is large enough in terms of $\a$ (specifically, the bounds from the statement of the Theorem~\ref{thm:kernel_seacrh_correct} must hold). 

We are going to work with the complete $\l$-ary tree of bit-channels, as described in Section~\ref{sect:mix}. Let $t$ be the depth of this tree, then there are $N = \l^t$ bit-channels at the last level, denoted as $W_i$ for $i\in [\l^t]$ (these notations depend on the depth $t$ of the tree at which we are looking, but it will always be clear from the context). Throughout this section we will denote such a tree of depth $t$ as $\T_t$.

We will again have a random process of going down the tree, starting from the root, and picking a random child of a current bit-channel at each step. To be more precise, the random process $\WW_i$ is defined as follows: $\WW_0 = W$ (the initial channel, i.e. the root of the tree), and $\WW_{j+1} = \left(\WW_j\right)_k$, where $k\sim [\l]$, and $\left(\WW_j\right)_k$ is the $k^{\text{th}}$ Ar\i kan's bit-channel of $\WW_j$ with respect to the corresponding kernel in the tree. This indeed is equivalent to a random walk down the tree. Then we also define the random processes $\ZZ_j = Z(\WW_j)$ and $\HH_j = H(\WW_j)$. Note that $\WW_t$ marginally is distributed as $W_i$ for $i \sim [N]$, where $N = \l^t$, i.e. $\WW_t$ is just a random bit-channel at the level $t$ of the tree. Further, we will also look at random processes $\WWb_j, \HHb_j, \ZZb_j$, which mean that we also do the binning procedure as described in the construction phase in Section~\ref{sect:cons}. Note that $\WWb_j$ are the channels that we actually track during the construction of the code, while $\WW_j$ are the \emph{true} bit-channels in the tree.

Finally, by $\exp(\bullet)$ we will denote $2^{\bullet}$ in this section, and we denote by $x^+ = \max\{x, 0\}$ the positive part of $x$.

\subsection*{Plan}

First, notice that building the tree $\T_t$ of bit-channels is itself a part of construction of our polar codes. This includes tracking the binned versions of the bit-channels, and picking the kernels using Algorithm~\ref{algo:kernel_search}. This part will stay exactly the same as it is described in Section~\ref{sect:cons}, with the binning parameter $\Q = N^3$, and the same threshold of $\l^{-4}$ in the Algorithm~\ref{algo:kernel_search}. The only part of the construction that is going to change is how we pick the set of good indices which we use to transmit information.

We will closely follow the analysis from~\cite[Appendices~B,\ C]{Wang-Duursma} (also appearing in~\cite{Wang18}), modified for our purposes. Specifically, we will prove the needed polarization of the construction presented in Section~\ref{sect:cons} in three steps (recall that $s = \log_2 \ell$):
\begin{enumerate}[label=\arabic*)]
    \item $\P\bigg[\ZZ_t \leq \exp(-2st)\bigg] \geq I(W) - \l^{-(1/2 - 10\a)t},$
    \item $\P\bigg[\ZZ_t \leq \exp\left(-2^{t^{1/3}}\right) \bigg] \geq I(W) - \l^{-(1/2 - 11\a)t + \sqrt{t}},$
    \item $\P\bigg[\ZZ_t \leq \exp\left(-st\cdot\l^{\a\cdot t}\right)\bigg] \geq I(W) - \l^{-(1/2 - 16\a)t + 2\sqrt{t}}$ \hspace{4pt} for $t = \O(\log^6 s)$.
\end{enumerate}
Moreover, for each step, we prove that the polarization at each step is \emph{poly-time constructible}:
\begin{defin}
We call the polarization $\P[\ZZ_t \leq p(t)] \geq R(t)$ to be \emph{poly-time constructible} if one can find at least $N\cdot R(t)$ indexes $i \in [N]$ such that $Z(W_i) \leq p(t)$, where $N = \l^t$, in time polynomial in $N$.
\end{defin}

Notice that if polarization $\P[\ZZ_t \leq p(t)] \geq R(t)$ is poly-time constructible, then by choosing these $N\cdot R(t)$ indexes as information bits of the code, a standard argument implies that one obtains a polar code of rate $R(t)$ and decoding error probability at most $N\cdot p(t)$. Moreover, since the indexes of the information bits were found in $\text{poly}(N)$ time, this makes the whole code construction complexity polynomial in $N$.

The polarization behavior from Step 3 with $t\geq \frac1{\a^2}$ will then correspond to polar codes with rate $I(W) - N^{-1/2 + 18\a}$ (i.e. codes with scaling exponent $(2+O(\a))$ and sub-exponentially small decoding error probability $N\cdot\exp\left(-st\cdot\l^{\a\cdot t}\right) = \exp(-N^{\a})$, with $\poly(N)$ construction time, which finishes the proof of the main result of this paper.

\subsection{Step 1}

\begin{lem}
\label{lem:step1}
$\P\bigg[\ZZ_t \leq \exp(-2st)\bigg] \geq I(W) - \l^{-(1/2 - 10\a)t}.$ Moreover, this polarization is poly-time constructible.
\end{lem}
\begin{proof}
This follows from the analysis of the construction we already have in the previous sections. Fix some $t$ and let $N = \l^t$. Then the following is implied from Section~\ref{sect:main_together} if one takes $\Q = N^3$, i.e. $c=3$:
\[ \P_{i \sim [N]}\bigg[H(W_i^{\bin}) \leq \dfrac1{N^4}\bigg] \geq I(W) - N^{-(1/2 - 10\a)}.\]

Note here that $H(W_i^{\bin})$ are the entropies of the binned bit-channels that we are actually tracking during the construction phase, so they are computable in polynomial time. This means that there is $\poly(N)$-time procedure which returns all the indices $i$ for which $H(W_i^{\bin}) \leq \frac1{N^4}$. Then $Z(W^{\bin}_i) < \sqrt{H(W_i^{\bin})} \leq \frac1{N^2}$ for these indices, so we have for the random process $\ZZ^{\bin}_t$:
\[ \P\Big[\ZZ^{\bin}_t \leq \l^{-2t} \Big] = \P\Big[\ZZ^{\bin}_t \leq 2^{-2st} \Big] = \P\Big[\ZZ^{\bin}_t \leq \exp\left(-2st\right) \Big] \geq I(W) - N^{-(1/2 - 10\a)},\]
and moreover, one can find at least $N(I(W) - N^{-(1/2 - 10\a)})$ indexes within $i \in [N]$ for which the inequality $Z(W^{\bin}_i) \leq \exp\left(-2st\right)$ holds in $\poly(N)$ time (just by returning the indices for which $H(W_i^{\bin}) \leq \frac1{N^4}$). Since it always holds $\ZZ_t \le \ZZb_t$, the statement of the lemma follows.
\end{proof}

\subsection{Step 2}
Next, we are going to strengthen the polarization of the construction, using the result of Lemma~\ref{lem:step1}. Specifically, we prove
\begin{lem}
\label{lem:step2}
$\P\bigg[\ZZ_n \leq \exp\left(-2^{n^{1/3}}\right) \bigg] \geq I(W) - \l^{-(1/2 - 11\a)n + \sqrt{n}}$. Moreover, this polarization is poly-time constructible.
\end{lem}
\begin{proof}
For this lemma, we fix $n$ to be the total depth of the tree (instead of $t$), and we want to prove the speed of polarization at level $n$.
To do this, we will divide the tree into $\sqrt{n}$ stages, each of depth $\sqrt{n}$, and apply the polarization we obtained at Step 1 at each stage. So, we look at $m$ being $\sqrt{n}$, $2\sqrt{n}, \dots, n-\sqrt{n}$.  Define the following events, starting with $E_0^{(0)} = \emptyset$ (again, closely following~\cite{Wang-Duursma}):
\begin{align}
    &A_m = \left\{\ZZ^{\bin}_m < \exp(-2sm)\right\}\setminus E_0^{(m-\sqrt{n})} \\
    &B_m = A_m \bigcap \left\{\sum_{i=1}^{s\sqrt{n}}g_{sm+i} \leq \beta\cdot s\sqrt{n}\right\}\\
    &E_m = A_m \setminus B_m\\
    &E_0^{(m)} = E_0^{(m-\sqrt{n})}\cup E_m,
\end{align}
where for now
one can think of $g_j$'s as of independent $\text{Bern}(1/2)$ random variables for all $j \in [s\cdot n]$. In the following several paragraphs we explain what these events are going to correspond to. First of all, the actual random variable we are tracking here is $\WW_n$, and its realizations are $\l^n$ bit-channels $W_i$ for $i\in[\l^n]$ at the last level of the tree. We can then think of events and subsets of bit-channels at level $n$ interchangeably. 

Notice that each bit-channel $W_i$ for $i\in [\l^n]$ corresponds to a unique path in the tree $\T_n$ from the root $W$ (the initial channel) to the leaf $W_i$ on the $n^{\text{th}}$ level. We will be interested in the bit-channels on these path, their binned versions, and the parameters of both versions (true and binned) of these channels during the ensuing arguments. We denote this path of true bit-channels as $W_i^{(0)} = W, W_i^{(1)}, \dots, W_i^{(n-1)}, W_i^{(n)} = W_i$. Clearly, this path is just a realization of a random walk $\WW_0, \WW_1, \dots, \WW_n$, when $\WW_n$ ends up being $W_i$. In the same way, we will denote by $W_i^{(k),\bin}$, for $k = 0, 1,\dots, n$ the binned version of the bit-channel along this path, and by $H_i^{(k)}$, $H_i^{(k),\bin}$, $Z_i^{(k)}$, and $Z_i^{(k),\bin}$ the corresponding parameters of these channels.

We are going to construct a set of ``good" bit-channels $E_0^{(n-\sqrt{n})}$ incrementally, by inspecting the tree from top to bottom. We start with the set $E_0^{(0)} = \emptyset$. Then, at each stage $m = \sqrt{n}$, $2\sqrt{n}, \dots, n-\sqrt{n}$, we find a set $E_m$ of bit-channels which we mark to be ``good" at level $m$. Precisely, the channel $W_i$, for some $i\in [\l^n]$, is going to be in $E_m$, if: a) it is not marked as good before that (i.e. it is not in $E_0^{(m-\sqrt{n})}$); b) the Bhattacharyya parameter $Z_i^{(m), \bin}$ is small, specifically smaller then $\exp(-2sm)$; and c) a certain condition holds for how the branches are chosen in the path for $W_i$ between levels $m$ and $m + \sqrt{n}$ in the tree (more details on this later). Here conditions a) and b) correspond together to the event $A_m$, while condition c) further defines the event $B_m$. Then the set $E_0^{(m)}$ will be the set of all bit-channels that we marked to be good up to the level $m$ in the tree, and in the end, by collecting all the bit-channels that we marked as good at the stages $m = \sqrt{n}, 2\sqrt{n}, \dots, n-\sqrt{n}$, we obtain the final set $E_0^{(n-\sqrt{n})}$.


Denote by corresponding lowercase letters the probabilities of the events described before, i.e. $a_m \coloneqq \P[A_m]$, etc.. Finally, let $q_m = I(W) - e_0^{(m)}$, i.e. $q_m$ is the gap between the capacity and the fraction of the channels which we marked as ``good" up to level $m$.


\medskip
To begin the formal analysis, let us first consider what happens in case of the event $A_m$. First, it means that $\ZZ^{\bin}_m < \exp(-2sm)$. But then we know that we are going to apply Ar\i kan's kernel $A_2^{\otimes s}$ to this bit-channel at level $m$, since the threshold for picking Ar\i kan's kernel in Algorithm~\ref{algo:kernel_search}, which we use in the construction phaze, is $\l^{-4} = \exp(-4s)$. This means that, conditioned on $A_m$, we have $\ZZ_{m+1} \leq  \ZZ_m\cdot 2^s \leq \ZZ^{\bin}_m\cdot 2^s < 2^s\cdot\exp(-2sm)$, where the first inequality follows from that we know how Bhattacharrya parameter evolves when we use basic Ar\i kan's transforms. Precisely, using the kernel $A_2^{\otimes s}$ is equivalent to using the basic $2\times 2$ kernel $A_2$ for $s$ times, and the kernel $A_2$ in the worst case doubles the Bhattacharyya parameter. Thus $s$ applications of $A_2$ can increase the Bhattacharyya parameter by at most a factor of $2^s$.

Then it is easy to see that even after we apply Ar\i kan's kernel  $A_2^{\otimes s}$ a total of $\sqrt{n}$ times, the Bhattacharyya parameter will still be below the threshold $\l^{-4}$: conditioned on $A_m$, one has $\ZZ_{m+\sqrt{n}} \leq \ZZ_m \cdot \left(2^s\right)^{\sqrt{n}} < \exp(-2sm)\cdot \exp(s\sqrt{n}) < \exp(-sm) < \l^{-4}$, as $m \geq \sqrt{n}$. It is easy to verify, using Proposition~\ref{prop:approx_accumulation} and the relation~\eqref{eq:Z-H} between the entropy and Bhattacharyya parameter of the bit-channel, that the binned parameter $\HHb_{m+j}$ will also be below $\l^{-4}$ for $j = 1, 2, \dots, \sqrt{n}$. This means that indeed for these $\sqrt{n}$ levels, the Ar\i kan's kernel was taken in the construction phase. Therefore, we know that only the kernel $A_2^{\otimes s}$ was applied at levels between $m$ and $m+\sqrt{n}$, which can also be viewed as applying the basic $2\times 2$ kernel $A_2$ for $s\sqrt{n}$ levels in the tree. Further this can be viewed as taking $s\sqrt{n}$ ``good" or ``bad" branches while going down the tree, where the good branch corresponds to squaring the Bhattachryya parameter, and the bad branch at most doubles it. Denote then by bits $g_{sm+i} \in \{0,1\}$, for $i \in [s\sqrt{n}]$, the indicators of these branches being good or bad, where $g_{sm+i} = 0$ means the branch is bad, and $g_{sm+i}=1$ means the branch is good. It is clear then that since we consider the random process of going down the tree choosing the next child randomly, then all $g_{sm+i}$'s are independent $\text{Bern}(1/2)$ random variables. These are exactly the random variables appearing in the definition of $B_m$.

Notice then that 
\[ \frac{b_m}{a_m} = \P\left[\sum_{i=1}^{s\sqrt{n}}g_{sm+i} \leq \beta\cdot s\sqrt{n}\right] \leq 2^{-s\sqrt{n}(1-h_2(\beta))} \leq 2^{-\gamma s\sqrt{n}} \ , \]
where we can take, for instance, $\beta = 1/20$ and $\gamma = 0.85$. The inequality follows from entropic bound on the sum of binomial coefficients (one could also just use the Chernoff bound).

Recall that we defined $q_m = I(W) - e_0^{(m)}$. We then can write $q_{m-\sqrt{n}} - a_m = I(W) - (e_0^{(m-\sqrt{n})} + a_m)$. But note that by definition, the event $\left\{\ZZb_m < \exp(-2sm)\right\}$ is a subevent of $A_m \cup E_0^{(m-\sqrt{n})}$, and thus using the bound from Lemma~\ref{lem:step1} (applied for the depth $m$) we know that
\[ (e_0^{(m-\sqrt{n})} + a_m) \geq \P[A_m \cup E_0^{(m-\sqrt{n})}] \geq \P[\ZZb_m < \exp(-2sm)] \geq  I(W) - 2^{(-1/2 + 10\a)sm} \ . \]
Therefore we conclude \[(q_{m-\sqrt{n}} - a_m)^+ \leq 2^{(-1/2 + 10\a)sm}.\]
We can then derive
\begin{align} q_m &= I(W) - e_0^{(m)} = I(W) - (e_0^{(m-\sqrt{n})} + e_m) = q_{m-\sqrt{n}} - e_m\\
&=q_{m-\sqrt{n}}\left(1 - \frac{e_m}{a_m}\right) + \frac{e_m}{a_m}(q_{m-\sqrt{n}} - a_m)\\
&\leq q^+_{m-\sqrt{n}}\cdot\frac{b_m}{a_m} + (q_{m-\sqrt{n}} -a_m)^+\\
&\leq q_{m-\sqrt{n}}^+\cdot 2^{-\gamma s\sqrt{n}} + 2^{(-1/2+10\a)sm}.
\end{align}
Thus we end up we the following recurrence on $q^+_m$ (recall that $\l = 2^s$):
\begin{align}
    q^+_{\sqrt{n}} &\leq 1 \\
    q^+_m &\leq q^+_{m-\sqrt{n}}\cdot\l^{-\gamma\sqrt{n}} + \l^{-\frac m2+10\a m}.
\end{align}
Solving this recurrence gives us $q^+_{n-\sqrt{n}} \leq \l^{-\frac{n}{2} + 11\a n + \sqrt{n}}$, since $\gamma > 1/2$. Therefore we can conclude
\begin{equation}
\label{eq:step2e0n}
    e_0^{(n-\sqrt{n})} \geq I(W) -  \l^{-\frac{n}{2} + 11\a n + \sqrt{n}}.
\end{equation}
Next, let us look at an arbitrary bit-channel (realization of $\ZZ_n$) for which the event $E_0^{(n-\sqrt{n})}$ happens, and prove that such a bit-channel is indeed ``good." Since $E_0^{(n-\sqrt{n})}$ happened, it means that $E_m$ happened at some stage, thus $\ZZ^{\bin}_m < \exp(-2sm)$ and $\sum_{i=1}^{s\sqrt{n}}g_{sm+i} \geq \beta\cdot s\sqrt{n}$, where $g_{sm+i}$ for $i\in[s\sqrt{n}]$ correspond to taking bad or good branches in the basic $2\times 2$ Ar\i kan's kernel. Similarly to Claim~\ref{cl:Z_evolution}, we then can bound
\[ \ZZ_{m+\sqrt{n}} < \left(2^{s\sqrt{n}}\ZZ_m\right)^{2^{\beta\cdot s\sqrt{n}}} < \left(2^{sm}\exp(-2sm)\right)^{2^{\beta\cdot s\sqrt{n}}} \leq \exp\left(-sm\cdot 2^{\beta\cdot s\sqrt{n}}\right). \] 
Then for the remaining $(n-m-\sqrt{n})$ levels of the tree, it is easy to see that the Bhattacharyaa parameter will also not ever be above the threshold of picking Ar\i kan's kernel in Algorithm~\ref{algo:kernel_search}, thus, similarly as before, we can argue that the Bhattacharyya parameter increases by at most a factor of $2^s$ at each level. Therefore, we derive
\begin{equation}
\label{eq:step2zn}
    \ZZ_n < 2^{s(n-m-\sqrt{n})}Z_{m+\sqrt{n}} \leq 2^{sn}\exp\left(-sm\cdot 2^{\beta\cdot s\sqrt{n}}\right) < \exp\left(-2^{n^{1/3}}\right),
\end{equation}
where the last inequality follows from $m\geq \sqrt{n}$, $\beta = \frac1{20}$, and the condition $s\geq \frac{11}{\a}$ from Theorem~\ref{thm:kernel_seacrh_correct} combined with the fact that $\a$ is small.

Since we proved that the event $E_0^{(n-\sqrt{n})}$ implies $\ZZ_n  < \exp(-2^{n^{1/3}})$, we conclude, using~\eqref{eq:step2e0n}:
\begin{equation}
\label{eq:step2polarization}
    \P[\ZZ_n  < \exp(-2^{n^{1/3}})] \geq e_0^{(n-\sqrt{n})} \geq I(W) -  \l^{-\frac{n}{2} + 11\a n + \sqrt{n}},
\end{equation}
which precisely proves the polarization that was stated in the lemma.

The only thing left to prove then is that this polarization is poly-time constructible. To do this, we show that one can find the set $E_0^{(n-\sqrt{n})}$ of bit-channels in poly-time (recall here the equivalence between events and subsets of the bit-channel at the level $n$ of the tree $\T_n$). But one can see that checking if a particular bit-channel $W_i$, for some $i\in [\l^n]$, is easy. Indeed, to check if $W_i$ is in $E_0^{(n-\sqrt{n})}$, it suffices to check if $W_i$ is in $E_m$ for any $m = \sqrt{n}, 2\sqrt{n}, \dots, n-\sqrt{n}$. But this corresponds to looking at a Bhattacharyya parameter $Z_i^{(m), \bin}$ and checking if it is smaller than $\exp(-2sm)$, and, if this is the case, also looking at how many ``good" branches (in the basic $2\times2$ Ar\i kan's transforms) there were within the next stage ($\sqrt{n}$ levels) in the tree $\T_n$. The latter can be done easily, since this information is essentially given by the index $i$ of the bit-channel $W_i$ (by its binary representation, to be precise). The former is actually also straightforward, since $Z_i^{(m), \bin}$ is the parameter of the binned bit-channel $W_i^{(m), \bin}$ that we are \emph{actually tracking} during the construction phase, so we have this channel written down explicitly, and thus calculating its Bhattacharyya parameter is simple. Therefore all this can be done in time, polynomial in $\l^n$, and then the whole set $E_0^{(n-\sqrt{n})}$ can be found in poly-time (we can also say that the event $E_0^{(n-\sqrt{n})}$ is poly-time checkable). This finishes the proof of this lemma.
\end{proof}

For the following step, we will use the event $E_0^{(n-\sqrt{n})}$ as was defined in the proof of the above lemma. For convenience, we denote it as $R_n = E_0^{(n-\sqrt{n})}$, for any integer $n$. What we will use is that $\P[R_n] \geq I(W) -  \l^{-\frac{n}{2} + 11\a n + \sqrt{n}}$; if $R_n$ happens, then $\ZZ_n < \exp\left(-2^{n^{1/3}}\right)$; and that for any bit-channel it can be checked in poly-time if $R_n$ happened, all of which is proven in Lemma~\ref{lem:step2}.

\subsection{Step 3}
Here we will finally prove the polarization that implies the main result of this paper:
\begin{lem}
\label{lem:step3}
$\P\bigg[\ZZ_t \leq \exp\left(-st\cdot\l^{\a\cdot t}\right)\bigg] \geq I(W) - \l^{-(1/2 - 16\a)t + 2\sqrt{t}}$ for $t\geq C\cdot \log^6s$, where $C$ is an absolute constant. Moreover, this polarization is poly-time constructible.
\end{lem}
\begin{proof}

We will again closely follow the approach from~\cite{Wang-Duursma}, though we are going to change the indexing notations to avoid any confusion with the previous step. We return to having the total depth of the tree to be $t$, and we will have $\sqrt{t}$ stages in the tree, each of length $\sqrt{t}$, similarly to the previous step. As before, we will define several events, starting with $C_0^{(0)} = \emptyset$ and $Q_0^{(0)} = \emptyset$. Then, for $n$ being $\sqrt{t}, 2\sqrt{t}, \dots, t-\sqrt{t}$, we define:
%
\begin{align}
    &C_n = R_n\setminus C_0^{(n-\sqrt{t})} \\
    &C_0^{(n)} = C_0^{(n-\sqrt{t})} \cup C_n \\
    &D_n = C_n \bigcap \left\{\sum_{i=1}^{s(t-n)}g_{i} \leq \a\cdot s\cdot t\right\}\\
    &Q_n = C_n \setminus D_n\\
    &Q_0^{(n)} = Q_0^{(n-\sqrt{t})}\cup Q_n,
\end{align}
where $R_n$ is defined at the end of previous step, and $g_i$'s can again be thought of as independent $\text{Bern}(1/2)$ random variables. The intuition behind what these events correspond to is almost the same as in Step 2, but the bit-channels in $D_n$ have conditions on branching from level $n$ down to the bottom level $t$ (instead of levels between $n$ and $n+\sqrt{t}$). Here, the channels in $Q_0^{(n)}$ are the channels that we mark as ``good" up to level $n$ in the tree, and we will be interested in the final set $Q_0^{(t-\sqrt{t})}$ of ``good" channels in the end.
We again denote by corresponding lowercase letters the probabilities of these events. Define also
\[ f_n = I(W) - c_0^{(n)} \quad \text{and} \quad p_n = I(W) - q_0^{(n)}  \ . \]

First, consider event $C_n$ happening. It means that $R_n$ happens, so $\ZZ_n < \exp\left(-2^{n^{1/3}}\right)$. Then at least for some time, we are going to pick Ar\i kan's kernel in the construction phase, since the Bhattacharyya parameter is small enough. But assuming that we take Ar\i kan's kernels all the way down to the bottom of the tree, one can see 
\[ \ZZ_t < \l^{t-n}\cdot \ZZ_n < \l^{t}\cdot\exp\left(-2^{n^{1/3}}\right) \leq 2^{st}\cdot\exp\left(-2^{t^{1/6}}\right) < 2^{-4s} = \l^{-4} \]
for $t \geq C\log^6 s$, where $C$ is large enough.
Again, by using  Proposition~\ref{prop:approx_accumulation} and~\eqref{eq:Z-H} it is easy to show that the entropy of the binned version of the bit-channel will also always be below the threshold $\l^{-4}$.  It means that we cannot in $(t-n)$ levels go over the threshold of choosing Ar\i kan's kernel, thus we indeed take Ar\i kan's kernel all the way down in the tree for the path for which $R_n$ happens. Thus, similarly to the proof of Lemma~\ref{lem:step2} in the Step~2, we can think of it as taking the basic $2\times 2$ Ar\i kan's kernels $s\cdot(t-n)$ times, starting at level $n$. Therefore if $R_n$ happens, the branching down from level $n$ can be viewed as taking ``good" or ``bad" branches in the $A_2$ kernels, so we again define indicator random variables $g_i$, for $i \in [s(t-n)]$, to denote these branches. It is clear that these random variables are going to be independent $\text{Bern}(1/2)$. These are exactly the random variables $g_i$, for $i\in[s(t-n)]$, appearing in the definition of $D_n$. 

We have 
\[ \frac{d_n}{c_n} = \P\left[\sum_{i=1}^{s(t-n)}g_{i} \leq \a s t\right] \leq 2^{-s(t-n)\left(1-h_2(\delta)\right)} \ , \]
where we denote $\delta \coloneqq \min\left\{\frac{\a t}{t-n}, 1\right\}$. The inequality again follows from the entropic inequality on the sum of binomial coefficients.

Recall that we denoted $f_n = I(W) - c_0^{(n)}$. The event $C_0^{(n)}$ contains the event $R_n$, thus ${f_n \leq  \l^{-\frac{n}{2} + 11\a n + \sqrt{n}}}$, which follows from the proof of Lemma~\ref{lem:step2}. Same inequality holds for $f_n^+$.

We will obtain a recurrence on $p_n - f_n^+$ as follows:
\begin{align}
    p_n - f_n^+ &= I(W) - q_0^{(n)} - (I(W) - c_0^{(n)})^+ \\&= p_{n-\sqrt{t}} - q_n - (f_{n-\sqrt{t}} - c_n)^+ \\
    &\leq  p_{n-\sqrt{t}} - q_n - \frac{q_n}{c_n}(f_{n-\sqrt{t}} - c_n)^+ \\
    &\leq p_{n-\sqrt{t}} - q_n - \frac{q_n}{c_n}(f_{n-\sqrt{t}}^+ - c_n) \\
    &\leq p_{n-\sqrt{t}} - f_{n-\sqrt{t}}^+ + \left(1 - \frac{q_n}{c_n}\right)f_{n-\sqrt{t}}^+ \\
    & = p_{n-\sqrt{t}} - f_{n-\sqrt{t}}^+ + \frac{d_n}{c_n}f_{n-\sqrt{t}}^+ \\
    &\leq p_{n-\sqrt{t}} - f_{n-\sqrt{t}}^+ + \l^{-(1/2 - 11\a)(n-\sqrt{t})+\sqrt{n}}\cdot 2^{-s(t-n)\left(1-h_2(\delta)\right)},
\end{align}
where recall that $\delta =  \min\left\{\frac{\a t}{t-n}, 1\right\}$. We want to obtain an upper bound on the additive term in the inequality above. Consider the following two cases:
\begin{enumerate}[label=\roman*)]
    \item $\delta > \frac1{10}$, i.e. $10\a t > t - n$, thus $n > (1-10\a)t$. Then we give up on the term  $2^{-s(t-n)\left(1-h_2(\delta)\right)}$ completely, and we can write
    \begin{equation}
        \l^{-(1/2 - 11\a)(n-\sqrt{t})+\sqrt{n}}\cdot 2^{-s(t-n)\left(1-h_2(\delta)\right)} \leq \l^{-(1/2 - 11\a)(1-10\a)t+\frac32\sqrt{t}} \leq \l^{-(1/2 - 16\a)t+\frac32\sqrt{t}};
    \end{equation}
    \item $\delta \leq \frac1{10}$, and then $h_2(\delta) < 1/2$. In this case we derive
    \begin{align}
        \l^{-(1/2 - 11\a)(n-\sqrt{t}) + \sqrt{n}}\cdot 2^{-s(t-n)\left(1-h_2(\delta)\right)} \leq \l^{-(1/2 - 11\a)n + \frac32\sqrt{t}}\cdot \l^{-1/2\cdot (t-n)} &= \l^{-1/2\cdot t + 11\a n + \frac32\sqrt{t}} \\&< \l^{-1/2\cdot t + 11\a t + \frac32\sqrt{t}}.
    \end{align}
\end{enumerate}

Putting the above together, we obtain
\begin{align}
    p_0 - f_{0}^+ &= 0 \\
    p_n - f_n^+ &\leq p_{n-\sqrt{t}} - f_{n-\sqrt{t}}^+ + \l^{-(1/2 - 16\a)t+\frac32\sqrt{t}}.
\end{align}
Therefore $p_{t-\sqrt{t}} - f_{t-\sqrt{t}}^+ \leq \sqrt{t}\cdot\l^{-(1/2 - 16\a)t+\frac32\sqrt{t}}$. Combining this with $f_{t-\sqrt{t}}^+ \leq \l^{-(1/2 - 11\a)(t-\sqrt{t})+\sqrt{t}}$, we obtain $p_{t-\sqrt{t}} \leq \l^{-(1/2 - 16\a)t+2\sqrt{t}}$, and thus 
\begin{equation}
\label{eq:step3:1}
 \P\left[Q_0^{(t-\sqrt{t})}\right] = q_0^{(t-\sqrt{t})} \geq I(W) - \l^{-(1/2 - 16\a)t+2\sqrt{t}}.
 \end{equation}

Let us now check that the event $Q_0^{(t-\sqrt{t})}$ is actually ``good" and allows us achieve the needed polarization. If $Q_0^{(t-\sqrt{t})}$ happens, then $Q_n$ happened for some $n = k\cdot\sqrt{t}$. It means that $C_n$, and therefore $R_n$ takes place, thus $\ZZ_n < \exp\left(-2^{n^{1/3}}\right)$. It also means that $D_n$ does not happen, and thus there is at least $\a s t$ ``good" branches taken in the way down the tree, which corresponds to $\a s t$ squarings of the Bhattacharyya parameter. Therefore 
\begin{equation}
\label{eq:step3:2}
    \ZZ_t \leq \left(\l^{t-n}\ZZ_n\right)^{2^{\a st}} < \left(2^{st}\exp\left(-2^{n^{1/3}}\right)\right)^{2^{\a st}} < \exp\left(-st\cdot 2^{\a st}\right) = \exp\left(-st\cdot \l^{\a t}\right) = \frac1N\exp\left(-N^{\a}\right),
\end{equation}
where the third inequality trivially follows from $n \geq \sqrt{t}$ and $t \geq C\log^6 s$ for large enough $C$. Combining this with~\eqref{eq:step3:1}, we obtain the desired polarization:
\begin{equation}
\label{eq:step3polarization}
    \P\left[\ZZ_t  <\exp\left(-st\cdot 2^{\a st}\right) \right] \geq q_0^{(t-\sqrt{t})} \geq I(W) - \l^{-(1/2 - 16\a)t+o(t)}.
\end{equation}

It only remains to argue that this polarization is poly-time constructible. But this easily follows from the fact that the event $R_n$ is poly-time checkable, which we proved in Step~2. Indeed, now for any bit-channel $W_i$, $i\in[\l^t]$, we need to check if it is in $Q_0^{(t-\sqrt{t})}$. This means that one need to see if $Q_n$ happened for some $n = k\sqrt{t}$. To do this, one checks in poly-time if $C_n$ happened, which reduces to checking $R_n$ (which can be done in poly-time). If $R_n$ happened, then the only thing to check is how many ``good" branches the remaining path to $W_i$ has, which is easily (in poly-time) retrievable information from the index $i$. Therefore, the event $Q_0^{(t-\sqrt{t})}$ is indeed poly-time checkable, which finishes the proof of the lemma.
\end{proof}

\appendix

\bigskip
\noindent \centerline{{\Large \textbf{Appendices}}}

\section{Proofs of entropic lemmas for BMS channels}
\label{app:BMS_lemmas}

In the following two proofs we use the representation of BMS channel $W$ as a convex combination of several BSC subchannels $W^{(1)}$, $W^{(2)}$, \dots, $W^{(m)}$, see the beginning of Section~\ref{sec:BMS_large_alphabet} for details. Each subchannel $W^{(j)}$ can output one of two symbols $z^{(0)}_j, z^{(1)}_j$, and $\Wj(\zj|0) = \Wj(\zzj|1)$, $\Wj(\zzj|0) = \Wj(\zj|1)$. The output alphabet for $W$ is thus $\Y = \{z^{(0)}_1, z^{(1)}_1, z^{(0)}_2, z^{(1)}_2, \dots, z^{(0)}_m, z^{(1)}_m\}$. Define for these proofs the ``flip" operator $\oplus\,:\, \Y \times\bit\to\Y$ as follows: $z_j^{(c)} \oplus b = z_j^{(b+c)}$, where $b,c \in \bit$, and $(b+c)$ is addition mod $2$. In other words, $z_j^{(c)}\oplus 0$ doesn't change anything, and $z_j^{(c)}\oplus 1$ flips the output of the subchannel $\Wj$ to the opposite symbol. Note then that $W^{(j)}(z_j^{(c)}\,|\,b) = W^{(j)}(z_j^{(c)}\oplus b\,|\,0)$. Finally, we overload the operator to also work on $\Y^{\l} \times \bit^{\l} \to \Y^{\l}$ by applying it coordinate-wise.
It then easily follows that $W^{\l}\left(\by\,|\,\bx \right) = W^{\l}\left(\by\oplus \bx \,|\,\mathbi{0}\right)$ for any $\by\in\Y^{\l}$ and  $\bx\in\bit^{\l}$.

\begin{proof}[Proof of Lemma~\ref{typical_entropy_BSC}] We can write 
\begin{align}
     \E_{g\sim G}\hspace{-1pt}\big[H^{(g)}(V_1|\bY)\big] \hspace{-1.5pt}&=\sum_g \P(G=g) \left(\sum_{\by \in \Y^{\l}} \P\nolimits^{(g)}[\bY=\by] H^{(g)}(V_1|\bY=\by) \right)  \nonumber\\
    &= \sum_g \P(G=g) \left(\sum_{\by \in \Y^{\l}}
    \left(\sum_{\bv \in \bit^k}\P\nolimits^{(g)}[\bY=\by,\bV=\bv] \right) 
    h \left( \frac{\P\nolimits^{(g)}[V_1=0,\bY=\by]}{\P\nolimits^{(g)}[\bY=\by]} \right) \right) \nonumber\\
    &= \frac{1}{2^{k}}
    \sum_{\bv \in \bit^k}
    \sum_g  \P(G=g) \sum_{\by \in \Y^{\l}}
    \P\nolimits^{(g)}[\bY=\by \big| \bV=\bv]
    h \left( \frac{\P\nolimits^{(g)}[V_1=0,\bY=\by]}{\P\nolimits^{(g)}[\bY=\by]} \right), \label{eq:calA}
\end{align}
where $h(x):=-x\log_2 x-(1-x)\log_2(1-x)$ is the binary entropy function.
Next, we show that for any fixed codebook $g$ and any fixed $\bv \in \bit^k$ it holds
\begin{align}
\label{eq:cond2:1}
    \hspace{-4pt}\sum_{\by \in \Y^{\l}}
    \P\nolimits^{(g)}[\bY=\by \big| \bV=\bv]  
    h\hspace{-3pt}\left( \frac{\P\nolimits^{(g)}[V_1=0,\bY=\by]}{\P\nolimits^{(g)}[\bY=\by]} \right) \hspace{-2pt}
    = \hspace{-4pt}
     \sum_{\by \in \Y^{\l}}\hspace{-2.5pt}
    \P\nolimits^{(g)}[\bY=\by \big| \bV=\mathbi{0}]
    h\hspace{-3pt}\left( \frac{\P\nolimits^{(g)}[V_1=0,\bY=\by]}{\P\nolimits^{(g)}[\bY=\by]} \right),  
\end{align}
where $\mathbi{0}$ is the all-zero vector.

First of all, we know that 
\begin{align}
\label{eq:cond2:2}
     \P\nolimits^{(g)}[\bY=\by \big| \bV=\bv]  = W^{\l}(\by\,|\,\bv G) = W^{\l}(\by\oplus\bv G\,|\,\bz) = \P\nolimits^{(g)}[\bY = \by \oplus\bv G\, |\, \bV = \bz],
\end{align}
as was discussed at the beginning of this appendix. In the same way, it's easy to see
\begin{align}
    \P\nolimits^{(g)}[\bY=\by] = \frac1{2^k}\sum_{\bu\in\bit^k}\P\nolimits^{(g)}[\bY=\by \big| \bV=\bu] 
    &= \frac1{2^k}\sum_{\bu\in\bit^k}\P\nolimits^{(g)}[\bY=\by \oplus \bv G\big| \bV=\bu + \bv] \\
    &= \frac1{2^k}\sum_{\bu + \bv\in\bit^k}\P\nolimits^{(g)}[\bY=\by \oplus \bv G\big| \bV=\bu + \bv] \\
    &= \P\nolimits^{(g)}[\bY = \by\oplus\bv G].\label{eq:cond2:3}\noeqref{eq:cond2:3}
\end{align}
The above equality uses the fact the we are considering linear codes, and $\bv G$ is an arbitrary codeword. It follows from the symmetry of linear codes that ``shifting" the output by a codeword does not change anything. Shifting here means the usual shifting for the BSC case, though for general BMS channel this is actually flipping the outputs or appropriate BSC subchannels, without changing which subchannel was actually used for which bit.

Denote now $\wV = \bV_{>1}$, and recall that we are considering fixed $\bv$ for now. Denote then also $v_1$ as the first coordinate of $\bv$ and $\wv = \bv_{>1}$. Then we derive similarly

\begin{equation}
\begin{aligned}
    \P\nolimits^{(g)}[V_1=0, \bY=\by] &= \frac1{2^k}\sum_{\wu\in\bit^{k-1}}\P\nolimits^{(g)}[\bY=\by \big|V_1=0, \wV=\wu] \\
    &= \frac1{2^k}\sum_{\wu\in\bit^{k-1}}\P\nolimits^{(g)}[\bY=\by \oplus \bv G \,|\, V_1 = v_1, \wV=\wu + \wv] \\
    &= \frac1{2^k}\sum_{\wu + \wv\in\bit^{k-1}}\P\nolimits^{(g)}[\bY=\by \oplus \bv G\,\big|\, V_1 = v_1, \wV=\wu + \wv] \\
    &= \P\nolimits^{(g)}[V_1 = v_1, \bY = \by\oplus\bv G].
\label{eq:cond2:4}
\end{aligned}
\end{equation}

Notice that $\P\nolimits^{(g)}[V_1 = v_1, \bY = \by\oplus\bv G] + \P\nolimits^{(g)}[V_1 = 1- v_1, \bY = \by\oplus\bv G] = \P\nolimits^{(g)}[\bY = \by\oplus\bv G]$, and thus using the symmetry of the binary entropy function around $1/2$ obtain
\begin{align}
    h\left(\dfrac{\P\nolimits^{(g)}[V_1 = v_1, \bY = \by\oplus\bv G] }{\P\nolimits^{(g)}[\bY = \by\oplus\bv G]} \right) = h\left(\dfrac{\P\nolimits^{(g)}[V_1 = 1 - v_1, \bY = \by\oplus\bv G] }{\P\nolimits^{(g)}[\bY = \by\oplus\bv G]} \right).
\end{align}
Using this and~\eqref{eq:cond2:2}--\eqref{eq:cond2:4} derive
\begin{align}
    \label{eq:cond2:5}
    & \P\nolimits^{(g)}[\bY=\by \big| \bV=\bv]  
    h \left( \frac{\P\nolimits^{(g)}[V_1=0,\bY=\by]}{\P\nolimits^{(g)}[\bY=\by]} \right) \\
    =& \P\nolimits^{(g)}[\bY = \by \oplus\bv G\, |\, \bV = \bz]h \left( \frac{\P\nolimits^{(g)}[V_1=0,\bY=\by \oplus\bv G]}{\P\nolimits^{(g)}[\bY=\by \oplus\bv G]} \right).
\end{align}
Finally, summing both parts over $\by\in\Y^{\l}$ and noticing that $\by\oplus\bv G$ will also range through all $\Y^{\l}$ in this case, we establish~\eqref{eq:cond2:1}. Then in~\eqref{eq:calA} deduce

\begin{align}
     \E_{g\sim G}\big[H^{(g)}(V_1|\bY)\big] &=
    \frac{1}{2^{k}}
    \sum_{\bv \in \bit^k}
    \sum_g  \P(G=g) \sum_{\by \in \Y^{\l}}
    \P\nolimits^{(g)}[\bY=\by \big| \bV=\bz]
    h \left( \frac{\P\nolimits^{(g)}[V_1=0,\bY=\by]}{\P\nolimits^{(g)}[\bY=\by]} \right) \\
    &=  \sum_{\by \in \Y^{\l}}\sum_g  \P(G=g)
    \P\nolimits^{(g)}[\bY=\by \big| \bV=\bz]
    h \left( \frac{\P\nolimits^{(g)}[V_1=0,\bY=\by]}{\P\nolimits^{(g)}[\bY=\by]} \right) \\
    &= \sum_{\by \in \Y^{\l}}
    \P[\bY=\by \big| \bV=\bz] \E_{g\sim G}\left[H^{(g)}(V_1|\bY=\by)\right],
\end{align}
since $\P\nolimits^{(g)}[\bY=\by \big| \bV=\bz]$ does not depend on the matrix $g$.
\end{proof}

\begin{proof}[Proof of Proposition~\ref{prop:Arikan-bit}] Let us unfold the conditioning in the LHS as follows
\begin{equation}
\label{eq:cond:1}
     H\left(U_i\;\Big\lvert\; W^{\l}(\bU\cdot K), \bU_{<i}\right) = \E_{\bw\sim\bit^{i-1}}\left[ H\left(U_i\;\Big\lvert\; W^{\l}(\bU\cdot K), \bU_{<i} = \bw\right) \right].
\end{equation}
We are going to show that the conditional entropy inside the expectation doesn't depend on the choice of $\bw$, which will allow us to restrict to $\bw = \mathbi{0}.$

Return now to the settings of the Proposition, and denote the (random) output $\bY = W^{\l}(\bU\cdot K)$. Let us now fix some $\bw \in\bit^{i-1}$ and consider $H\left(U_i\;\Big\lvert\; \bY   , \bU_{<i} = \bw\right)$. Unfolding the conditional entropy even more, derive
\begin{equation}
\label{eq:cond:2}
     H\left(U_i\;\Big\lvert\; \bY   , \bU_{<i} = \bw\right) = \sum_{\by\in\Y^{\l}}\P[\bY = \by\,|\,\bU_{<i}=\bw]\cdot H\left(U_i\;\Big\lvert\; \bY = \by  , \bU_{<i} = \bw\right).
\end{equation}

Denote now by $B$ the first $(i-1)$ rows of $K$, and
thus $\bY = W^{\l}(\bU\cdot K) = W^{\l}(\bU_{<i}\cdot B + \bU_{\geq i}\cdot G)$. We then have
\begin{align}
    \P[\bY = \by\,|\,\bU_{<i}=\bw] 
    &= \sum_{\bv\in\bit^{k}}\frac1{2^k}\P[\bY = \by\,|\,\bU_{<i}=\bw, \bU_{\geq i} = \bv] \\ 
    &= \sum_{\bv\in\bit^{k}}\frac1{2^k}W^{\l}\Big(\by\,\Big\lvert\,\bw\cdot B + \bv\cdot G \Big) \\ 
    &= \sum_{\bv\in\bit^{k}}\frac1{2^k}W^{\l}\Big(\by \oplus \bw B\,\Big\lvert\, \bv\cdot G \Big) \\
    &= \sum_{\bv\in\bit^{k}}\frac1{2^k}\P[\bY = \by\oplus \bw B\,|\,\bU_{<i}=\bz, \bU_{\geq i} = \bv]   \\
    &= \P[\bY = \by\oplus\bw B\,|\, \bU_{<i} = \bz].
    \label{eq:cond:3}
\end{align}

For the entropy i the RHS of~\eqref{eq:cond:2}, observe
\begin{align}
    H\left(U_i\;\Big\lvert\; \bY = \by  , \bU_{<i} = \bw\right) = h\Big(\P\left[U_i = 0\,|\,\bY=\by, \bU_{<i}=\bw\right]\Big),
\end{align}
where $h(\cdot)$ is a binary entropy function. Out of the definition of conditional probability, obtain
\begin{align}
    \P\left[U_i = 0\,|\,\bY=\by, \bU_{<i}=\bw\right] 
    &= \dfrac{\P[U_i = 0, \bY=\by\,|\,\bU_{<i} = \bw]}{\P[\bY=\by\,|\,\bU_{<i}=\bw]} \\ 
    &= \dfrac{\P[U_i = 0, \bY=\by\oplus\bw B\,|\,\bU_{<i} =\bz]}{\P[\bY=\by\oplus\bw B\,|\,\bU_{<i}=\bz]} \\
    &= \P\left[U_i = 0\,|\,\bY=\by\oplus\bw B, \bU_{<i}=\bz\right],
\end{align}
where the second equality also uses~\eqref{eq:cond:3} (and similar equality with $U_i=0$ inside the probability, which is completely analogical to~\eqref{eq:cond:3}). Therefore, deduce in~\eqref{eq:cond:2}
\begin{align}
    H\left(U_i\;\Big\lvert\; \bY   , \bU_{<i} = \bw\right) &= \sum_{\by\in\Y^{\l}}\P[\bY = \by\oplus\bw B\,|\,\bU_{<i}=\bz]\cdot H\left(U_i\;\Big\lvert\; \bY = \by \oplus\bw B , \bU_{<i} = \bz\right) \\
    &= \sum_{\bzz\in\Y^{\l}}\P[\bY = \bzz\,|\,\bU_{<i}=\bz]\cdot H\left(U_i\;\Big\lvert\; \bY = \bzz, \bU_{<i} = \bz\right) \\
    &= H\left(U_i\;\Big\lvert\; \bY   , \bU_{<i} = \bz\right),
\end{align}
since $\bzz = \by \oplus \bw B$ ranges over all $\Y^{\l}$ for $\by \in \Y^{\l}$. Therefore, in~\eqref{eq:cond:1} there is no actual dependence on $\bw$ under the expectation in the RHS, and thus
\begin{equation}
\label{eq:cond:4}
     H\left(U_i\;\Big\lvert\; W^{\l}(\bU\cdot K), \bU_{<i}\right) = H\left(U_i\;\Big\lvert\; W^{\l}(\bU\cdot K), \bU_{<i} = \bz\right).
\end{equation}

Finally, note that we can take $\bV = \bU_{\geq i}$, since it is uniformly distributed over $\bit^k$, and then $V_1 = U_{i}$. Since $\bU\cdot K = \bU_{\geq i}\cdot G = \bV\cdot G$ when $\bU_{<i} = \bz$, we indeed obtain 
\begin{align}
\label{eq:cond:final}
     H\left(U_i\;\Big\lvert\; W^{\l}(\bU\cdot K), \bU_{<i}\right) = H\left(U_i\;\Big\lvert\; W^{\l}(\bU\cdot K), \bU_{<i} = \bz\right) = H\left(V_1\;\Big\lvert\; W^{\l}(\bV\cdot G)\right). &\qedhere
\end{align}
\end{proof}
\section{Proofs in Section~\ref{sec:concentration}}
\label{app:moments}
\begin{proof}[Proof of Claim~\ref{cl:exp_of_Di}]
     Denote for convenience the distribution $\O_i \coloneqq \text{Binom}(d_i, p_i)$. Note that $\E_{\chi_i\sim\O_i}\left[\frac{\chi_i}{d_i}\right] = p_i$. Then we derive
      \begingroup
 \allowdisplaybreaks
     \begin{align*}
         \left|\E_{\chi_i\sim\D_i}\left[\frac{\chi_i}{d_i}\right] - p_i\right| &= \left|\E_{\chi_i\sim\D_i}\left[\frac{\chi_i}{d_i}\right] - \E_{\chi_i\sim\O_i}\left[\frac{\chi_i}{d_i}\right]\right|\\ &= \left| \sum_{s\in[0:d_i]}\frac{s}{d_i}\P_{\chi_i\sim\D_i}[\chi_i=s] -  \sum_{s\in[0:d_i]}\frac{s}{d_i}\P_{\chi_i\sim\O_i}[\chi_i=s]   \right|\\
         &\overset{\eqref{truncated_distr}}{=}  \left| \sum_{s\in\T_1^{(i)}}\frac{s}{d_i}\P_{\chi_i\sim\O_i}[\chi_i=s]\cdot\theta_i^{-1} -  \sum_{s\in[0:d_i]}\frac{s}{d_i}\P_{\chi_i\sim\O_i}[\chi_i=s]   \right| \\
         &= \left| \sum_{s\in\T_1^{(i)}}\frac{s}{d_i}\P_{\chi_i\sim\O_i}[\chi_i=s]\cdot\left(\theta_i^{-1} - 1\right) -  \sum_{s\notin\T_1^{(i)}}\frac{s}{d_i}\P_{\chi_i\sim\O_i}[\chi_i=s]   \right|\\
         &\leq \sum_{s\in\T_1^{(i)}}\frac{s}{d_i}\P_{\chi_i\sim\O_i}[\chi_i=s]\cdot\left(\theta_i^{-1} - 1\right) +  \sum_{s\notin\T_1^{(i)}}\frac{s}{d_i}\P_{\chi_i\sim\O_i}[\chi_i=s].
 \end{align*}
 \endgroup
 We have $\sum\limits_{s\notin\T_1^{(i)}}\frac{s}{d_i}\P_{\chi_i\sim\O_i}[\chi_i=s] \leq \sum\limits_{s\notin\T_1^{(i)}}\P_{\chi_i\sim\O_i}[\chi_i=s] \overset{\eqref{theta_def}}{=} (1 - \theta_i) \overset{\eqref{indiv_Chernoff}}{\leq} 2\l^{-(\log\l)/3}$.\\
 Next, $\sum_{s\in\T_1^{(i)}}\frac{s}{d_i}\P_{\chi_i\sim\O_i}[\chi_i=s] \leq \E_{\chi_i\sim\O_i}\left[\frac{\chi_i}{d_i}\right] \leq 1$, and $\theta_i^{-1} - 1 = \frac{1 - \theta_i}{\theta_i} \leq 2(1-\theta_i) \leq 4 \l^{-(\log\l)/3}$.\\
 Combining the above together, conclude $\left|\E\left[\frac{\chi_i}{d_i}\right] - p_i\right| \leq 6\l^{-(\log\l)/3} \leq \frac1{\l} \leq \frac1{d_i}$.
\end{proof}

\begin{proof}[Proof of Claim~\ref{cl:mean_absolute}]
Using the result of Claim~\ref{cl:exp_of_Di} derive
\begin{equation}
\label{eq:mean_absolute}
\E\Big\lvert\chi_i - \E[\chi_i]\Big\lvert \leq \E\Big\lvert\chi_i - p_id_i\Big\lvert + \E\Big\lvert p_id_i - \E[\chi_i]\Big\lvert \leq  \E\Big\lvert\chi_i - p_id_i\Big\lvert + 1.
\end{equation}
From~\eqref{theta_def},~\eqref{distributions_almost_equal}, and definition~\eqref{chernoff_intervals} of $\T_1^{(i)}$ for $i\in F_2$ observe also the following:
 \begingroup
 \allowdisplaybreaks
\begin{align}
    \E_{\chi_i\sim\D_i}\Big\lvert\chi_i - p_id_i\Big\lvert &= \sum_{s\in\T_1^{(i)}} \Big\lvert s - p_id_i\Big\lvert\cdot\P_{\eta_i\sim\O_i}[s]\cdot\theta_i^{-1} \\
    &= \sum_{s\in\T_1^{(i)}} \Big\lvert s - p_id_i\Big\lvert\cdot\P_{\eta_i\sim\O_i}[s] + \sum_{s\in\T_1^{(i)}} \Big\lvert s - p_id_i\Big\lvert\cdot\P_{\eta_i\sim\O_i}[s]\cdot(\theta_i^{-1} - 1)\\
    &\leq \sum_{s\in\T_1^{(i)}} \Big\lvert s - p_id_i\Big\lvert\cdot\P_{\eta_i\sim\O_i}[s] + \sqrt{d_ip_i}\log\l\cdot\underbrace{\sum_{s\in\T_1^{(i)}}\P_{\eta_i\sim\O_i}[s]}_{\theta_i}\cdot\left(\frac{1 - \theta_i}{\theta_i}\right)\\
    &= \sum_{s\in\T_1^{(i)}} \Big\lvert s - p_id_i\Big\lvert\cdot\P_{\eta_i\sim\O_i}[s] + \sqrt{d_ip_i}\log\l\cdot(1 - \theta_i) \\
    &= \sum_{s\in\T_1^{(i)}} \Big\lvert s - p_id_i\Big\lvert\cdot\P_{\eta_i\sim\O_i}[s] + \sum_{s\notin\T_1^{(i)}}\sqrt{d_ip_i}\log\l\cdot\P_{\eta_i\sim\O_i}[s] \\
    &\leq \sum_{s\in\T_1^{(i)}} \Big\lvert s - p_id_i\Big\lvert\cdot\P_{\eta_i\sim\O_i}[s] + \sum_{s\notin\T_1^{(i)}}\Big\lvert s - p_id_i\Big\lvert\cdot\P_{\eta_i\sim\O_i}[s] =  \E_{\eta_i\sim\O_i}\Big\lvert\eta_i - p_id_i\Big\lvert.
\end{align}
\endgroup
Combining this with~\eqref{eq:mean_absolute}, obtain the needed.
\end{proof}

\section{Proof in Section~\ref{sec:BMS_any_alphabet}}
\label{app:upgraded_calc}
Here we formally show that the channel $\wW$ we constructed in Section~\ref{sec:BMS_any_alphabet} is indeed upgraded with respect to $W$. Recall that $W$, $\wW$, and $W_1$ are defined in~\eqref{eq:W_def},~\eqref{eq:wW_def}, and~\eqref{eq:W1_def} correspondingly, and our goal is to prove~\eqref{eq:upgrad_calculation}.
First, to check that $W_1$ is a valid channel, observe
\[ \sum_{k\in[m],\; c\in\bit} W_1\left(z_k^{(c)} \;\Big\lvert\;\wt{z_j^{(b)}}\right) = \sum_{k\in T_j}\left( W_1\left(z_k^{0} \;\Big\lvert\;\wt{z_j^{(b)}}\right) + W_1\left(z_k^{1} \;\Big\lvert\;\wt{z_j^{(b)}}\right)\right) = \sum_{k\in T_j} \dfrac{q_k}{\sum\limits_{i\in T_j}q_i} = 1.  \]
 Finally, for any $k\in[m],\;c\in\bit$, let $j_k$ be such that $k \in T_{j_k}$. Then we have for any $x\in \bit$
 \[ \sum_{j\in[\sqrt{\l}],\;b\in\bit}\wW\left(\wt{z_j^{(b)}}\,\Big\lvert\,x\right)W_1\left(z_k^{(c)} \;\Big\lvert\;\wt{z_j^{(b)}}\right) = \sum_{b\in\bit} \wW\left(\wt{z_{j_k}^{(b)}}\,\Big\lvert\,x\right)W_1\left(z_k^{(c)} \;\Big\lvert\;\wt{z_{j_k}^{(b)}}\right).  \]
 Now, if $x = c$, we derive
 \begingroup
 \allowdisplaybreaks
 \begin{align*}
 \sum_{b\in\bit} \wW & \left(\wt{z_{j_k}^{(b)}}  \,\Big\lvert\,x \right) W_1\left(z_k^{(c)} \;\Big\lvert\;\wt{z_{j_k}^{(b)}}\right)\\ 
 &=  \wW\left(\wt{z_{j_k}^{(x)}}\,\Big\lvert\,x\right)W_1\left(z_k^{(x)} \;\Big\lvert\;\wt{z_{j_k}^{(x)}}\right) + \wW\left(\wt{z_{j_k}^{(1-x)}}\,\Big\lvert\,x\right)W_1\left(z_k^{(x)} \;\Big\lvert\;\wt{z_{j_k}^{(1-x)}}\right)\\
 &=\sum_{i\in T_{j_k}}q_i\cdot(1-\theta_{j_k})\cdot\dfrac{q_k}{\sum\limits_{i\in T_{j_k}}q_i}\cdot\left(1 - \dfrac{p_k - \theta_{j_k}}{1-2\theta_{j_k}} \right) + \sum_{i\in T_{j_k}}q_i\cdot\theta_{j_k}\cdot\dfrac{q_k}{\sum\limits_{i\in T_{j_k}}q_i}\cdot\left(\dfrac{p_k - \theta_{j_k}}{1-2\theta_{j_k}} \right) \\
 &= q_k \left( 1 - \theta_{j_k} - (1 - \theta_{j_k})\cdot \left(\dfrac{p_k - \theta_{j_k}}{1-2\theta_{j_k}} \right) + \theta_{k_j}\cdot\left(\dfrac{p_k - \theta_{j_k}}{1-2\theta_{j_k}} \right) \right) \\
 &= q_k \left( 1 - \theta_{j_k} - (1 - 2\theta_{j_k})\cdot \left(\dfrac{p_k - \theta_{j_k}}{1-2\theta_{j_k}} \right) \right) = q_k\cdot(1-p_k). 
 \end{align*}
 \endgroup
Otherwise, then $x = 1-c$, obtain
 \begingroup
 \allowdisplaybreaks
\begin{align*}
     \sum_{b\in\bit} \wW & \left(\wt{z_{j_k}^{(b)}}  \,\Big\lvert\,x \right) W_1\left(z_k^{(c)} \;\Big\lvert\;\wt{z_{j_k}^{(b)}}\right)\\
 &=  \wW\left(\wt{z_{j_k}^{(x)}}\,\Big\lvert\,x\right)W_1\left(z_k^{(1-x)} \;\Big\lvert\;\wt{z_{j_k}^{(x)}}\right) + \wW\left(\wt{z_{j_k}^{(1-x)}}\,\Big\lvert\,x\right)W_1\left(z_k^{(1-x)} \;\Big\lvert\;\wt{z_{j_k}^{(1-x)}}\right)\\
 &=\sum_{i\in T_{j_k}}q_i\cdot(1-\theta_{j_k})\cdot\dfrac{q_k}{\sum\limits_{i\in T_{j_k}}q_i}\cdot\left(\dfrac{p_k - \theta_{j_k}}{1-2\theta_{j_k}} \right) + \sum_{i\in T_{j_k}}q_i\cdot\theta_{j_k}\cdot\dfrac{q_k}{\sum\limits_{i\in T_{j_k}}q_i}\cdot\left( 1 - \dfrac{p_k - \theta_{j_k}}{1-2\theta_{j_k}} \right) \\
 &= q_k \left( (1 - \theta_{j_k})\cdot \left(\dfrac{p_k - \theta_{j_k}}{1-2\theta_{j_k}} \right) + \theta_{j_k} -  \theta_{j_k}\cdot\left(\dfrac{p_k - \theta_{j_k}}{1-2\theta_{j_k}} \right) \right) \\
 &= q_k \left( (1 - 2\theta_{j_k})\cdot \left(\dfrac{p_k - \theta_{j_k}}{1-2\theta_{j_k}} \right) + \theta_{j_k} \right) = q_k\cdot p_k.
 \end{align*}
 \endgroup
 Therefore, for any $k\in[m]$ and $c,x \in \bit$ it holds
 \[  \sum_{j\in[\sqrt{\l}],\;b\in\bit}\wW\left(\wt{z_j^{(b)}}\,\Big\lvert\,x\right)W_1\left(z_k^{(c)} \;\Big\lvert\;\wt{z_j^{(b)}}\right) = W\left(z^{(c)}_k\;\Big\lvert\;x\right).  \]

\section{Proof of Proposition~\ref{prop:eqv}}
\label{app:proofeqv}
We still use $\bU_{[1:N]}$ to denote the information vector and use $\bX_{[1:N]}=\bU_{[1:N]} M^{(t)}$ to denote the encoded vector. Assume that $\bU_{[1:N]}$ consists of $N$ i.i.d. Bernoulli-$1/2$ random variables.
Similarly to the example in Section~\ref{sect:bit}, we define the random vectors $\bV_{[1:N]}^{(j)},\bU_{[1:N]}^{(j)}$ for $j=t-1,t-2,\dots,1$ recursively
\begin{equation}  \label{eq:dfrecu}
\begin{aligned}
\bV_{[1:N]}^{(t-1)} &= \bU_{[1:N]} D^{(t-1)},  \\
\bU_{[1:N]}^{(j)} &= \bV_{[1:N]}^{(j)} Q^{(j)}
\text{~for~}j=t-1,t-2,\dots,1 ,  \\
\bV_{[1:N]}^{(j)} &=  \bU_{[1:N]}^{(j+1)} D^{(j)}
\text{~for~}j=t-2,t-3,\dots,1 ,  \\
\bX_{[1:N]} &=  \bU_{[1:N]}^{(1)} D^{(0)}  .
\end{aligned}
\end{equation}
Moreover, let $\bU_{[1:N]}^{(t)}:=\bU_{[1:N]}$.
We will prove the following two claims:
\begin{enumerate}
    \item For every $a=1,2,\dots,t$, the following $\ell^{t-a}$ random vectors
    $$
    (\bU_{[h\ell^a+1:h\ell^a+\ell^a]}^{(a)},\bY_{[h\ell^a+1:h\ell^a+\ell^a]}), \quad
    h=0,1,\dots,\ell^{t-a}-1
    $$
    are i.i.d.
    
    \item For every $a=1,2,\dots,t$ and every $i\in[\ell^a]$, we write $\tau_a(i)=(i_1,i_2,\dots,i_a)$, where $\tau_a$ is the $a$-digit expansion function defined in \eqref{eq:deftau}. Then for every $h=0,1,\dots,\ell^{t-a}-1$ and every $i\in[\ell^a]$, we have
    \begin{equation}\label{eq:cl2}
    \P( U_{h\ell^a+i}^{(a)}
    \to (\bU_{[h\ell^a+1:h\ell^a+i-1]}^{(a)},\bY_{[h\ell^a+1:h\ell^a+\ell^a]}) )
    \equiv
    W_{i_1,\dots,i_a}(K_1^{(0)},K_{i_1}^{(1)},\dots,K_{i_1,\dots,i_{a-1}}^{(a-1)})
    .
    \end{equation}
\end{enumerate}
Note that Proposition~\ref{prop:eqv} follows immediately from taking $a=t$ in \eqref{eq:cl2}. Therefore, we only need to prove these two claims.

We start with the first claim.
By \eqref{eq:defD}, for every $j=0,1,\dots,t-1$, the matrix $D^{(j)}$ is a block diagonal matrix with $\ell^{t-j-1}$ blocks on the diagonal, where each block has size $\ell^{j+1}\times \ell^{j+1}$, and all the $\ell^{t-j-1}$ blocks are the same.
According to \eqref{eq:defpi}--\eqref{eq:defQ}, the permutation matrix $Q^{(j)}$ keeps the first $t-j-1$ digits of the $\ell$-ary expansion to be the same and performs a cyclic shift on the last $j+1$ digits.
Therefore, for every $j=1,\dots,t-1$, the permutation matrix $Q^{(j)}$ is also a block diagonal matrix with $\ell^{t-j-1}$ blocks on the diagonal, where each block has size $\ell^{j+1}\times \ell^{j+1}$, and all the $\ell^{t-j-1}$ blocks are the same.
Therefore, for every $j\in[t]$, the matrix $M^{(j)}$ defined in \eqref{eq:defMj} can be written in the following block diagonal form
\begin{equation} \label{eq:bdia}
M^{(j)}:=\underbrace{\{\overline{M}^{(j)},\overline{M}^{(j)},\dots,\overline{M}^{(j)}\}}_{
        \text{number of }\overline{M}^{(j)} \text{ is } \ell^{t-j}}  ,
\end{equation}
where the size of $\overline{M}^{(j)}$ is $\ell^j \times \ell^j$.
By the recursive definition \eqref{eq:dfrecu}, one can show that for every $j\in[t]$, we have
$$
\bX_{[1:N]}=\bU_{[1:N]}^{(j)} M^{(j)}.
$$
Combining this with \eqref{eq:bdia}, we obtain that for every $a\in[t]$ and every $h=0,1,\dots,\ell^{t-a}-1$,
\begin{equation}\label{eq:blk}
\bX_{[h\ell^a+1:h\ell^a+\ell^a]}=\bU_{[h\ell^a+1:h\ell^a+\ell^a]}^{(a)} \overline{M}^{(a)}.
\end{equation}
Since $\bX_{[1:N]}$ consists of $N$ i.i.d. Bernoulli-$1/2$ random variables, the following $\ell^{t-a}$ random vectors
    $$
    (\bX_{[h\ell^a+1:h\ell^a+\ell^a]},\bY_{[h\ell^a+1:h\ell^a+\ell^a]}), \quad
    h=0,1,\dots,\ell^{t-a}-1
    $$
    are i.i.d.
Combining this with \eqref{eq:blk}, we conclude that the random vectors
    $$
    (\bU_{[h\ell^a+1:h\ell^a+\ell^a]}^{(a)},\bY_{[h\ell^a+1:h\ell^a+\ell^a]}), \quad
    h=0,1,\dots,\ell^{t-a}-1
    $$
    are also i.i.d.
This proves claim 1.

Next we prove claim 2 by induction.
The case of $a=1$ is trivial. Now we assume that \eqref{eq:cl2} holds for $a$ and prove it for $a+1$.
In light of claim 1, we only need to prove \eqref{eq:cl2} for the special case of $h=0$ because the distributions for different values of $h$ are identical, i.e. we only need to prove that
\begin{equation} \label{eq:aplus1}
\P( U_i^{(a+1)}
\to (\bU_{[1:i-1]}^{(a+1)},\bY_{[1:\ell^{a+1}]}) )
\equiv
W_{i_1,\dots,i_{a+1}}(K_1^{(0)},K_{i_1}^{(1)},\dots,K_{i_1,\dots,i_a}^{(a)}) 
\quad\quad \forall i\in[\ell^{a+1}] .
\end{equation}

For a given $i\in[\ell^{a+1}]$, we write its $(a+1)$-digit expansion as $\tau_{a+1}(i)=(i_1,i_2,\dots,i_{a+1})$.
By \eqref{eq:dfrecu}, we know that
$\bV_{[1:N]}^{(a)} =  \bU_{[1:N]}^{(a+1)} D^{(a)}$.
By \eqref{eq:defD}, the matrix $D^{(a)}$ is a block diagonal matrix with $\ell^{t-1}$ blocks on the diagonal, where each block has size $\ell \times \ell$. (Note that these $\ell^{t-1}$ blocks are not all the same unless $a=0$.)
Therefore, for every $h=0,1,\dots,\ell^{t-1}-1$, there is a bijection between the two vectors $\bV_{[h\ell+1:h\ell+\ell]}^{(a)}$ and $\bU_{[h\ell+1:h\ell+\ell]}^{(a+1)}$. Consequently, there is a bijection between the two vectors
$\bU_{[1:i-i_{a+1}]}^{(a+1)}$ and $\bV_{[1:i-i_{a+1}]}^{(a)}$, so
we have
\begin{equation}  \label{eq:df41}
\P( U_i^{(a+1)}
\to (\bU_{[1:i-1]}^{(a+1)},\bY_{[1:\ell^{a+1}]}) )
\equiv
\P( U_i^{(a+1)}
\to (\bU_{[i-i_{a+1}+1:i-1]}^{(a+1)},\bV_{[1:i-i_{a+1}]}^{(a)},\bY_{[1:\ell^{a+1}]}) ) .
\end{equation}
By \eqref{eq:defD}, we also have that
\begin{equation}  \label{eq:df17}
\bV_{[i-i_{a+1}+1:i-i_{a+1}+\ell]}^{(a)} =  \bU_{[i-i_{a+1}+1:i-i_{a+1}+\ell]}^{(a+1)} K_{i_1,i_2,\dots,i_a}^{(a)}   .
\end{equation}
Let $\hat{i}:=(i-i_{a+1})/\ell$, so $\tau_a(\hat{i})=(i_1,i_2,\dots,i_a)$.
According to the induction hypothesis,
$$
\P( U_{\hat{i}}^{(a)}
\to (\bU_{[1:\hat{i}-1]}^{(a)},\bY_{[1:\ell^a]}) )
\equiv
W_{i_1,\dots,i_a}(K_1^{(0)},K_{i_1}^{(1)},\dots,K_{i_1,\dots,i_{a-1}}^{(a-1)})  .
$$
Combining this with the relation $\bU_{[1:N]}^{(a)} = \bV_{[1:N]}^{(a)} Q^{(a)}$ and \eqref{eq:df41}--\eqref{eq:df17}, we can prove \eqref{eq:aplus1} with the ideas illustrated in Fig.~\ref{fig:top3}--\ref{fig:bot3}. 
This completes the proof of claim 2 as well as Proposition~\ref{prop:eqv}.

\section*{Acknowledgment}

The authors are grateful to Hamed Hassani for useful discussions and sharing his insights on random coding theorems during the initial stages of this work.  They also thank the anonymous
reviewers for their careful reading and valuable suggestion and corrections to the final paper.

\bibliographystyle{alpha}
\bibliography{arikan-meets-shannon}

\end{document}